%% file: Parallel_Stream_Processing_Against_Workload_Skewness_and_Variance.tex
\documentclass[conference]{IEEEtran}

\usepackage{graphicx}
\usepackage{balance}  % for  \balance command ON LAST PAGE  (only there!)
\usepackage{CJK}
\usepackage{algorithm}
\usepackage{algorithmicx}
\usepackage[noend]{algpseudocode}
\usepackage{amsmath}
\usepackage{url}
\usepackage{bm}
\usepackage{color}
\usepackage{cite}
\usepackage{multirow}
\usepackage{amsfonts}
\usepackage{times}
\newtheorem{restate}{Theorem}

\usepackage{graphicx}
\usepackage{subfigure}

\floatname{algorithm}{algorithm}

\algdef{SE}[DOWHILE]{Do}{DoWhile}{\algorithmicdo}[1]{\algorithmicwhile\ #1}%
\newcommand{\eat}[1]{}

\newtheorem{definition}{Definition}
\newtheorem{lemma}{Lemma}

\newtheorem{theorem}{Theorem}
\usepackage{epsfig}

\ifCLASSINFOpdf
\else
\fi
\begin{document}
%
% paper title
% can use linebreaks \\ within to get better formatting as desired
\title{Parallel Stream Processing Against Workload Skewness and Variance}

% author names and affiliations
% use a multiple column layout for up to three different
% affiliations
\author
{\IEEEauthorblockN{Junhua Fang$^{\dag}$ ~~  Rong Zhang$^{\dag}$ ~~ Tom Z.J.Fu$^{\ddag}$ ~~ Zhenjie Zhang$^{\ddag}$ ~~ Aoying Zhou$^{\dag}$ ~~	Junhua Zhu$^{\diamond}$}
	\\
	$^{\dag}$Institute  for Data Science and Engineering, Software Engineering Institute,\\East China Normal University, Shanghai, China\\
	\{jh.fang,rzhang,ayzhou\}@sei.ecnu.edu.cn\\        
	$^{\ddag}$Advanced Digital Sciences Center, Illinois at Singapore Pte. Ltd.\\
	\{tom.fu,zhenjie\}@adsc.com.sg\\
	$^{\diamond}$Huawei Technologies Co. Ltd. \{junhua.zhu\}@outlook.com\\
}

% use for special paper notices
%\IEEEspecialpapernotice{(Invited Paper)}

% make the title area
\maketitle

\begin{abstract}
Key-based workload partitioning is a common strategy used in parallel stream processing engines, enabling effective key-value tuple distribution over worker threads in a logical operator. While randomized hashing on the keys is capable of balancing the workload for key-based partitioning when the keys generally follow a static distribution, it is likely to generate poor balancing performance when workload variance occurs on the incoming data stream. This paper presents a new key-based workload partitioning framework, with practical  algorithms to support dynamic workload assignment for stateful operators. The framework combines hash-based and explicit key-based routing strategies for workload distribution, which specifies the destination worker threads for a handful of keys and assigns the other keys with the hashing function. When short-term distribution fluctuations occur to the incoming data stream, the system adaptively updates the routing table containing the chosen keys, in order to rebalance the workload with minimal migration overhead within the stateful operator. We formulate the rebalance operation as an optimization problem, with multiple objectives on minimizing state migration costs, controlling the size of the routing table and breaking workload imbalance among worker threads. Despite of the NP-hardness nature behind the optimization formulation, we carefully investigate and justify the heuristics behind key (re)routing and state migration, to facilitate fast response to workload variance with ignorable cost to the normal processing in the distributed system. Empirical studies on synthetic data and real-world stream applications validate the usefulness of our proposals and prove the huge advantage of our approaches over state-of-the-art solutions in the literature.
\end{abstract}
% IEEEtran.cls defaults to using nonbold math in the Abstract.
% This preserves the distinction between vectors and scalars. However,
% if the conference you are submitting to favors bold math in the abstract,
% then you can use LaTeX's standard command \boldmath at the very start
% of the abstract to achieve this. Many IEEE journals/conferences frown on
% math in the abstract anyway.

% no keywords

% For peer review papers, you can put extra information on the cover
% page as needed:
% \ifCLASSOPTIONpeerreview
% \begin{center} \bfseries EDICS Category: 3-BBND \end{center}
% \fi
%
% For peerreview papers, this IEEEtran command inserts a page break and
% creates the second title. It will be ignored for other modes.
\IEEEpeerreviewmaketitle

\input{introduction}
\input{overview}
\input{solution}

\input{optimization}
\input{evaluation}

\input{relatedwork}
\input{conclusion}	

\bibliographystyle{abbrv}
\bibliography{sigproc}

\appendix
	
%\section*{Acknowledgment}
\input{appendix-upboundproof}
%\input{appendix-lowboundexplanation}
\input{appendix-routingtablechange}
\input{appendix-addparametertest}

\end{document}

%% file: introduction.tex
\section{Introduction}\label{sec:introduction}

%\small {\it The analysis of variance is not a mathematical theorem, but rather a convenient method of arranging the arithmetic.}
%\begin{flushright} Sir Ronald Fisher, 1934
%\end{flushright}
%\normalsize

% workload variance in distributed stream processing engine

Workload skewness and variance are common phenomena in distributed stream processing engines. When massive stream data flood into a distributed system for processing and analyzing, even slight distribution change on the incoming data stream may significantly affect the system performance. Existing optimization techniques for stream processing engines are designed to exploit the distributed processor, memory and bandwidth resources based on the computation workload, but potentially generate suboptimal performance when the evolving workload deviates from expectation. Unfortunately, workload evolution is constantly happening in real application scenarios (e.g., surveillance video analysis \cite{chen2005computer} and online advertising monitoring \cite{kutare2010monalytics}). It raises new challenges to distributed system on solutions to handle the dynamics of data stream while maintaining high resource utilization rate at any time.

%While long-term workload variance, e.g. the growth of tweets during office hours, is usually handled by resource rescheduling, by requesting and reassigning computation resource, it is too expensive and impractical to apply such heavyweight operations for short-term distribution fluctuations, which may last only for a very short period of time. To improve the robustness of the distributed system, it is more desirable to deploy lightweight protocols in the system, in order to quickly respond to short-term workload variance with minimal overhead to the overall performance of the distributed system.

% why load balancing within operator
In distributed stream processing system, abstract operators are connected in form of a directed graph to support complex processing logics over the data stream. Traditional load balancing approaches in distributed stream processing engines attempt to balance the workload of the system, by evenly assigning a variety of heterogenous tasks to distributed nodes \cite{2005designBorealis,xing2005dynamic,xing2006providing,ahmad2004network,khandekar2009cola}. Such strategies may not perform as expected in distributed stream processing systems, because of the lack of balance on the homogeneous tasks within the same abstract operator. In Fig.~\ref{fig:intro:motivation}, we present an example to illustrate the potential problem with such strategies. In the example, there are three logic operators in the pipeline, denoted by rectangles. There are three concrete task instances running in \emph{operator 2}, denoted by circles. The number of incoming tuples to the first task instance is two times of that to the second and third task instances, due to the distribution skewness on the tuples. Even if the system allocates the tasks in a perfect way to balance the workload when allocating task instances to computation nodes, the processing efficiency may not be optimal. Because of the higher processing latency in the first task instance of \emph{operator 2}, \emph{operator 1} is forced to slow down its processing speed under backpushing effect, and \emph{operator 3} may be suspended to wait for the complete intermediate results from \emph{operator 2}. This example shows that load balancing between task instances within individual logical operators is more crucial to distributed stream processing engines, to improve the system stability and guarantee the processing performance.

\begin{figure}[h]
\centering
\epsfig{file=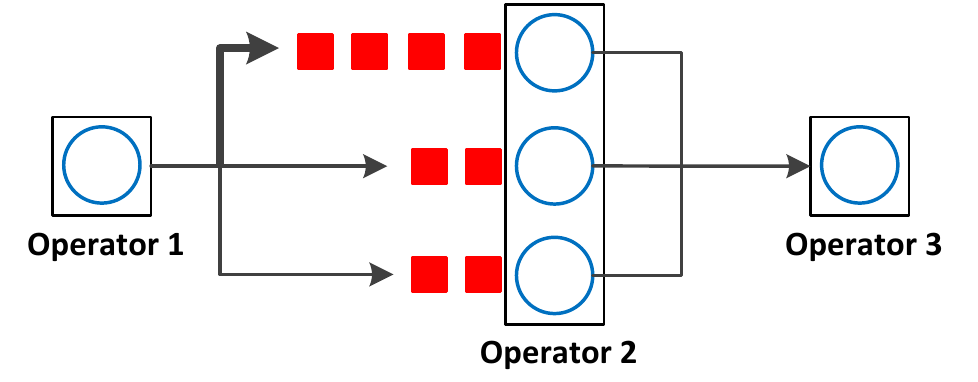, width=1\columnwidth}
\caption{The potential problem of workload imbalance within operators in real distributed stream processing engine.}
\label{fig:intro:motivation}
\end{figure}

% discussion on long-term and short-term

There are two types of workload variance in distributed stream processing engines, namely \emph{long-term} workload shift and \emph{short-term} workload fluctuation. Long-term workload shifts usually involve distribution changes on incoming tuples driven by the change in physical world (e.g., regular burst of tweets after lunch time), while workload fluctuations are usually short-term and random in nature. Long-term workload shifts can only be solved by applying heavyweight resource scheduling, e.g., \cite{fu2015drs}, which reallocates the computation resource based on the necessity. Computation infrastructure of the distributed system may request more (less resp.) resource, by adding (returning resp.) virtual machines, or completely reshuffling the resource between logical operators according to the demands of operators. Such operations on the infrastructure level are inappropriate for short-term workload fluctuations, usually too expensive and render suboptimal performance when the fluctuation is over. It is thus more desirable to adopt lightweight protocols within system, to smoothly redistribute the workload between task instances, minimize the impact on the normal processing, and achieve the objective of load balancing within every logical operator. This paper focuses on such a dynamic workload assignment mechanism for individual logical operators in a complex data stream processing logic, especially against short-term workload fluctuations. Note that existing solutions to long-term workload shifts are mostly orthogonal to the mechanisms for short-term workload fluctuations, both of which can be invoked by the system optionally based on the workload characteristics.

The methods which focus on designing load balance algorithms for stream system can be divided into two categories: \textbf{split-key-based} and \textbf{non-split-key-based}.
Split-key-based method splits data of the same key and distributes them to the parallel processing instances. The representative algorithm is implemented in PKG \cite{nasir2015power}. 
However, this approach distories the semantics of key-based operations,  resulting in additional processing on some operators.
%For example, both the model of matrix and bipartite graph arrange the parallel processing instances into fixed shape to ease the load imbalance, but the fixed structure also limits their suitability for join operations.
In PKG, for example, it splits data of the same key into multiple subsets and  distributes them selectively to different instances to avoid load imbalance. However, this approach leads to some negative impact on system performance.
As shown in Fig.~\ref{fig:exp:split-key-agg}, aggregations must contain some partial result operators and an additional merge operator if implemented on a split-key-based architecture. In Fig.~\ref{fig:split-key-join}, it shows the join operation on this architecture.  The data of key $ k_{1} $ from $ R $ stream are split into two parts and assigned to instances $ d_{1} $ and $ d_{2} $. To ensure the correctness of join operation, as the routing in work \cite{lin2015scalable}, data of $k_1$ from stream $ S $ should be broadcasted to both instances $d_1$ and $d_2$.

Non-split-key-based methods can guarantee the semantics of Key-based operation, but it weakens system balance capability severely. The representative algorithm is implemented in $ Readj $ \cite{gedik2014partitioning}.
It treats the data of a key as a complete granularity and does not split it during load adjustment.
However, when the number of keys is huge, $ Readj $ is of high time complexity for it considers all possible swaps by pairing tasks and keys to find the best key movement to alleviate the workload imbalance.
Furthermore, $Readj$ just considers to adjust the big load keys to make workload balance, which will incur expensive migration cost for stateful operator, such as join.

Then it is a urgent to design a new method which can not only guarantee the key-based semantics, but also balance system workload quick and efficient.

\begin{figure}[htp]
	\centering
	\begin{tabular}[t]{c}
		\hspace{-23pt}
		\subfigure[Aggregation Operation]{
			\includegraphics[height = 3.9cm, width = 5.2cm]{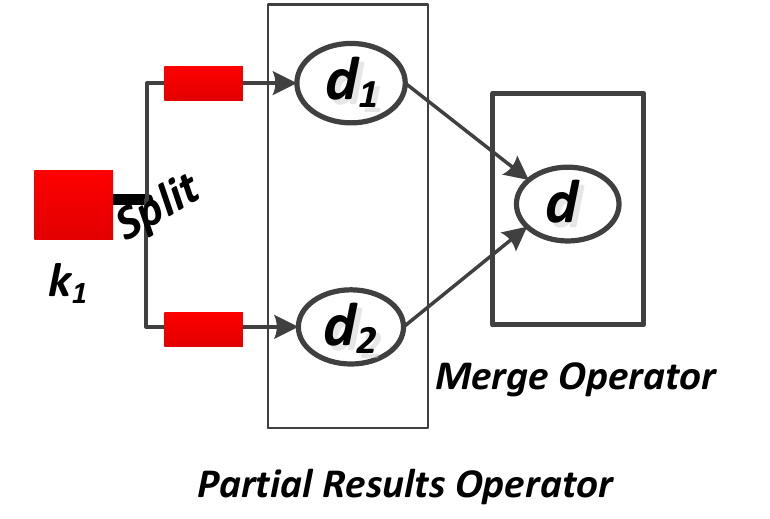}
			\label{fig:exp:split-key-agg}
		}
		\hspace{-36pt}
		\subfigure[Join Operation]{
			\includegraphics[height = 3.9cm, width = 5.2cm]{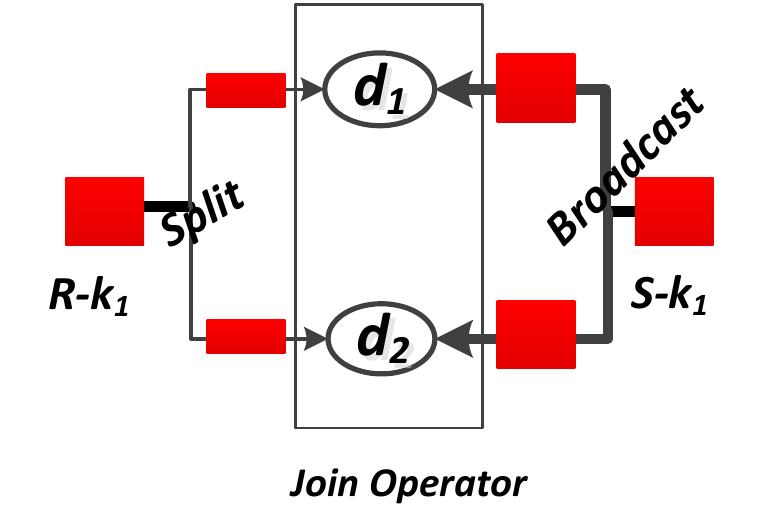}
			%\caption{Stock Data}
			\label{fig:split-key-join}
		}
	\end{tabular}
	%\vspace{-5pt}
	\caption{Sketch for operations that using split-key-based to make workload balance.}
	\label{fig:partitionoperators}
	% \vspace{-10pt}
\end{figure}

Our proposal in this paper is based on a mixed strategy of key-based workload partitioning, which explicitly specifies the destination worker threads for a handful of keys and assigns all other keys with the randomized hashing function. This scheme achieves high \emph{flexibility} by easily redirecting the keys to new worker threads with simple editing on the routing table. It is also highly \emph{efficient} when the system sets the maximal size of the routing table, thus controlling the memory overhead and calculation cost with the routing table. Workload redistribution with the scheme is \emph{scalable} and \emph{effective}, by allowing the system to respond promptly to the short-term workload fluctuation even when there are a large number of keys present in the incoming data stream. To fully unleash the power of the scheme, it is important to design a monitoring and controlling mechanism on top of system, making optimal decisions on routing table update to achieve intra-operator workload balancing. Recent research work in \cite{gedik2014partitioning}, although employing similar workload distribution strategy, only considers migration of \emph{hot} keys with high frequencies, limiting the optimizations within much smaller configuration space. We break the limit in this paper with a new solution for distributed systems to explore possible optimizations with all candidate keys for the routing table, thus maximizing the resource utilization with ignorable additional cost. Specifically, the technical contributions of this paper include:

\begin{itemize}
\item[-] We design a general strategy to generate the partition function for data redistribution under different stream dynamic changes at runtime, which achieves scalability, effectiveness and efficiency by a single shot.
\item[-] We propose a lightweight computation model to support rapid migration plan generation, which incurs minimal data transmission overhead and processing latency.
\item[-] We present a detailed theoretical analysis for proposed migration algorithms, and prove its usability and correctness.
\item[-] We implement our algorithms on Storm and give extensive experimental evaluations to our proposed techniques by comparing with existing work using abundant datasets. We explain the results in detail.
\end{itemize}

The remainder of this paper is organized as follows. Section \ref{sec:overview} introduces the overview and preliminaries of our problem. Section \ref{sec:solution} presents our balancing algorithms to support our mixed workload distribution scheme. Section \ref{sec:optimization} proposes the optimization techniques used in the implementation of our proposal. Section \ref{sec:evaluations} presents empirical evaluations of our proposal. Section \ref{sec:relatedwork} reviews a wide spectrum of related studies on stream processing, workload balancing and distributed systems. Section \ref{sec:conclusion} finally concludes the paper and addresses future research directions.
%\vspace{-1em}

%% file: overview.tex
%%overview
\section{Preliminaries}\label{sec:overview}

%\subsection{Model of Distributed Stream Processing}

A distributed stream processing engine (DSPE) deploys abstract stream processing logics over interconnected computation nodes for continuous stream processing. The abstract stream processing logic is usually described by a directed graphical model  (e.g., Storm \cite{toshniwal2014storm}, Heron \cite{kulkarni2015twitter} and Spark Streaming \cite{zaharia2013discretized}), with a vertex in the graph denoting a computation operator and an edge denoting a stream from one operator to another. Each data stream consists of key-value pairs, known as \emph{tuples}, transmitted by network connection between computation nodes. The computation logic with an operator is a mapping function with an input tuple from upstream operator to a group of output tuples for downstream operators.
%{\color{red}Furthermore, for each operator in the topology there is a corresponding upstream operator that can be responsible for its workload balance, and therefore, the balanced mothed between upstream and downstream operator can be used in the whole topology.}

To maximize the throughput of stream processing and improve the utilization rate of the computation resource, the workload of a logical operator is commonly partitioned and concurrently processed by a number of threads, known as \emph{tasks}. The upstream operator is aware of the concrete tasks and sends the output tuples to the tasks based on a global partitioning strategy. All concrete tasks within an operator process the incoming tuples independently. Key-based workload partitioning is now commonly adopted in distributed stream processing engines, such that tuples with the same key are guaranteed to be received by the same concrete task for processing. An operator is called \emph{stateful operator}, if there is a memory space used to keep intermediate results, called \emph{states}, of the keys based on the latest tuples. Basically, a \emph{state} is associated with an active key in the corresponding task in a stateful operator, which is used to maintain necessary information for computation. The state, for example, can be used to record the counts of the words or recent tuples in the sliding window. Because of the tight binding between key and state, when a key is reassigned to another task instance, its state must be migrated as well, in order to ensure the correctness of computation outcomes.

The workload partitioning among concrete tasks is the model as a mapping from key domain to running tasks in the successor operator. A straightforward solution to workload partitioning is the employment of hashing function (e.g., by consistent hashing), which chooses a task for a specific key in a random manner. The computational cost of task selection for a tuple is thus constant. As discussed in previous section, despite of the huge advantages of hashing on memory consumption and computation cost, such scheme may not handle well with workload variance and key skewness. Another option of workload distribution is to explicitly assign the tuples based on a carefully optimized routing table, which specifies the destination of the tuples by a map structure on the keys. Although such an approach is more flexible on dynamic workload repartitioning, the operational cost on both memory and computation is too high to afford in practice.

\begin{figure}
	\centering
	\includegraphics[height = 3.6cm, width = 8.2cm]{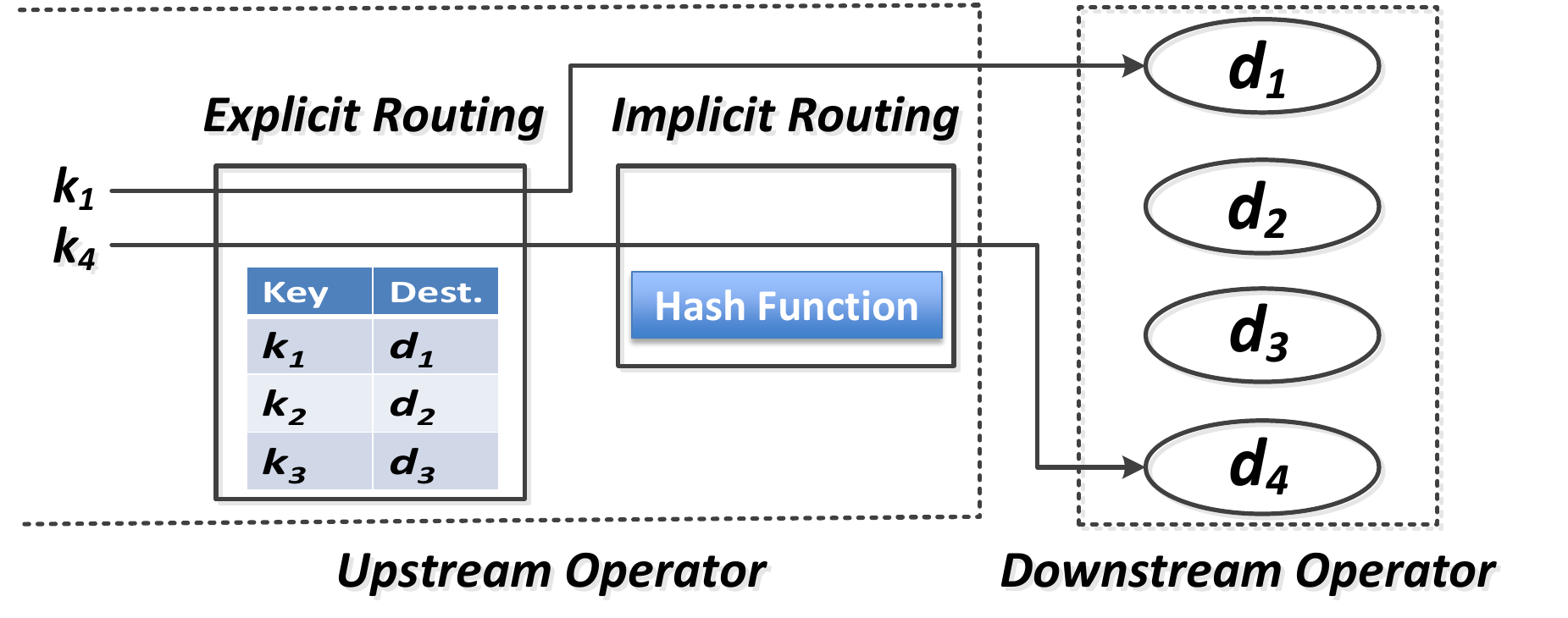}
	\centering \caption{The scheme of mixed routing with a small routing table and a hash function.}	\label{fig:overview:routing}
	% \vspace{-10pt}
	%\label{fig:Parameters}
\end{figure}
% mixed routing strategy

In this paper, we develop a new workload partitioning framework based on a mixed routing strategy, expecting to balance the hash-based randomized strategy and key-based routing strategy. In Fig.~\ref{fig:overview:routing}, we present an example of the strategy with one data stream between two operators. A routing table is maintained in the system, but contains routing rules for a handful of keys only. When a new output tuple is generated for the downstream operator, the upstream operator first checks if the key exists in the routing table. If a valid entry is found in the table, the tuple is transmitted to the target concrete task instance specified by the entry, otherwise a hashing function is applied on the key to deterministically generate the target task id for the tuple. By appropriately controlling the routing table with a maximal size constraint, both the memory and computation cost of the scheme are acceptable, while the flexibility and effectiveness are achieved by updating the routing table in response to the evolving distribution of the keys.

As illustrated in the previous section, workload balancing between tasks from the same operator is crucial and it is the major problem we aim to tackle with. With the mixed routing strategy, we can solve the problem by focusing on the construction and update of the routing table with the constrained size, without considering the global structure of processing topology and workload. Therefore, our discussion in the following sections is focuses on one single operator and its routing table. Note that our approach is obviously applicable to complex stream processing logics, as evaluated in the experimental section.

\subsection{Data and Workload Models}
\vspace{3pt}
In our model, the time domain is discretized into intervals with integer timestamps, i.e., $(T_1, T_2, \ldots, T_i,\ldots)$. At the $i$-th interval, given a pair of upstream operator $U$ and downstream operator $D$, we use $\mathcal{U}$ and $\mathcal{D}$ to
denote the set of task instances within upstream operator $U$ and downstream operator $D$, respectively. We also use $N_U = |\mathcal{U}|$ and $N_D = |\mathcal{D}|$ to denote the numbers of task instances in $U$ and $D$, respectively. A tuple is tuple $\tau=(k, v)$, in which $k$ is the key of the tuple from key domain $\mathcal{K}$ and $v$ is the value carried by the tuple.
We assume $N_U$ and $N_D$ are predefined without immediate change. The discussion on dynamic resource rescheduling, i.e., changing $N_U$ and $N_D$, is out of the scope of this paper, since it involves orthogonal optimization techniques on global resource scheduling (e.g., \cite{fu2015drs}).

A key-based workload partitioning mechanism works as a mapping $F:\mathcal{K}\rightarrow \mathcal{D}$, such that a tuple $(k,v)$ is sent to task instance $F(k)$ by evaluating the key $k$ with the function $F$. Without loss of generality, we assume a universal hashing function $h:\mathcal{K}\rightarrow \mathcal{D}$ is available to the system for general key assignment. A typical implementation of hash function is consistent hashing, which is believed to be a suitable option for balancing keys among instances.  However, it does not consider so much about key granularities which is the number of data of the same key. Hence, the balanced state of parallel processing has to be obtained by moving keys among instances even with the consistent hash and we use routing table to record the mappings from keys to processing destinations which are not the basic hash destinations. A routing table $A$ of size $N_A$ contains a group of pairs from $\mathcal{K}\times \mathcal{D}$, specifying the destination task instances for keys existing in $A$. The mixed routing strategy shown in Fig. \ref{fig:overview:routing} is thus modelled by the following equation:

\begin{equation}
	F(k) =
\begin{cases}
		d,& \text{if}~\exists\; (k, d)\in A,\\
		h(k),&\text{otherwise}.
\end{cases}
\label{EQ:Fl_def}
\end{equation}

Therefore, workload redistribution is enabled by editing the routing table $A$ with an assignment function $F(\cdot)$. In the following, we provide formal analysis on the general properties of the assignment function $F(\cdot)$.

\vspace{6pt}
\noindent\textbf{Computation Cost: }We use $g_i(k)$ to denote the frequency of tuples with key $k$ in time interval $T_i$, and define the computation cost $c_i(k)$ by the amount of CPU resource necessary for all these tuples with key $k$ in time interval $T_i$. Generally speaking, $c_i(k)$ increases with the growth of $g_i(k)$. Unless specified, we do not make any assumption on the correlation between $g_i(k)$ and $c_i(k)$, both of which are measured in the distributed system and recorded as statistics, in order to support decision making on the update of $F(\cdot)$. The total workload with a task instance $d$ in downstream operator $D$ within time interval $T_i$ is calculated by
$L_i(d, F) = \sum_{\{k | F(k) = d, k\in\mathcal{K}\}} c_i(k)$.

\vspace{6pt}
\noindent\textbf{Load Balance: }Load balance among task instances of the downstream  operator $D$ is the essential target of our proposal in this paper. Specifically, we define the balance indicator $\theta_i(d, F)$ for task instance $d$ under assignment function $F$ during time interval $T_i$ as $\theta_i(d, F) = \frac{\left|L_i(d, F) - \bar{L}_i\right|}{\bar{L}_i}$, where $\bar{L}_i = \frac{1}{N_D}\sum_{d\in\mathcal{D}} L_i(d, F)$ is the average load of all task instances in $\mathcal{D}$. As it is unlikely to achieve absolute load balancing with $\theta_i(d,F)=0$ for every task instance $d$, an upper bound $\theta_{\max}$ is usually specified by the system administrator, such that the workload of task instance $d$ is approximately balanced if $\theta_i(d,F)\leq \theta_{\max}$.
\eat{{\color{red}$\theta_{\max}$ is a parameter that configured by user according their actual situation, such as a load stable scene can set smaller $\theta_{\max}$ for its low frequency balance adjustment.}}

\vspace{6pt}
\noindent\textbf{Memory Cost: }For stateful operators, the system is supposed to maintain historical information, e.g., statistics with the keys, for processing and analysing on newly arriving tuples. We assume that each operator maintains states independently on individual time interval $T_i$ and only the last $w$ time intervals are needed by any task instance. It means that the task instance erases the state from time interval $T_{i-w}$ after finishing the computation on all tuples in time interval $T_i$. This model is general enough to cover almost all continuous stream processing and analytical jobs (e.g., stream data mining over sliding window). The memory consumption for tuples with key $k$ in $T_i$ is thus measured as $s_i(k)$, and the total memory consumption for key $k$ is the summation over last $w$ intervals on the time domain, as $S_i(k, w) = \sum_{j=i-w+1}^{i} s_j(k)$.

\vspace{6pt}
\noindent\textbf{Migration Cost: }Upon the revision on assignment function $F$, certain key $k$ may be moved from one task instance to another. The states associated with key $k$ must be moved accordingly to ensure the correctness of processing on following tuples with key $k$. The migration cost is thus modelled as the total size of states under migration. By replacing function $F$ with another function $F'$ at time interval $T_i$, we use $\Delta(F,F')=\{k~|~F(k)\neq F'(k),k\in\mathcal{K}\}$. The key state migration includes all the historical states within the given window $w$. Thus, the total migration cost, denoted by $M_i\left(w, F, F'\right)$, can be defined as:

\begin{equation}
M_i\left(w, F, F'\right) = \sum_{k\in\Delta(F,F')} S_i(k, w).\label{EQ:migCost_def}
\end{equation}

\begin{table}[pt]
%\vspace{-12pt}
\small
\centering \caption{Table of Notations}
\label{tab:notation}
\begin{tabular}{|c|p{6cm}|}
%\small
\hline
%\diagbox{parm}{desc} & Range & Default & Description \\
%\diagbox pc& Range & Default & Description \\
Notations & Description  \\
\hline $T_i$ & The $i$-th time interval \\
\hline $U$ & Upstream operator \\
\hline $D$ & Downstream operator \\
\hline $N_U, N_D$ & Numbers of task instances in $U$ and $D$\\
\hline $\mathcal{D}$ & Instances set of downstream operator \\
\hline $(k,v)$ & Key-value pair on the data stream\\
\hline $k$ & Key of the tuple\\
\hline $v$ & Kalue of the tuple\\
\hline $d$ & Task instance in downstream operator \\
\hline $A$ & The routing table available to $U$ \\
\hline $N_A$ & Number of entries in $A$ \\
\hline $c_i(k)$& Computation cost of all tuples with key $k$ in $T_i$\\
\hline $g_i(k)$ & Frequency of key $k$ in time interval $T_i$ \\
\hline $L_i(d,F)$ & Total workload of task instance $d$ under assignment function $F$ \\
\hline $\bar{L}_i$ & Average load of all instances in $U$ in time interval $T_i$ \\
\hline $\theta_i(d,F)$ & Load balance factor of task instance $d$ \\
\hline $\theta_{\max}$ & Upper bound of imbalance tolerance \\
\hline $S_i(k,w)$ & Memory cost of key $k$ with $w$ time intervals at $T_i$ \\
\hline $\Delta(F,F')$ & Keys with different destination under $F$ and $F'$ \\
\hline $M_i(w,F,F')$ & Total migration cost by replacing $F$ with $F'$ at time interval $T_i$ \\
\hline
\end{tabular}
\vspace{-5pt}
\end{table}

All notations used in the rest of the paper are summarized in Tab.~ \ref{tab:notation}.

\subsection{Problem Formulation}\label{sec:pro_formulation}

Based on the model of data and workload, we now define our dynamic workload distribution problem, with the objectives on (i) load balance among all the downstream instances; (ii) controllable size on the routing table; and (iii) minimization on state migration cost. These goals are achieved by controlling the routing table in the assignment function, under appropriate constraints for performance guarantee. Specifically, to construct a new assignment function $F'$ as a replacement for $F$ in time interval $T_i$, we formulate it as an optimization problem, as below:
\begin{eqnarray}
&\min\limits_{F'(\cdot)} & M_i(w, F, F')\nonumber\\
&\text{s.t.} &\theta(d, F') \leq \theta_{\max}, \forall d \in \mathcal{D},\label{EQ:definition}\\
& & N_A \leq A_{max}\nonumber,
\end{eqnarray}

in which $F$ is the old assignment function and $F'$ is the variable for optimization. The target of the program above is to minimize migration cost, while meeting the constraints on load balance factor and routing table size with user-specified balance bounds $\theta_{\max}$ and $A_{\max}$ which is the maximum constrained size of $A$.

%From equations (\ref{EQ:Load_def})-(\ref{EQ:migCost_def}), both $M_i(w)$ and $\theta_i(d)$ are affected by $F_i(k)$, and so is the routing table size $A_i$.

It is worthwhile to emphasize that the new assignment function is constructed at the beginning of a new time interval $T_i$. The optimization is thus purely based on the statistical information from previous time interval $T_{i-1}$. The metrics defined in previous subsection are estimated with frequencies $\{g_{i-1}(k)\}$ over the keys, the computation costs $\{c_{i-1}(k)\}$ and the memory consumption $S_{i-1}(k,w)$.

The problem of initializing the keys in $\mathcal{K}$, with the task instance set $\mathcal{D}$ and load balance constraint $\theta_{max}$, is a combinatorial NP-hard problem, as it can be reduced to Bin-packing problem\cite{karmarkar1982efficient}. Even Worse, our optimization problem also puts constraints on the maximal table size and migration cost.
Specifically, even if the parallel instances achieve  the balance state,but the routing table size exceeds the predefined space limits as shown in Equ.~\ref{EQ:definition}, this balanced state should not conform to the requirement.
Therefore, in the following section, we discuss a number of heuristics with careful analysis on their usefulness.

%% file: solution.tex
%%Solution section
%\section{Solution Description}\label{sec:solution}
\vspace{-3pt}
\section{Algorithms}\label{sec:solution}
%\vspace{-0.5em}

In this section, we introduce algorithms to solve the optimization problem raised in previous section, targeting to construct a new assignment function $F'$ by updating the routing table $A$. We simply assume that all necessary statistics are available in the system for algorithms to use. The implementation details, including measurement collection, are discussed in the following section.

Since the optimization problem is clearly NP-hard, there is no polynomial algorithm to find global optimum, unless \textbf{P=NP}.
In the rest of the section, we firstly describe a general workflow for a variety of heuristics, such that all algorithms based on these heuristics follow the same operation pattern. We then discuss a number of heuristics with objectives on routing table minimization and migration minimization. A mixed algorithm is introduced to combine the two heuristics in order to accomplish the constraints in the optimization formulation with a single shot.
%
%Analysis on the algorithms is finally presented to illustrate the theoretical properties of the approaches.
%In the rest of the section, we first present an approximate algorithm based on the least-load first principle. We will also present two heuristics with objectives on routing table minimization and migration minimization respectively. Finally, we will propose a mixed algorithm to combine the two heuristics in order to accomplish the constraints in the optimization formulation with a single shot. Analysis on the algorithms are presented to illustrate the theoretical properties of the approaches.

%\textbf{ZZJ: describe the general workflow.}
%\textbf{General workflow:}

%All of the algorithms in this section follow the general workflow presented in Figure \ref{}.

Generally speaking, the system follows the steps below when constructing a new assignment function $F'$. 

\noindent\textbf{Phase I (Cleaning):} It attempts to \emph{clean} the routing table $A$ by removing certain entries in the table. This is equivalent to moving the keys in the entries back to the original task instance assignment, decided by the hash function. Different algorithms may adopt different cleaning strategies to shrink the existing routing table in $F$. Note that such a temporary removal does not physically migrate the corresponding keys, but just generates an intermediate result table for further processing. 

\noindent\textbf{Phase II (Preparing):} It  identifies candidate keys for migration from overloaded task instances, i.e., $\{ d | L(d) > L_{\max}\}$, where $L_{\max} = (1+\theta_{\max})\bar{L}$. Different selection criteria, such as keys with highest computation cost first, and largest computation cost per unit memory consumption first (concerning about migration cost), and etc, can be applied by the algorithm to select keys and disassociate their assignments from the corresponding task instances. These disassociated keys will be temporarily put into a candidate key set (denoted by $\mathcal{C}$) for processing in the third step of the workflow.
%
%If an instance $d_1$, for example, is responsible for two keys $k_1$ and $k_2$ with computation costs shown in Tab.~~\ref{table:runnning_example}. If the average workload of all instances is $\bar{L} = 6$ and $\theta_{\max} = 0.1$ (hence, we have $L_{\max} = (1+\theta_{\max})\bar{L} = 6.6$), the total workload of $d_1 (= 14)$ is clearly above the maximal tolerance. Therefore, $d_1$ is marked as an overloaded instance by the system. If employing highest workload first strategy, $k_1$ are added into the candidate set, since it has higher workloads than $k_2$.
%
%For example, if an instance $d_1$ is responsible for two keys $k_1$ and $k_2$ with computation costs 8 and 6 respectively, and the average workload of all instances is assumed to be $\bar{L} = 6$ and $\theta_{\max} = 0.1$ (which infers $L_{\max} = (1+\theta_{\max})\bar{L} = 6.6$), then the total workload of $d_1 (= 14)$ is clearly above the maximal tolerance. Hence, $d_1$ is marked as an overloaded instance. Given the key selection strategy to be the highest computation cost first, $k_1$ will be disassociated from $d_1$ and added into $\mathcal{C}$.
%
%\noindent\textbf{Phase III (Assignment)} reshuffle the keys in the candidate set by manipulating the routing table, in order to balance the workloads. Again, different assignment strategies can be used in this phase, which are expected to optimize with the formulation in Eq. (7)-(9).

\noindent\textbf{Phase III (Assigning):} It reshuffles the keys in the candidate set by manipulating the routing table, in order to balance the workloads. In particular, all algorithms proposed in this paper invoke the Least-Load Fit Decreasing (LLFD) subroutine, which will be described shortly, in this phase.
%
%\vspace{-0.5em}
%\begin{table}[tp]
%	\centering
%	\caption{A toy example of 10 keys with their computation costs and memory consumption and assignments on two instances, where $w = 1$.}\label{table:runnning_example}
%	\begin{tabular}{|c|c|c|c|c|c|c|c|c|c|c|}
%		\hline
%		$k$ & $k_1$ & $k_2$ & $k_3$ & $k_4$ & $k_5$ & $k_6$ & $k_7$ & $k_8$ & $k_9$ & $k_{10}$\\
%		\hline
%		$c(k)$ & 8 & 6 & 3 & 2 & 2 & 1 & 1 & 1 & 1 & 1\\
%		\hline
%		$S(k, w)$ & 8 & 6 & 3 & 2 & 2 & 1 & 1 & 1 & 1 & 1\\
%		\hline
%		$F(k)$ & $d_1$ & $d_1$ & $d_2$ & $d_2$ & $d_2$ & $d_2$ & $d_2$ & $d_2$ & $d_2$ & $d_2$\\
%		\hline
%\end{tabular}
%\end{table}\vspace{-0.5em}

%\section{Solution Description}\label{sec:solution}

\begin{algorithm}[H]
	% \small
	\caption{Least-Load Fit Decreasing Algorithm}
	\label{alg:LLDF}
	\begin{algorithmic}[1]
		\Require key candidate $\mathcal{C}$, task instances in $\mathcal{D}$, imbalance tolerance factor $\theta_{\max}$, key selection criteria $\psi$
		\Ensure $A'$
		\ForAll{$d$ in $\mathcal{D}$}
		\State Initialize estimation $\hat{L}(d)=L_{i-1}(d)$
		\EndFor
		\ForAll{$k$ in $\mathcal{C}$ in descending order of $c_{i-1}(k)$}
		\ForAll{$d$ in $\mathcal{D}$ in ascending order of $L_{i-1}(d)$}
		\If{Adjust$(k, d, \mathcal{C}, \theta_{\max}) = $ TRUE}
		\If {$h(k) \neq d$}
		\State Add entry $(k, d)$ to $A'$
		\EndIf
		\State Update $\hat{L}(d)$; remove $k$ from $\mathcal{C}$; break;
		\EndIf
		\EndFor
		\EndFor
		\State \Return $A'$
		\Function {Adjust}{$k, d, \mathcal{C}, \theta_{\max}$}			
		\State $L_{\max}\leftarrow (1+\theta_{\max})\bar{L}_{i-1}$
		\If {$L_{i-1}(d) + c_{i-1}(k) < L_{\max}$}
		\State
		\Return TRUE
		\ElsIf{$\exists\;\mathcal{E}$ selected by $\psi$ and satisfying (i)-(iii)}
		\ForAll{$k \in \mathcal{E}$}
		\State Disassociate $k$ from $d$
		\State Add $k$ to $\mathcal{C}$
		\EndFor
		\State
		\Return TRUE
		\Else
		\State
		\Return FALSE
		\EndIf
		\EndFunction
	\end{algorithmic}
\end{algorithm}

\subsection{Least-Load Fit Decreasing (LLFD)}\label{sec:solution:llfd}
%\vspace{-0.5em}
In this part of the section, we introduce \emph{Least-Load Fit Decreasing} (LLFD) subroutine, which will be applied by all the proposed algorithms in Phase III, based on the idea of prioritizing keys with larger workloads.
%The algorithm based on this strategy is also frequently called by the other strategies, as a subroutine.
The design of LLFD is motivated by the classic First Fit Decreasing (FFD) used in conventional bin packing algorithms. The pseudo codes of LLFD are listed in Algorithm~\ref{alg:LLDF}.
%\vspace{-0.5em}

%The pseudocodes of LLFD is listed in Algorithm~\ref{alg:LLDF}.

\begin{figure}
	\centering
	\hspace{-16pt}
	\includegraphics[height = 8.2cm, width = 8.5cm]{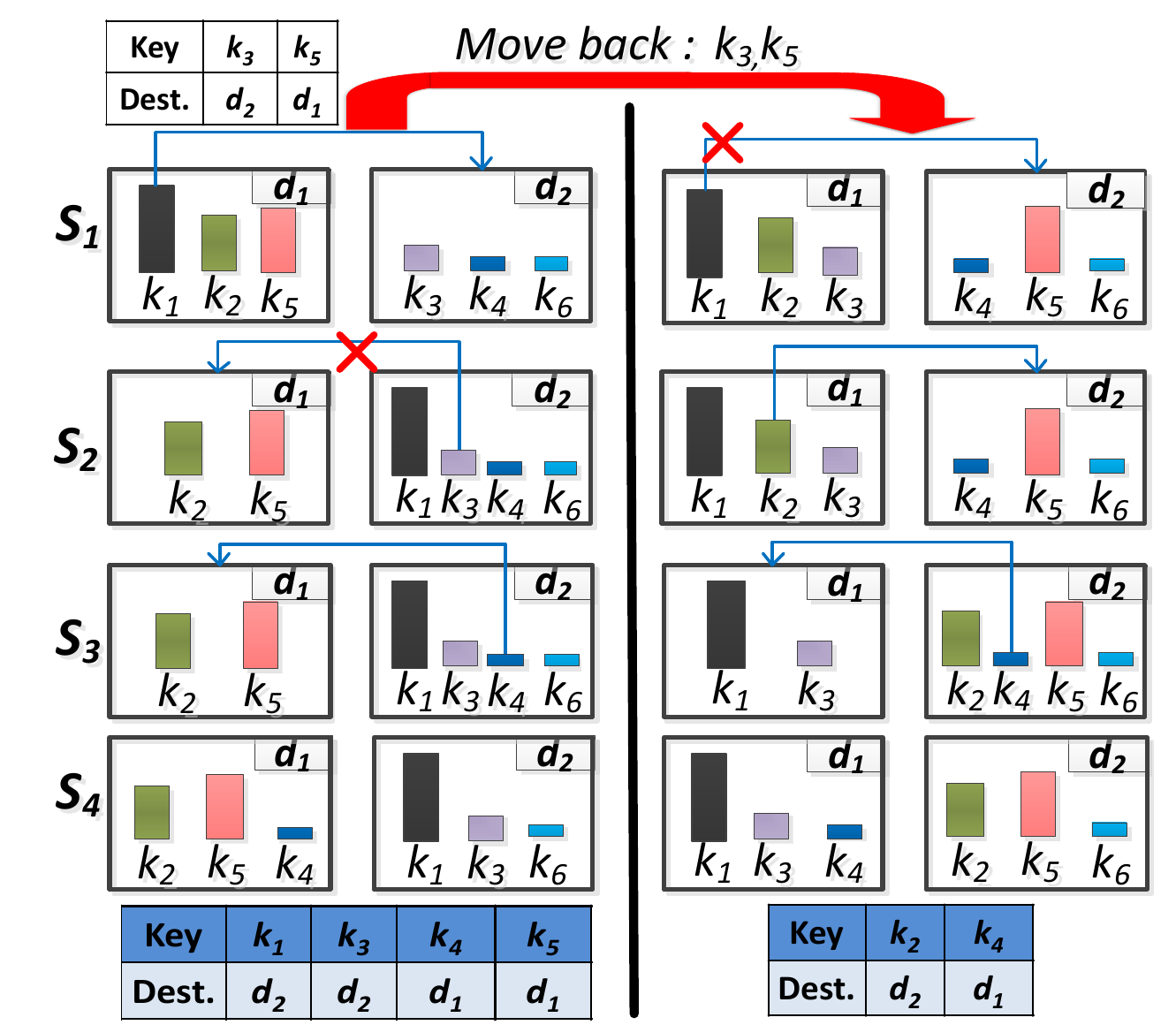}
	\centering \caption{Running examples for LLFD and MinTable with the constraints that the routing table size is not more than 2 and all instances should be absolute balance. The heights of the bars indicate the workloads of the corresponding keys. Each $S_j$ with $j=\{1,2,3,4\}$ is a running step in the algorithms. The original routing table is at top of the figure, and the result routing tables are listed at the bottom.}	\label{fig:solution:example}
	\vspace{-10pt}
\end{figure}

%LLFD first sorts keys in $\mathcal{C}$ in a descending order by their computation cost (Line 1).
%Next it iteratively assigns keys in this order to an instance which:
%
%(i) has the least total workload so far (Line 2); and
%
%(ii) pass the feasibility check (Line 3).

%The details about feasibility check will be discussed shortly.

%Finally, the algorithm returns $F'$ by combining $A'$ and $h(k)$.
%\vspace{-0.5em}

%\textbf{Feasibility check mechanism in LLFD: }

%Even if an instance originally has the least amount of total workloads, it happens that associated with a key of very high computation cost makes it become overloaded.
Generally speaking, LLFD sorts the keys in the candidate set in a non-increasing order of their computation costs and iteratively assigns the keys to task instances, such that (i) it generates the least total workload (Line 4); and (ii) it tries to adjust the key assignment, if the new destination task instance is overloaded after the migration (Line 5). If such key-to-instance pair is inconsistent with default mapping by hashing (Line 6), an entry $(k, d)$ is then added to the routing table $A$ (Line 7). After each iteration, LLFD updates the total workload of the corresponding instance $d$ and removes $k$ from the candidate set (Line 8). The iteration stops and returns the result routing table, when the candidate set turns empty (Line 9).

Basically, the algorithm moves the ``heaviest'' key to the task instance with minimal workload so far, which may generate another overloaded task instance (referred as ``re-overloading'' problem), if this key is associated with extremely heavy cost. Consider the toy example on the left side of Fig.~\ref{fig:solution:example}. There are two instances: $d_1$ is responsible for keys $k_1, k_2$ and $k_5$ with costs 7, 4 and 5 respectively, generating $L(d_1) = 16$, and $d_2$ is associated with keys $k_3, k_4$ and $k_6$ with cost 2, 1 and 1 respectively, generating $L(d_2) = 4$. Suppose $\theta_{\max} = 0$, meaning that the total workloads on both instances are required to be equal (i.e., average workload $\bar{L} = 10$). It is clear that $d_1$ is overloaded and $k_1$, which incurs the largest computation cost, is expected to be disassociated from $d_1$. Although $L(d_1)$ decreases to 9, it is still larger than $L(d_2)$. Based on the workflow of LLFD, $k_1$ is assigned to $d_2$, only to overload $d_2$ as a consequence. To tackle the problem, we add a new function, called \emph{Adjust}, to avoid the happening of such conflicts.

Specifically, if re-overloading does not happen after an assignment, i.e., $L_{i-1}(d) + c_{i-1}(k) \leq L_{\max} = (1 + \theta_{\max})\bar{L}_{i-1}$, this assignment is acceptable and \emph{Adjust} immediately returns a TRUE (Lines 12-13).
Otherwise (Lines 14-20), \emph{Adjust} attempts to construct a nonempty key set (called \textit{exchangeable} set and denoted by $\mathcal{E}$), by applying the selection criteria $\psi$ (e.g., highest workload first). The \textit{exchangeable} set must satisfy the following three conditions: (i)  $\mathcal{E} \subseteq \{k' | F(k') = d\}$; (ii)  $\forall k' \in \mathcal{E}, c_{i-1}(k') < c_{i-1}(k)$; and (iii)  $L_{i-1}(d) + c_{i-1}(k) - \sum_{k'\in\mathcal{E}} c_{i-1}(k') \leq L_{\max}$. Basically, (i) means that only keys originally associated with $d$ are selected for disassociation. (ii) tries not to choose a key with larger computation workload for disassociation, ensuring the decrease of the total workloads on instance $d$. Finally, (iii) ensures that instance $d$ does not become overloaded, after the assignment (Lines 15-17).

Recall the running example in which LLFD tries to assign $k_1$ to $d_2$, which makes $d_2$ overloaded. A TRUE is returned by \emph{Adjust} because there exists an $\mathcal{E} = \{k_3\}$ satisfying constraints (i) - (iii). Therefore, $k_1$ is assigned to $d_2$, while $k_3$ is disassociated from $d_2$ and put into $\mathcal{C}$. Next, LLFD attempts to assign $k_3$ to $d_1$, because $d_1$ has less total workload at this moment.
However, a FALSE (a red cross shown on left side of $S_2$ in Fig.~\ref{fig:solution:example}) is returned by \emph{Adjust} because overloading occurs (since $L(d_1) + c(k_3) = 11 > L_{\max}$) and no valid $\mathcal{E}$ exists, when neither of the two keys associated with $d_1$ ($k_2$ and $k_5$) has smaller computation workload than that of $k_3$, violating constraint (ii).
After this failure, LLFD is forced to consider another option, by keeping $k_3$ to $d_2$. Luckily, a TRUE is returned this time, because a valid \textit{exchangeable} set $\mathcal{E} = \{k_4\}$ exists.
After disassociating $k_4$ from $d_2$ and putting it into $\mathcal{C}$, $d_2$ is responsible for $k_1, k_3$ and $k_6$ only, the keys with $d_1$ remains unchanged, and $k_4$ is now in $\mathcal{C}$. The algorithm does not terminate until $\mathcal{C}$ becomes empty, after $k_4$ is assigned to $d_1$, finally reaching perfect balance at $L(d_1) = L(d_2) = 10$.

%An essential question on LLFD is to what extent it can guarantee load balance among instances? The two statements below (proofs are given in Appendix) ensure that LLFD is capable of achieving a bounded degree of load imbalance.

In the following, we present formal analysis on the robustness and soundness of LLFD on basic load balancing problem, with proofs available in  Appendix.~\ref{appsec:proof}.
	
\begin{theorem}\label{theorem:1}
	If there is a solution for absolute load balancing, LLFD always finds a solution resulting with balancing indicator $\theta_i(d,F)$ no worse than $\frac{1}{3}(1-\frac{1}{N_D})$ for any task instance $d_i$.
\end{theorem}

%Figure illustrate the assignment procedure of this example.
%\textbf{Load balance guarantee of LLFD: }

\subsection{MinTable and MinMig Heuristics}\label{subsec:MinTableandMinMigHeuristics}

The general workflow described above is essentially effective in guaranteeing load balance constraints, e.g., the LLFD sub-prodedure. To address the optimizations on routing table minimization and migration cost minimization, we discuss two heuristics, namely MinTable and MinMig in this part of the section.

%In the following, we introduce the algorithms with these two heuristics based on the general workflow introduced at the beginning of the section.

%\vspace{-0.5em}
\begin{algorithm}
%\small
\caption{MinTable Algorithm}
\label{alg:MinTable}
\begin{algorithmic}[1]	
		\State Phase I: \textit{Move back} all keys in $A$.
		\State $\psi \leftarrow$ highest computation cost $c(k)$ first
		\State Phase II: According to $\psi$, select and disassociate keys from each of the overloaded instances, put them into $\mathcal{C}$
		\State Phase III: $A' \leftarrow$ LLFD ($\mathcal{C}, \mathcal{D},\theta_{max},\psi $)
		\State \Return $A'$
\end{algorithmic}
\end{algorithm}

The pseudocodes of MinTable is shown in Algorithm~\ref{alg:MinTable}. In order to minimize routing table size, in Phase I, all entries in routing table $A$ are erased. The highest computation workload first criterion, which emphasizes on the computation cost, is used for the second and third phases, so that minimal number of entries are added into the new routing table $A'$ during the key re-assignment and load rebalance process.

The two toy examples in Fig.~\ref{fig:solution:example} demonstrate how MinTable helps to achieve a smaller routing table while keeping load balance constraints fulfilled. The example on left side of Fig.~\ref{fig:solution:example} initially has two entries in routing table, i.e., ($k_3$, $d_2$) and ($k_5$, $d_1$). LLFD is directly applied to achieve absolute load balance $L(d_1) = L(d_2)$, but resulting in a routing table with four entries at the end. In contrast, before applying LLFD, the example on right side of Fig.~\ref{fig:solution:example} moves back $k_3$ and $k_5$ (i.e., cleaning the routing table). Finally, it results in a routing table with only two entries. The pseudo code of MinMig is shown in Algorithm~\ref{alg:MinMig}. Although the removal of keys from the routing is \emph{virtual} only, it increases the possibility of key migrations. Therefore, there is no cleaning run in the first phase at all.

\begin{algorithm}
    %\small
	\caption{MinMig Algorithm}
	\label{alg:MinMig}
	\begin{algorithmic}[1]	
		\State Phase I: Do nothing.
		\State $\psi \leftarrow$ largest $\gamma_i(k, w)$ first, where $\gamma_i(k, w) = \frac{c_i(k)^\beta}{S_i(k, w)}$
		\State Phase II: According to $\psi$, select and disassociate keys from each of the overloaded instances, put them into $\mathcal{C}$
		\State Phase III: $A' \leftarrow$ LLFD ($\mathcal{C}, \mathcal{D},\theta_{max},\psi $)
		\State \Return $A'$
	\end{algorithmic}
\end{algorithm}
\vspace{-5pt}
To characterize both computation and migration cost, we propose the \textit{migration priority index} for each key, defined as $\gamma_i(k, w) = c_i(k)^\beta S_i(k, w)^{-1}$. Its physical meaning is straightforward, that is, a key with larger computation cost per unit memory consumption has the higher priority to be migrated. The weight scaling factor $\beta$ is used to balance the weights between these two factors under consideration. Consider $k_1$ and $k_2$ in Fig.~\ref{fig:solution:example} and assume window $w = 1$. We have $c(k_1) = S(k_1, w) = 7$ and $c(k_2) = S(k_2, w) = 4$. If we give equal weights to both $c(k)$ and $S(k, w)$, i.e., $\beta = 1$, then $\gamma(k_1, w)$ = $\gamma(k_2, w)$ = 1. When we assign more importance to migration cost, i.e., $\beta = 0.5$, $k_2$ gains higher priority for migration.
In addition, $\beta$ also affects the size of the result routing table, i.e., the larger $\beta$, the smaller size of routing table, which will be shown in the experiment results in Appendix~\ref{appsec:proof}.
The largest $\gamma_i(k, w)$ first criterion, which is aware of both computation and migration cost, is used
during both key re-assignment (Phase II) and load balance process (Phases III), in order to minimize the bandwidth used to migrate the states of keys (e.g., the tuples in sliding window for join operator).

\subsection{Mixed Algorithm}\label{sec:advanced}

Based on the discussion on the heuristics, we discover that there are tradeoffs between routing table minimization and migration cost minimization. Therefore, we propose a mixed algorithm to intelligently combine the two heuristics MinTable and MinMig, in order to produce the best-effort solutions towards our target optimization in Eq.~\ref{EQ:definition}.

The basic idea is to properly mix MinTable (Phase I) and MinMig (Phases II and III). In the first phase, the mixed strategy \textit{moves back} $n$ keys, which are selected from $A$, based on the smallest memory consumption $S_{i-1}(k, w)$ first criteria.
The rest two phases simply follow the procedure of MinMig, in which the largest $\gamma_i(k, w)$ first criteria is used to initialize candidate key set $\mathcal{C}$ and applied by LLFD in the last phase. For the Mixed algorithm, the most challenging problem is how to pick up the number of keys for back moves, i.e., $n\in[0, N_A]$ during the cleaning phase. Actually, MinTable and MinMig works on two extremes of the spectrum in this step, such that $n=N_A$ in MinTable and $n=0$ in MinMig.

\begin{algorithm}
    %\small
	\caption{Mixed Algorithm}
	\label{alg:Mixed}
	\begin{algorithmic}[1]
		\State $\eta\leftarrow$ smallest memory consumption $S_i(k, w)$ first.
		\State $\psi \leftarrow$ largest $\gamma_i(k, w)$ first, where $\gamma_i(k, w) = \frac{c_i(k)^\beta}{S_i(k, w)}$
		\State $n \leftarrow 0$
		\State $A_{backup}\leftarrow A$
		\Do
		\State $A\leftarrow A_{backup}$
		\State Phase I: According to $\eta$, select $n$ keys from $A$ and \textit{move back} them
		\State Phase II: According to $\psi$, select and disassociate keys from each of the overloaded instances, put them into $\mathcal{C}$
		\State Phase III: $A' \leftarrow$ LLFD ($\mathcal{C}, \mathcal{D},\theta_{max},\psi $)		
		\State $n = N_{A'} - A_{max}$
		\DoWhile{$n > 0$}
		\State
		\Return $A'$
	\end{algorithmic}

\end{algorithm}

Obviously, brute force search (named as \emph{Mixed$_{\mbox{BF}}$}) could be applied to try with every possible $n = 1, 2, \dots, N_A$, with the optimal $n^*$ returned after evaluating the solution with every $n$.
Alternatively, we propose a faster heuristic in Algorithm~\ref{alg:Mixed}. It only tries a small number of values, which are the amount of table entries overused in the last trial (Line 10). The trial starts from $n=0$ (Line 3, same as MinMig), and stops when it results in an updated $A'$ of acceptable size, i.e., $N_{A'} \leq A_{max}$ (Lines 11-12). Note that the efficiency of the algorithm is much better than \emph{Mixed$_{\mbox{BF}}$}, although it may not always find the optimal $n^*$ as \emph{Mixed$_{\mbox{BF}}$} does. Obviously, the size of the result routing table by the mixed algorithm is no smaller than that of MinTable approach. Similarly, the migration cost of the result assignment function is no smaller than that of the MinMig approach. However, mixed algorithm is capable of hitting good balance between the heuristics, as is proved in our empirical evaluations.Furthermore, the balance status generated by Mixed as the following theorem (proof in Appendix.~\ref{appsec:proof}):
\begin{theorem}\label{theo:theorem3}
	Balance status generated by Mixed is not worse than that generated by LLFD.
	\end {theorem}

%% file: optimization.tex
\section{Implementation Optimizations}\label{sec:optimization}
%\vspace{-0.5em}
The overall working mechanism of the rebalance control component, as is implemented in our distributed stream processing engine, is illustrated in Fig.~\ref{fig:optimization:overall}. In the figure, each operation step is numbered to indicate the order of their execution.
\begin{figure}[h]
	\centering
	\hspace{-19pt}
	\epsfig{file=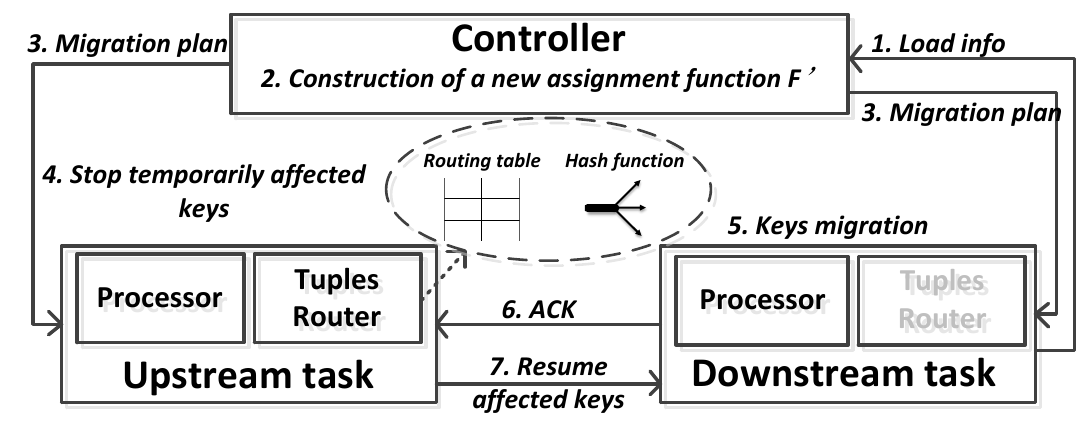, width=0.95\columnwidth}
	%\vspace{-15pt}
%	\vspace{-0.5em}
	\caption{Overall workflow.}
	\label{fig:optimization:overall}
	\vspace{-9pt}
\end{figure}

At the end of each time interval (e.g., 10 seconds as the setting in our experiments), the instances of an operator report the statistical information collected during the past interval to a \textit{controller} module (step 1). The information from each instance $d$ includes the computation cost $c_{i-1}(k)$ and window-based memory consumption $S_{i-1}(k, w)$ of each key assigned to it.
On receiving the reporting information, the controller starts the optimization procedure (step 2) introduced in Section~\ref{sec:solution}. It first evaluates the degree of workload imbalance among the instances and decides whether or not to trigger the construction of a new assignment function $F'$ to replace the existing $F$. If the system identifies load imbalance, it starts to execute Mixed algorithm (Algorithm~\ref{alg:Mixed}) to generate new $A'$ and $F'$.

%Next it broadcasts the new $F'$ together with a \textbf{Pause} signal to all the instances of upstream operator for them to replace the obsolete one and temporarily stop working. Meanwhile, it calculates the keys in $\Delta(F,F')$ that need migration and notifies the corresponding downstream instances.

After calculating the keys in $\Delta(F,F')$ for migration, the controller broadcasts both $F'$ and $\Delta(F,F')$, together with a \textbf{Pause} signal to all the instances of upstream operator for them to update the obsolete $F$, and temporarily stop sending (but caching locally) data with keys in $\Delta(F,F')$ (steps 3 and 4). Meanwhile, the controller notifies the corresponding downstream instances (step 3).

Finally, the instances of downstream operator begin migrating the states of keys after the notification from the controller (step 5) and acknowledge the controller when migration is completed (step 6). As soon as the controller receives all the acknowledgments, it sends out a \textbf{Resume} signal to all instances of the upstream operator, ordering the tasks to start sending data with keys in $\Delta(F,F')$, since all the downstream instances are equipped with the new assignment function (step 7). It is worth noting that during the key state migration, there is no interruption of normal processing on the data with keys not covered by $\Delta(F,F')$.

One potential problem in the workflow above is the cost of transmitting statistical information with the keys in step 1, which could easily contain millions of unique keys in real application domains. The huge size of the key domain may degenerate the scalability of the algorithms, on growing computational complexity and memory consumption for these metrics. To alleviate the transmission problem, we propose a compact representation for the keys with acceptable information loss for the algorithms.

%\subsection{Compact Representations for Keys}

The basic idea is to merge the keys with common characteristics and represent them by a single record in the statistical data structure. To accomplish this goal, we design a new 6-dimensional vector structure for the statistical information, $(d', d, d^h, v_c, v_S, \#)$, in which $d'$ denotes the instance to which a key will be assigned next; $d$ is the instance with which the keys are currently associated during the reporting period (i.e., $d = F(k)$); $d^h$ is the instance assigned by the hash function (i.e., $d^h = h(k)$); $v_c$ denotes the value of computation workload; $v_S$ is the value of window-based memory consumption; and $\# ( > 0)$ is the number of keys satisfying these five conditions. %Note that the value-quintuple \{$d'$, $d^h$, $d$, $v_c$, $v_S$\} of any record in this collection must be unique.
For example, a vector $(d_1, d_2, d_1, 4, 4, 2)$ indicates that there are two keys with computation workload 4 and memory consumption 4. They are currently associated with instance $d_2$, and the instance suggested by hash function is $d_1$, indicating that the routing table $A$ at the upstream operator must contain an entry for them. Finally, the record also implies that they will be assigned to $d_1$, meaning that the entry for them in $A$ is deleted and the \textit{move back} operation is executed on these two keys.

By employing this compact representation, the whole key space $\mathcal{K}$ is transformed to the 6-dimensional vector space (denoted by $\mathcal{K}^c$). The upper bound on the size of the vector space is approximately $K^c = |\mathcal{K}^c| = O(N_D^3\times |c(k)| \times |S(k, w)|)$, where $N_D$ is the number of downstream instances, which is usually a small integer; $|c(k)|$ and $|S(k, w)|$ represent the total numbers of distinct values on computation workload and memory consumption in current sliding window, respectively.

\vspace{-3pt}
\subsection{Mixed Algorithm over Compact Representations}

Apparently, the compact representation brings significant benefits, by reducing both time and space complexity of the Mixed algorithm (and MinTable and MinMig as well) proposed in Section~\ref{sec:solution}. In the following, we briefly describe how Mixed algorithm is revised based on the compact representation. To make a clear description, let us revisit the Mixed algorithm and look into the steps using the compact representations.

%It is worth noting that the Mixed algorithm (and MinTable and MinMig as well) proposed in Section~\ref{sec:solution} can easily adapt to this compact representation for statistical information of keys. To show the clear picture, let us revisit the Mixed algorithm and look into the relevant steps.

\noindent\textbf{Phase I (Cleaning):} According to the smallest $S_{i-1}(k, w)$ first criterion, the Mixed algorithm selects $n$ keys from $A$ and moves them back to original instance based on the hash function. In the compact representation, the adapted Mixed algorithm does not target on any individual key, but the 6-dimensional vectors. In result, a \textit{back-move} of keys is equivalent to modifying the value of $d'$ to be the same as $d^{h}$ of the selected vectors. The vector $(d_1, d_2, d_1, 4, 4, 2)$ mentioned above presents an example back-move of a number of keys.

\noindent\textbf{Phase II (Preparing):} When investigating the workloads of a particular instance, say $d_1$, the adapted Mixed algorithm calculates the weighted sum of $v_c\times\#$ of all records containing $d' = d_1$. If we have two records $(d_1, d_2, d_1, 4, 4, 2)$ and $(d_1, d_2, d_2, 8, 8, 1)$, for example, the total workload with respect to the keys is estimated as 16. Next, when the adapted Mixed algorithm needs to select keys from an overloaded instance and puts them into the candidate set $\mathcal{C}$, the algorithm again targets on those vectors in compact representations, and simply replaces the value of $d'$ with a $nil$, indicating a virtual removal of the keys linked to the vector. For example, if $(d_2, d_2, d_1, 4, 4, 2)$ is disassociated from $d_2$ and put into $\mathcal{C}$, the adapted Mixed algorithm finally rewrites the record as $(nil, d_2, d_1, 4, 4, 2)$. Note that when a record is added into $\mathcal{C}$, i.e., containing $d' = nil$, it is likely that there already exists a record in $\mathcal{C}$ with exactly the same values on $d$, $d^h$, $v_c$ and $v_S$. According to the definition on the uniqueness of the compact representation, these two records need to be merged by summing on the number field `$\#$'.

\noindent\textbf{Phase III (Assigning):} The similar adaptation in Phase I and II is also applied to LLFD, with only one exception that the expected routing table $A'$ can not be directly derived but rather indirectly calculated. This is because the final results returned by the adapted LLFD is still in a compact form, as a 6-dimensional tuple. In order to derive $A'$ and $F'$, a series of additional actions are taken, including (i) picking up those needing migration from the records returned by adapted LLFD, i.e., tuples with $d' \neq d$; (ii) selecting keys originally associated with instance $d$ and computation cost at $v_c$, according to the selection criteria $\psi$ and based on the original complete statistical information of keys collected by the controller; and (iii) adding them to the key migration set $\Delta(F,F')$.
Finally, the adapted algorithm returns the final result, with $F'$ induced by combining $F$ and $\Delta(F,F')$, and $A'$ derived with $F'$ together with $h(k)$.

\vspace{-3pt}
\subsection{Discretization on $v_c$ and $v_S$}

As is emphasized above, the size of the 6-dimensional vector space $K^c = O(N_D^3\times |c(k)| \times |S(k, w)|)$ depends on $|c(k)|$ and $|S(k, w)|$. In practice, the values of computation cost and memory consumption could be highly diversified, leading to large $|c(k)|$ and $|S(k, w)|$, and consequently huge vector space $K^c$. It is thus necessary to properly discretize the candidate values used in $c(k)$ and $S(k, w)$, in order to control the complexity blown by $|c(k)|$ and $|S(k, w)|$.

Value discretization can be done in a straightforward way. However, an over-simplistic approach could cause huge deviations on the approximate values from the real ones. Such deviations may affect the accuracy of workload estimation and jeopardize the usefulness of the migration component. For instance, assume that there are 10 keys with computation costs $c(k_1) = 8$, $c(k_2) = 6$, $c(k_3) = 3$, $c(k_4) = c(k_5) = 2$ and $c(k_6) = \dots = c(k_{10}) = 1$. If a simple piecewise constant function is used, i.e., $\xi(x) = 2$ when $x\in[1, 3]$, $\xi(x) = 5$ when $x\in[4, 6]$, $\xi(x) = 8$ when $x\in[7, 9]$, and 0 otherwise. Despite of a very small $|c(k)| = 2$, the total deviation, following the formula below, caused by the approximation is fairly large:
\[|\delta| = \left|\sum_{i=1}^{10} \delta_i\right| = \left|\sum_{i=1}^{10} c(k_i) - \xi(c(k_i))\right|.\]
Fig.~\ref{fig:level_naive} illustrates how this simple piecewise constant function fails and the deviation $\delta_i$ with respect to each value point.
\begin{figure}[htp]
	\centering
	\begin{tabular}[t]{c}
		\hspace{-26pt}
		\subfigure[A simple approach]{
			\includegraphics[height = 3.8cm, width = 4.6cm]{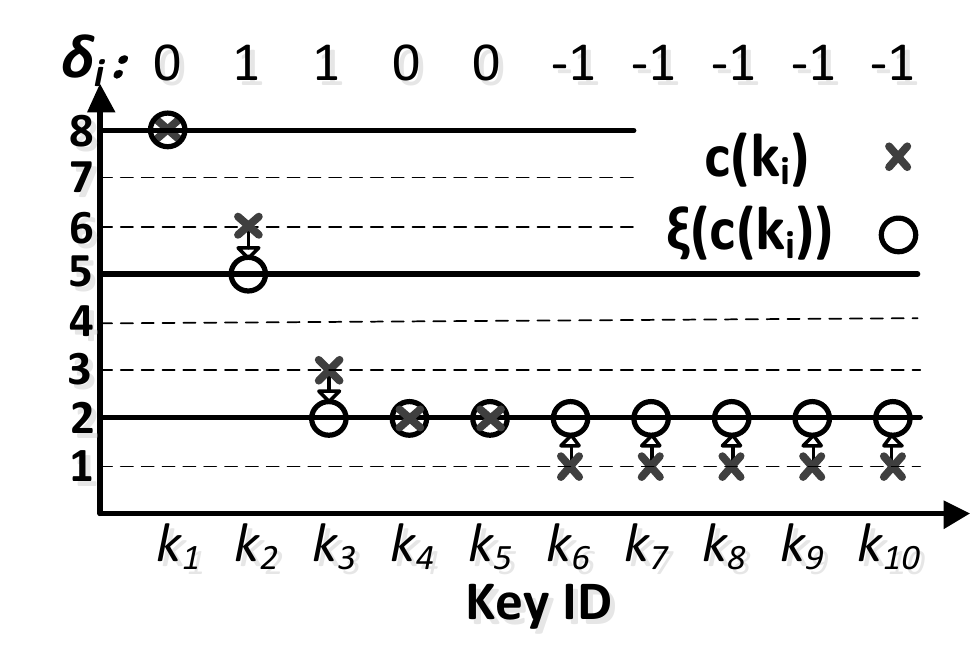}
			\label{fig:level_naive}
		}
		\hspace{-12pt}
		\subfigure[Our proposed approach]{
			\includegraphics[height = 3.8cm, width = 4.6cm]{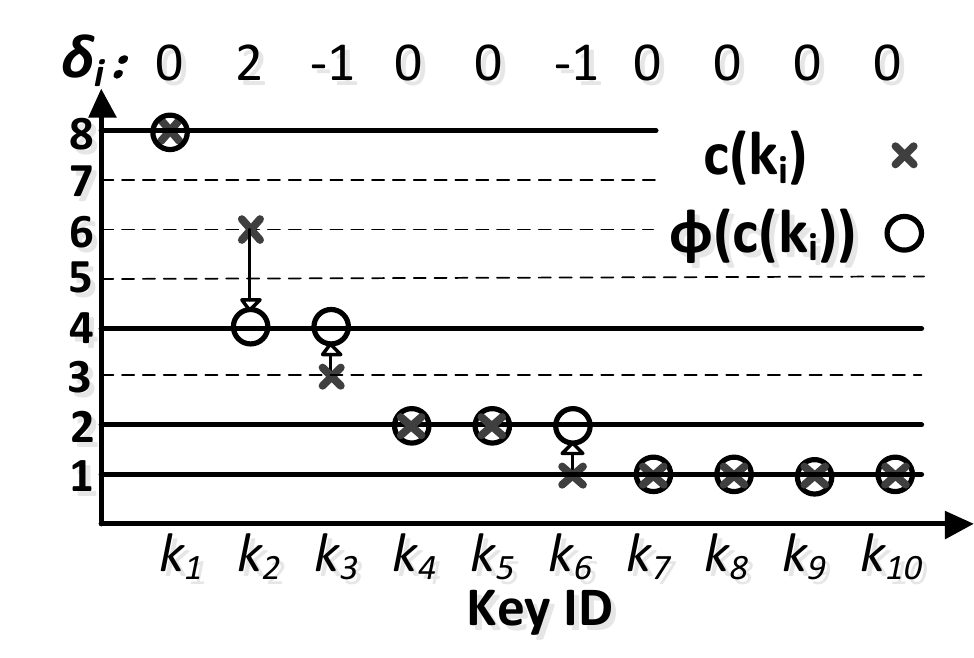}
			%\caption{Stock Data}
			\label{fig:level_our}
		}
	\end{tabular}
%	\vspace{-12pt}
	\caption{An example of comparing the simple piecewise discretization function (a) and our proposed approach (b).}
	\label{fig:levelexample}
%	\vspace{-10pt}
\end{figure}

To tackle the problem, we propose an improved discretization approach, denoted by $\phi(x)$, which involves two steps. In the first step, it generates a finite number of representative values. Secondly, instead of using the nearest representative for each value independently, our approach constructs the discretized values in a more holistic manner. Assume the input value is a series of $n$ numbers in a non-increasing order by their values of which the smallest is at least 1 (after normalization), i.e., $x_1$, $x_2$, $\dots$, $x_n$, $\forall i, i' \in [1, n], i < i', x_i \geq x_{i'} \geq 1$.

In the first step, a simple method (half-linear-half-exponential, HLHE) with a parameter $R$ (called the degree of discretization) is applied to determine the representative values, where we require $R = 2^r, r = 0, 1, 2, \dots$. Therefore, a total number of \[m = r + \lfloor\frac{\max( {x_i})}{R}\rfloor = r+s\]
representative values are generated and reorganized as a strictly decreasing series, $y_1$, $y_2$, $\dots$, $y_m$,
%$\forall j, j' \in [1, m], j < j', y_j > y_{j'}$,
where $y_1 = s \times R$, $y_2 = (s-1) \times R$, $\dots$, $y_s = R$ (the linear part), and $y_{s + 1} = R / 2 = 2^{r-1}$, $y_{s+2} = 2^{r-2}$, $\dots$, $y_{m-1} = 2$, $y_{m} = 1$ (the exponential part).

In the second stage, a greedy method is applied to finalize the discretization by adopting an optimization framework. The basic principle is to minimize the accumulated error of all values, such that the sum over an arbitrary set of approximate values tend to be an accurate estimation to the sum over original values. Specifically, for each $x_i < y_1$, two representative values $j\in[2, m]$ that $y_{j-1} > x_i \geq y_j$ can be used to approximate $x_i$. We define such $y_{j-1}$ and $y_j$ as candidate representative values for $x_i$. For the remaining ($x_i \geq y_1$), they only have one candidate representative value, which is $y_1$. For each $x_i$, one of the candidate representative values is chosen when there are two options, denoted by $\phi(x_i)$, so that the total deviation $|\delta|$ is minimized. In particular, if the current accumulated deviation is positive, $x_i$ is represented by the larger value $y_{j-1}$ in order to cancel the over-counting. Otherwise, $x_i$ chooses the representative value $y_{j}$.

In the example of Fig.~\ref{fig:level_our}, we let $r = 2$ and $R = 4$, thus $m = 2 + \frac{8}{4} = 4$. There are four representative values, e.g., $y_1 = 8$, $y_2 = 4$, $y_3 = 2$, $y_4 = 1$. At the time $k_3$, whose $c(k_3) = 3$, is processed, the two representative values for it are $y_2 = 4$ and $y_3 = 2$ respectively. Since the accumulated deviation caused by $k_1$ and $k_2$ equals to 2, we have $\phi(c(k_3)) = y_2 = 4$. This results in a reduction on the accumulated deviation by 1. When our proposed approach terminates, according to Fig.~\ref{fig:level_our}, the total deviation is zero, while the simple piecewise constant function generates a total deviation at $|\delta| = 3$.

Based on discussion above, we have the following theorem: 
	\begin{theorem}\label{theorem:2}
		The value discretization always can be done perfectly($|\delta|\sim 0$) by the above two steps.
	\end{theorem}
\begin{proof}
	Due to data skew, the number of small load key is always more than that of big load key. Futhermore, the smallest representative values $ r = 0, 1, 2 $ will not cause the accumulated error of values.
	Then, the value discretization always can be done perfectly($|\delta|\sim 0$) by the above two steps.
\end{proof}

%% file: evaluation.tex
\section{Evaluations}\label{sec:evaluations}

%\subsection{Experimental setup}\label{sec:eva:setup}

%\textbf{\emph{Environment:}}

\eat{
\begin{figure}[htp]
\centering
\begin{tabular}[t]{c}
\hspace{-20pt}
   \subfigure[\emph{social} Data]{
      \includegraphics[height = 3.3cm, width = 4.6cm]{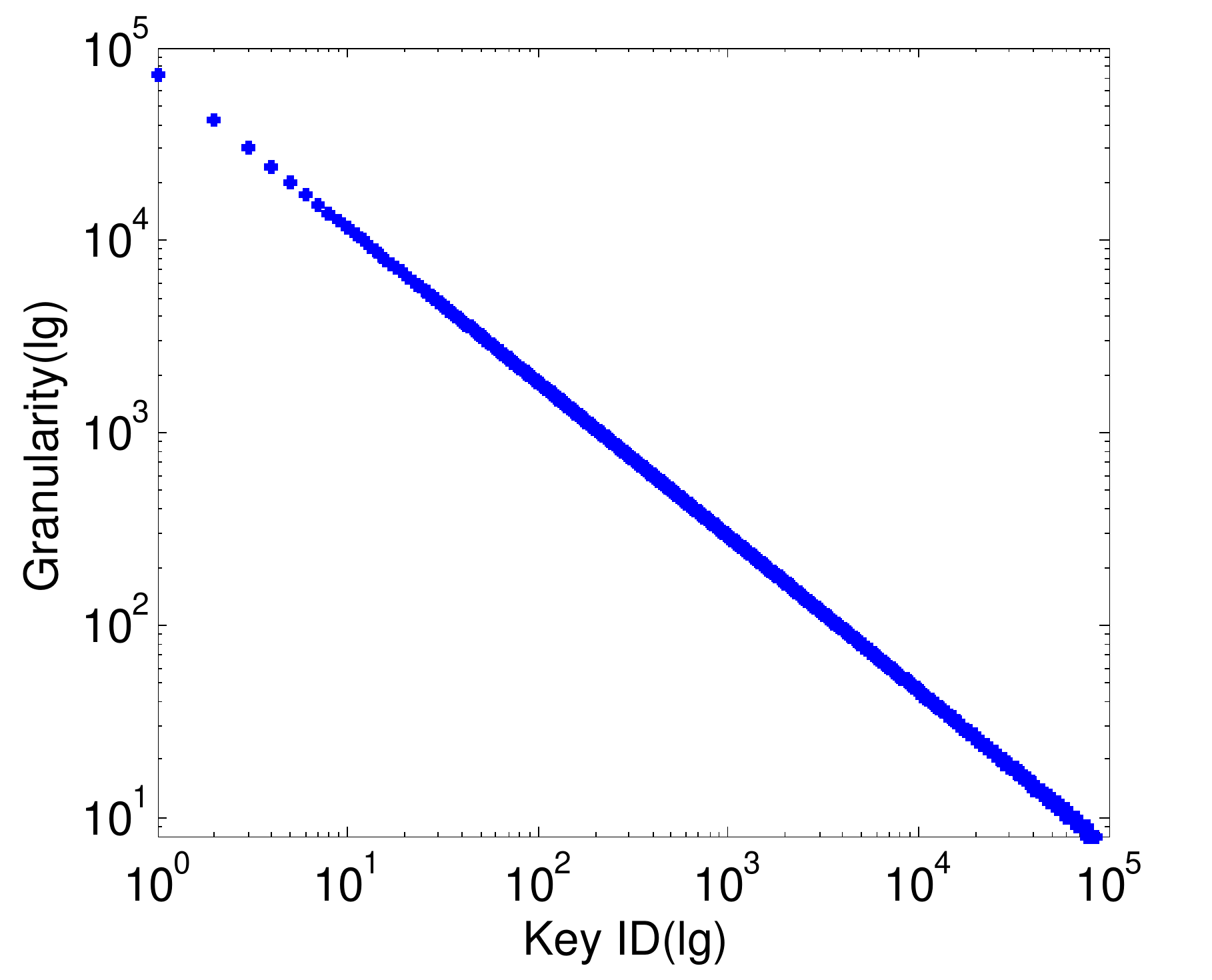}
      \label{fig:exp:StormRealTheta}
}
\hspace{-12pt}
\subfigure[\emph{Stock} Data]{
\includegraphics[height = 3.2cm, width = 4.6cm]{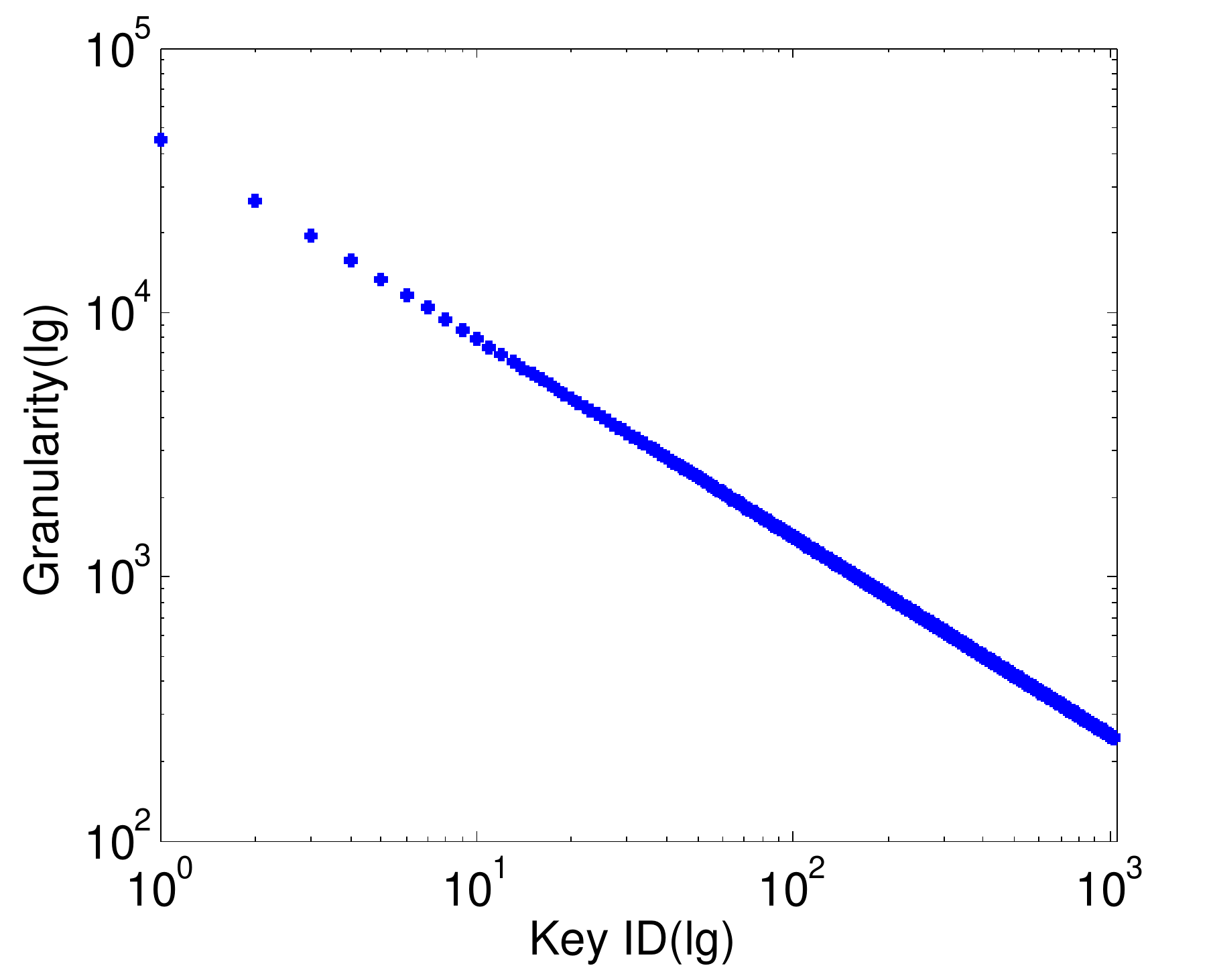}
%\caption{Stock Data}
\label{fig:Stockdata}
 }
\end{tabular}
%\vspace{-5pt}
 \caption{Distribution on real data}
 \label{fig:StormReal}
% \vspace{-10pt}
\end{figure}
}

In this section, we evaluate our proposals by comparing against a handful of baseline approaches. All of these approaches are implemented and run on top of \emph{Apache Storm}~\cite{url:Storm}  under the same task configuration $N_D$ and routing table size $N_A$. To collect the workload measurements, we add a load reporting module into the processing logics when implementing them in \emph{Storm}'s topologies. Migration and scheduling algorithms are injected into the codes of \emph{controllers} in \emph{Storm} to enable automatic workload redistribution. We use the consistent hashing\cite{karger1997consistent} as our basic hash function and configure the parallelism of spout at 10. By controlling the latency on tuple processing, we force the distributed system to reach a saturation point of CPU resource for the $N_D$ number of processing tasks with the requirement of absolute load balancing ($\theta_{max}=0$). We show the results are averages of 5 runs. The \emph{Storm} system (in version 0.9.3) is deployed on a 21-instance HP blade cluster with CentOS 6.5 operating system. Each instance in the cluster is equipped with two Intel Xeon processors (E5335 at 2.00GHz) having four cores and 16GB RAM. Each core is exclusively bound with a worker thread during our experiments.

%Defaultedly we set bottleneck resource as $CPU$ on Storm system. For each instance, we simulate CPU consumption per key by adding a fixed delay to the processing (e.g. 1 millisecond per tuple) and then we can get the saturation point (Storm Ui: capacity=1) on each instance. We have set $\textit{setMaxSpoutPending(50)}$ to control tuples entering system. We run experiments on both synthetic data and real data.

\begin{table}[H]
\vspace{-10pt}
\centering
\small
\caption{Parameter Settings}\label{table:parameters}
\begin{tabular}{|l|l|p{0.45\columnwidth}|}
\hline
%\diagbox{parm}{desc} & Range & Default & Description \\
%\diagbox pc& Range & Default & Description \\
& Range &  Description \\
\hline $K$ &[$5\cdot 10^{3}$,$10^{4}$,$10^{5}$,\textbf{\emph{10}}$^{6}$] & Size of key domain\\
\hline $z$ & [0,...,\textbf{\emph{0.85}},...,1.0]&  Distribution skewness \\
\hline $f$ & [0,...,\textbf{\emph{1.0}},...,2.0]& Fluctuation rate \\
\hline $\theta_{\max}$ & [0,...,\textbf{\emph{0.08}},...,1.0] & Tolerance on load imbalance \\
\hline $\beta$ & [1,...,\textbf{\emph{1.5}},...,2.0] & Migration selection factor \\
\hline $r$ & [0,1,2,\textbf{\emph{3}},4,5,6,7,8] &  Level partition distance \\
\hline $w$ & [1,\textbf{\emph{5}},10,15,20] & Time interval \\
%\hline $W$ & $ [0,50] $ & $ 5 $ & Buffered state window \\
%\hline $\mathbb{D} $ & $ [10^{4},10^{8}] $ & $ 10^{7} $ & Range of key size among which it generates $|\mathcal{K}|=K$ \\
\hline ${N_D}$ & [1,5,10,\textbf{\emph{15}},20,...,40] &  Number of task instances \\
\hline $N_A$ & [0,...,\textbf{\emph{3}}$\cdot$\textbf{\emph{10}}$^{3}$,...,$5\cdot 10^{4}$] & Size of the routing table\\
%\hline $\alpha$ & $ [1,4] $ & $ 1 $ & Imbalance tolerance for multi-dimension \\
\hline
\end{tabular}
%\vspace{-5pt}
\end{table}

\vspace{3pt}
\noindent\textbf{Synthetic Data: }Our synthetic workload generator creates snapshots of tuples for discrete time intervals from an integer key domain $K$. The tuples follow Zipf distributions controlled by skewness parameter $z$, by using the popular generation tool available in Apache project. %\footnote{\url{http://commons.apache.org/proper/commons-math/apidocs/org/apache/commons/math3/distribution/ZipfDistribution.html}}.
We use parameter $f$ to control the rate of distribution fluctuation across time intervals. At the beginning of a new interval, our generator keeps swapping frequencies between keys from different task instances until the change on workload is significant enough, i.e., $\frac{|L_i(d)-L_{i-1}(d)|}{\bar{L}}\geq f$.
Parameter $\theta_{\max}$ is defined as the tolerance of load imbalance, measured as the ratio of maximal workload to minimal workload among the task instances. Our algorithm \emph{Mixed} is controlled by two more parameters $\beta$ and $r$, as defined in previous sections.
The range of the parameters tested in the experiments are summarized in Tab.~\ref{table:parameters}, with default values highlighted in bold font. We also employ \emph{TPC-H} tool \emph{DBGen}~\cite{url:TPCH} to generate a synthetic warehousing workload. And we revise \emph{Q5} in \emph{TPC-H} into a continuous query over sliding window as our testing target, because \emph{Q5} includes all primitive database operations.

%Stream Dynamics is represented by $cv$. Both the change of flow rate and key distribution can generate \textit{cv}.
%$\textit{cv}^{+}$ and $\textit{cv}^{-}$ denote load increasing/decreasing on instances from time $t$ to $t+1$, calculated by:
%$$\left\{
%\begin{aligned}
%cv^{+} = \sum_{d_{i} \in \mathcal{D} \bigwedge L^t(d_i) \leq  L^{t+1}(d_i)} {\frac{L^{t+1}(d_{i}) - L^{t}(d_{i})}{\overline{L}}} \\
%cv^{-} = \sum_{d_{i} \in \mathcal{D} \bigwedge L^{t+1}(d_i) \leq L^t(d_i) }{\frac{L^{t}(d_{i}) - L^{t+1}(d_{i})}{\overline{L}}}
%\end{aligned}
%\right.
%$$
%Supposing we have $N_D$ instances at $t+1$, $0 \leq cv^{+},cv^{-} < ({N_D}^{t+1}-1) $, where low bound $ 0 $ represents load of each instance is unchanged form $ t $ to $ t+1 $ and the up bound $(N_D^{t+1}-1) $ is an extreme case with almost all loads located on one instance.
%In order to measure the influence of stream variance to our performance, we fix the total steam volume for each unit time and  control the change of $\textit{cv}^+$ to achieve different stream dynamics.   $|\textit{cv}^{+} - \textit{cv}^{-}| \sim 0$ simulates the flow with almost no change that is $\textit{cv}=0$. If  $|\textit{cv}^{+} - \textit{cv}^{-}|\gg 0$, it means a big variance in data distribution.  $\textit{cv}^{+}$ is a parameter to measure the degree of change to $\textit{cv}$.

\vspace{3pt}
\noindent\textbf{Real-World Data: } We also use two real workloads in the experiments. The first \emph{Social} workload includes 5-day feeds from a popular microblog service, in which each feed is regarded as a tuple with words as its keys. There are over 5,000,000 tuples covering 180,000 topic words as the keys. The second workload includes 3-day \emph{Stock} exchange records, consisting of over 6,000,000 tuples with 1,036 unique keys (Stock ID) for stock transactions. For both datasets, we take each day as a time interval, so the workload inside one window size consists of the tuples in the last 24 hours. We run word count topology on \emph{Social} data, which continuously maintaining current tuples in memory and updating the appearance frequency of topic words in social media feeds. We run self-join on \emph{Stock} data over sliding window, used to find potential high-frequency players with dense buying and selling behavior. These two workloads evolve in completely different patterns. To be specific, the word frequency in \emph{Social} data usually changes slowly, while \emph{Stock} data contains more abrupt and unexpected bursts on certain keys.

\vspace{3pt}
\noindent\textbf{Baseline Approaches: }We use \emph{Mixed} to denote our proposed algorithm mixing two types of heuristics. We also use \emph{Mixed$_{\mbox{BF}}$} to denote the brute force version of \emph{Mixed} method, which completely rebuilds the routing table from scratch at each scheduling point. We use \emph{MinTable} to denote the algorithm always trying to find migration plan generating minimal routing table. Finally, we also include \emph{Readj} and \emph{PKG} as baseline approaches, which are known as state-of-the-art solutions in the literature. \emph{Readj} is designed to minimize the load of restoring the keys based on the hash function, implemented by key rerouting over the keys with maximal workload. The migration plan of keys for load balance is generated by pairing tasks and keys. For each task-key pair, their algorithm considers all possible swaps to find the best move alleviating the workload imbalance.
%Then, keys are migrated according to the migration plan.
	In \emph{Readj}, $\sigma$ is a configurable parameter, deciding which keys should take part in action of swap and move. Given a smaller $\sigma$, \emph{Readj} tend to track more candidate keys and thus finding better migration plans.
	In order to make fair comparison, in each of the experiment, we run \emph{Readj} with different $\sigma$s and only report the best result from all attempts.
\emph{PKG} \cite{nasir2015power} is a load balancing  method without migration at runtime. It balances the workload of tasks by splitting keys into smaller granularity and distributing them to different tasks based on randomly generated plan. 
%
%In other words, \emph{PKG} preferentially assigns the tuples of one key into lowly loaded task among the tasks which have been marked as the destinations.
Here, we only use \emph{PKG} approach for simple aggregation processing in the experiments, because it does not support complex stateful operations, such as join.
Due to the unique strategy used by \emph{PKG}, aggregation topologies run on \emph{PKG} must contain a special downstream operator in the topology, which is used to collect and merge partial results with respect to every key, from two independent workers in the upstream operator. Moreover, in the open source version of \emph{PKG}\footnote{\emph{https://github.com/gdfm/partial-key-grouping}}, there is a parameter $p$ indicating the time interval between two consecutive result merging. After careful investigation with experiments, we find a larger $p$ prolongs the response time of tuple processing, reduces the additional computation cost and limits the maximal number of live tuples (known as maximal pending tuples in \emph{Storm}) under processing in the system. We finally chose $p$ at 10 milliseconds and set maximal pending tuples at 50, which generally maximizes the throughput of \emph{PKG} in all settings.
%\emph{p}  influences the value of \textit{complete\ latency} under the guarantee mechanism(ack mechanism) for message reliability. The \textit{complete\ latency} is the time period of processing a tuple, from the time which is emitted from the spout to the time that it is acked on the spout. The ack mechanism is used to track the processing time for the whole tuple tree. In order to ensure the running conditions (e.g., complete latency and reliability) of each algorithm are identical in the following experiments, we use the source code} of \emph{PKG} by adding the ack mechanism to it. Futhermore, we set $\emph{p}=0.01$ second and the maximum number of tuples that can be appended to a spout task at any given time is 50.
%For each pair of the task instances, \emph{Readj} considers all combinations of swaps among the \emph{hot} keys, and then returns the one with minimal migration overhead.
Note that we do not include \emph{LLFD} and \emph{MinMig} algorithms in the experiments, because both of them can not control the size of routing tables, therefore blowing off the memory space of the tasks in some cases.

\vspace{3pt}
\noindent\textbf{Evaluation Metrics:} In the experiments, we report the following metrics. Workload skewness (i.e., $\frac{\max L(d)}{\overline{L}}$), is the ratio of maximal workload on individual task instance to the average workload.
Migration cost reveals the percentage of states associated with the keys involved in migration over the states maintained by all task instances.
Throughput is the average number of tuples the system processes in unit second. Average generation time is the average time spent on the generation of migration plan in Storm controller. Finally, processing latency is the average latency of individual tuples, based on the statistics collected by Storm itself. In the rest of the section, we report the average values for these metrics over complete processing procedure, as well as the minimal and maximal values when applicable, to demonstrate the stability of different balance processing algorithms.

%\subsection{Performance Evaluation}
%\subsection{Load Skewness Phenomenon}\label{sec:eva:skewness}

\begin{figure}
\centering
\begin{tabular}[t]{c}
\hspace{-20pt}
   \subfigure[ \# of task instances]{
      \includegraphics[height = 3.6cm, width = 4.6cm]{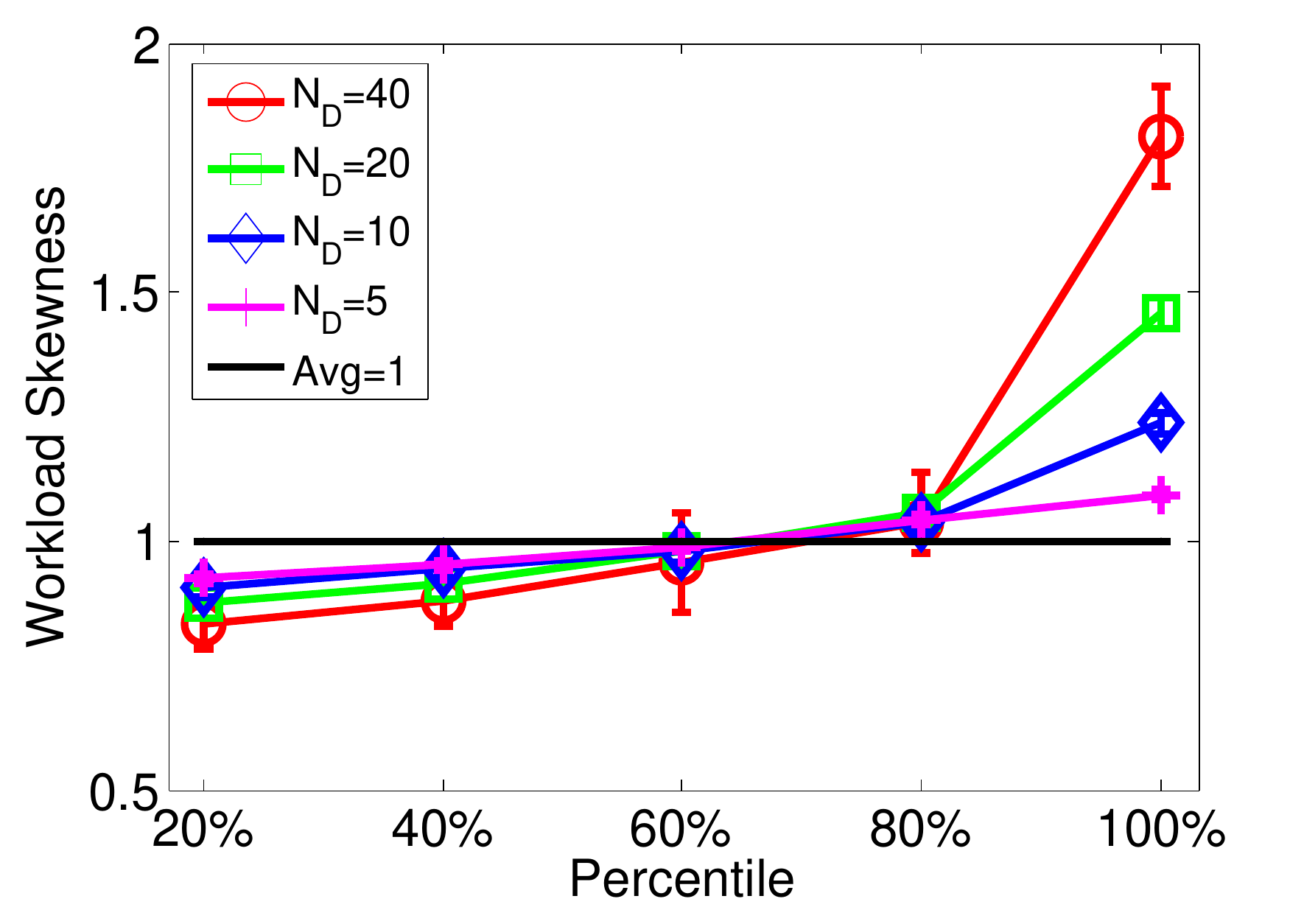}
      \label{fig:exp:parameterN}
   }
\hspace{-12pt}
   \subfigure[Size of key domain]{
      \includegraphics[height = 3.6cm, width = 4.6cm]{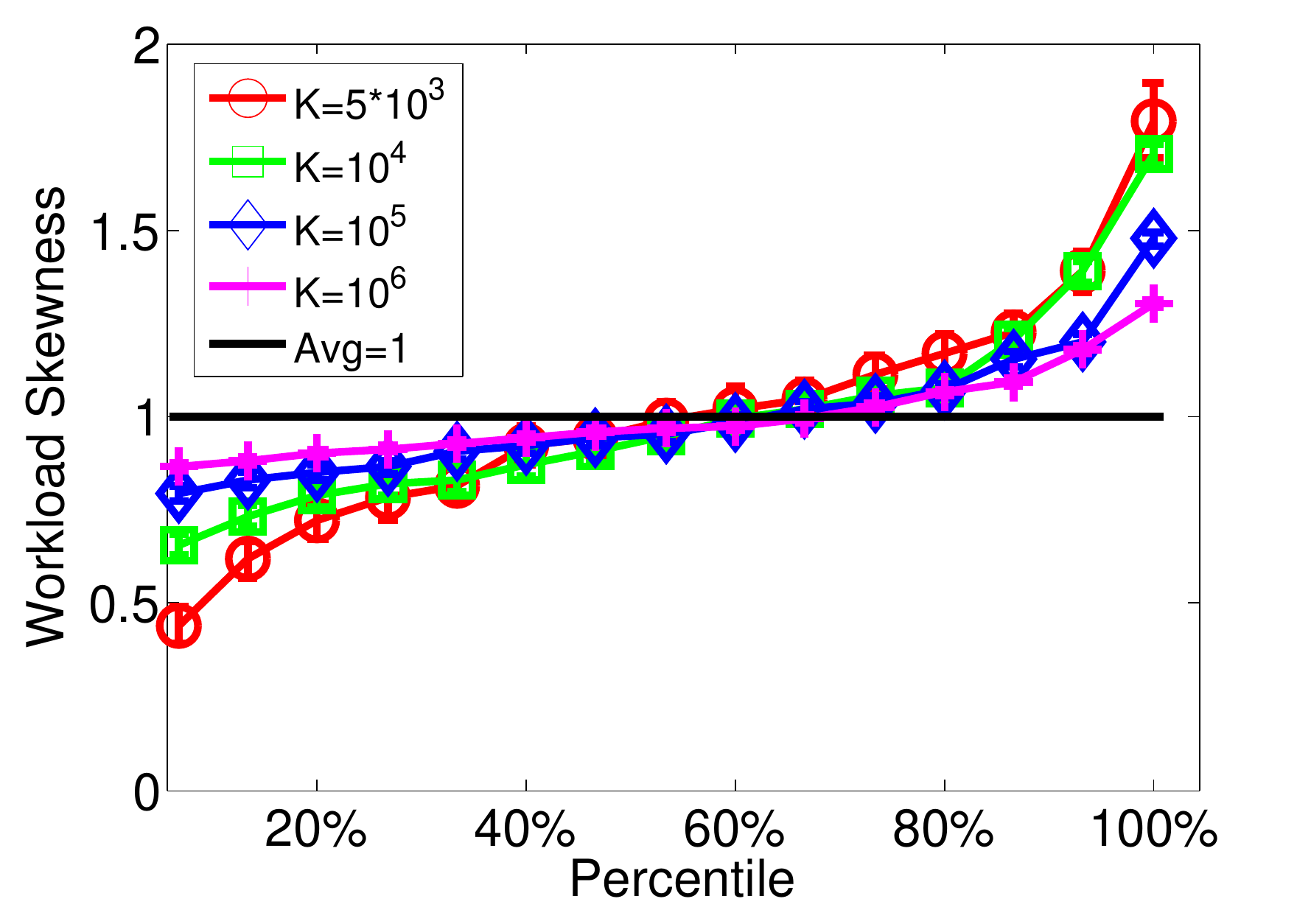}
      \label{fig:exp:parameterK}
}
\end{tabular}
% \vspace{-15pt}
 \centering \caption{Cumulative distribution of workload skewness under hash-based scheme.}\label{fig:exp:skew}
 \vspace{-10pt}
 %\label{fig:Parameters}
\end{figure}

\vspace{3pt}
\noindent\textbf{Load Skewness Phenomenon: }To understand the phenomenon of workload skewness with traditional hash-based mechanism, we report the workload imbalance phenomenon on the task instances by changing the number of task instances and the size of key domain, respectively.
%
%Fig.~\ref{fig:exp:parameterN} and ~\ref{fig:exp:parameterK} show the load skew distribution by consistent hash under defaulted data skew settings with different size of $N_D$ and $K$. $Avg$ means the absolute load distribution among instances by random distribution.
%
The results of load imbalance in Fig. \ref{fig:exp:skew} are presented as the cumulative distribution of average workload among the task instances over 50 time intervals. Fig. \ref{fig:exp:parameterN} implies that the skewness grows when increasing the number of task instances. When there are 40 instances (i.e., $N_D=40$),  the maximal workload at 100\% percentile is almost 2.5 times larger than the minimal workload. Fig. \ref{fig:exp:parameterK} shows that the workload imbalance is also highly relevant to the size of key domain. When there are more keys in the domain, the hash function generates more balanced workload assignment. In Fig.\ref{fig:exp:parameterK}, the maximal workload for $K=$5,000 is around $4$ times larger than the minimal one and is much larger than the maximal load under larger key domain size (e.g., $K=$1,000,000). Therefore, workload imbalance for intra-operator parallelism is a serious problem and cannot be easily solved by randomized hash functions.
%In Fig.~\ref{fig:exp:parameterN}, $20 \%$ represents the position of the ordered instances, e.g. the first instance for $N_D=5$, the second one for $N_D=10$ or the $8^{th}$ one when $N_D=40$ along the ordered list.
%In Fig.~\ref{fig:exp:parameterN}, $InstanceSkewDis$ is heavily related to $N_D$, and  the heaviest one has two times of load compared to average load e.g. $InstanceSkewDis\approx 2$ for $N_D=40$. In Fig.~\ref{fig:exp:parameterK}, it shows that when we increase $K$, it makes load more balanced with defaulted $15$ instances. From both figures, we can see that skewness on the overloaded instances are severer than that on the underloaded ones for they have bigger $InstanceSkewDis$ values.

\vspace{3pt}
%\subsection{Impact of Algorithm Parameters}
\noindent\textbf{Impact of Algorithm Parameters: }We test the algorithm parameters on synthetic datasets using two window sizes (i.e., $w=1$ and $w=5$), in order to understand their impacts for short and long term aggregation over stream data. When $w=1$, migration decisions are made based on the current stateful and instantaneous workload. When $w=5$, more state information in the last five intervals are included in the decision making procedure.

\begin{figure}
\centering
\begin{tabular}[t]{c}
\hspace{-20pt}
   \subfigure[ \# of instances \emph{vs} scheduling efficiency]{
      \includegraphics[height = 3.9cm, width = 4.7cm]{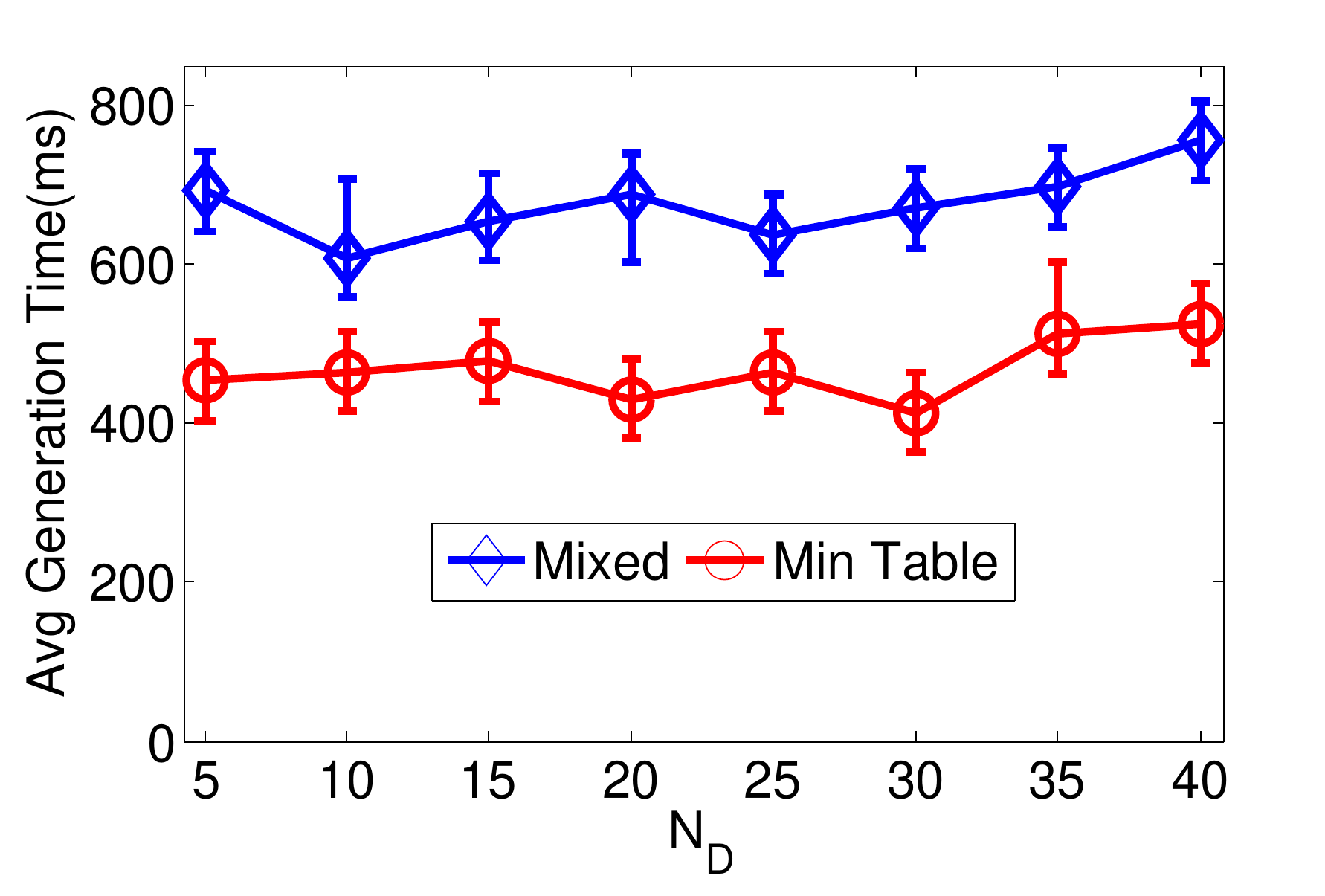}
      \label{fig:exp:InstanceNumberEfficiency}
   }
\hspace{-12pt}
   \subfigure[\# of instances \emph{vs} migration cost]{
      \includegraphics[height = 3.6cm, width = 4.7cm]{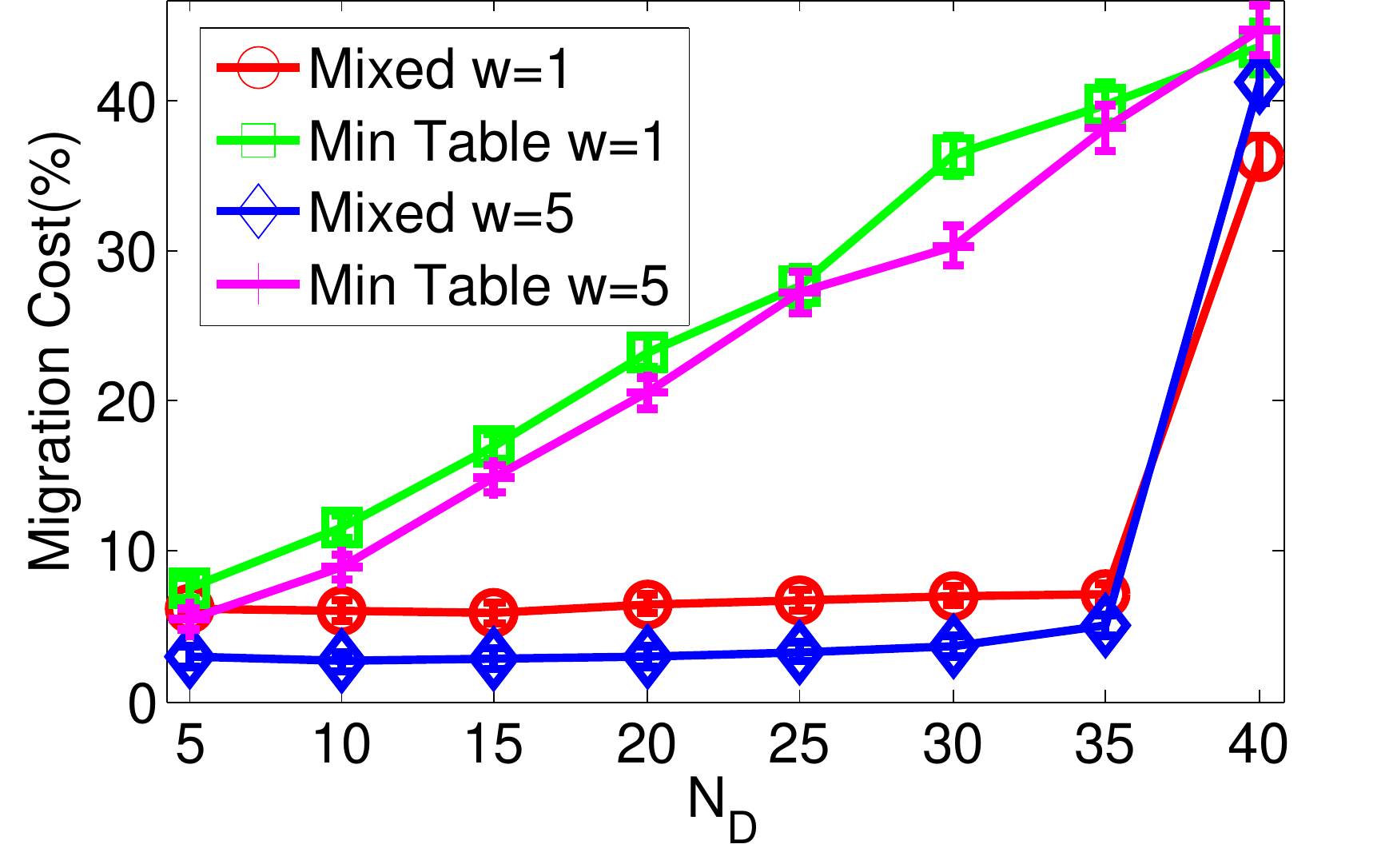}
      \label{fig:exp:InstanceNumberMC}
}
\end{tabular}
 \vspace{-5pt}
 \centering \caption{Performance with varying number of task instances.}\label{fig:exp:InstanceNumber}
 \vspace{-10pt}
 %\label{fig:Parameters}
\end{figure}

Although the increase on $N_D$ produces more workload imbalance, our migration algorithm \emph{Mixed} performs well, by generating excellent migration plan, as shown in Fig.~\ref{fig:exp:InstanceNumber}. \emph{Mixed} costs a little additional overhead over \emph{MinTable} algorithm for balancing, but its migration cost is much lower than \emph{MinTable} when $N_D\leq 35$ for both $w=1$ and $w=5$, as presented in Fig.~\ref{fig:exp:InstanceNumberMC}. The cleaning step in \emph{MinTable} algorithm also leads to even higher skewness and much more migration cost in order to achieve load balancing. When $w=5$, \emph{Mixed} keeps more historical tuples which can be used as the migration candidates. This makes the migration easier and less expensive, when compared to the case with $w=1$. When $N_D> 35$, however, the migration cost of \emph{Mixed} jumps, almost reaching the cost of \emph{MinTable} when $N_D=40$.  This is because the outcome of \emph{Mixed} algorithm degenerates to that of \emph{MinTable} algorithm, when the minimal routing table size needed for target load balancing exceeds the specified size of the table in the system.

Fig.~\ref{fig:exp:Theta} displays the efficiency of migration plan generation and the corresponding migration cost with varying workload balancing tolerance parameter $\theta_{\max}$. As expected, Migration scheduling runs faster on synthetic dataset with larger $\theta_{\max}$ in Fig.~\ref{fig:exp:thetaEfficiency}.  When $\theta_{\max}\geq 0.2$, the efficiency of \emph{Mixed} catches that of \emph{MinTable}. If stronger load balancing (i.e., smaller $\theta_{\max}$) is specified, system pays more migration cost as shown in Fig.~\ref{fig:exp:thetaMC}, basically due to more keys involved in migration. But \emph{MinTable} incurs three times of the migration cost of \emph{Mixed} under the same balance requirement. Even for strict $\theta_{\max}=0.02$ (almost absolutely balanced), the algorithm is capable of generating the migration plan within 1 second. Moreover, migration cost with larger window size (i.e., $w=5$) shrinks, as the historical states provide more appropriate candidate keys for migration plan generation.

\begin{figure}
\centering
\begin{tabular}[t]{c}
\hspace{-20pt}
 \subfigure[$\theta_{\max}$ \emph{vs} scheduling efficiency]{
      \includegraphics[height = 3.9cm, width = 4.7cm]{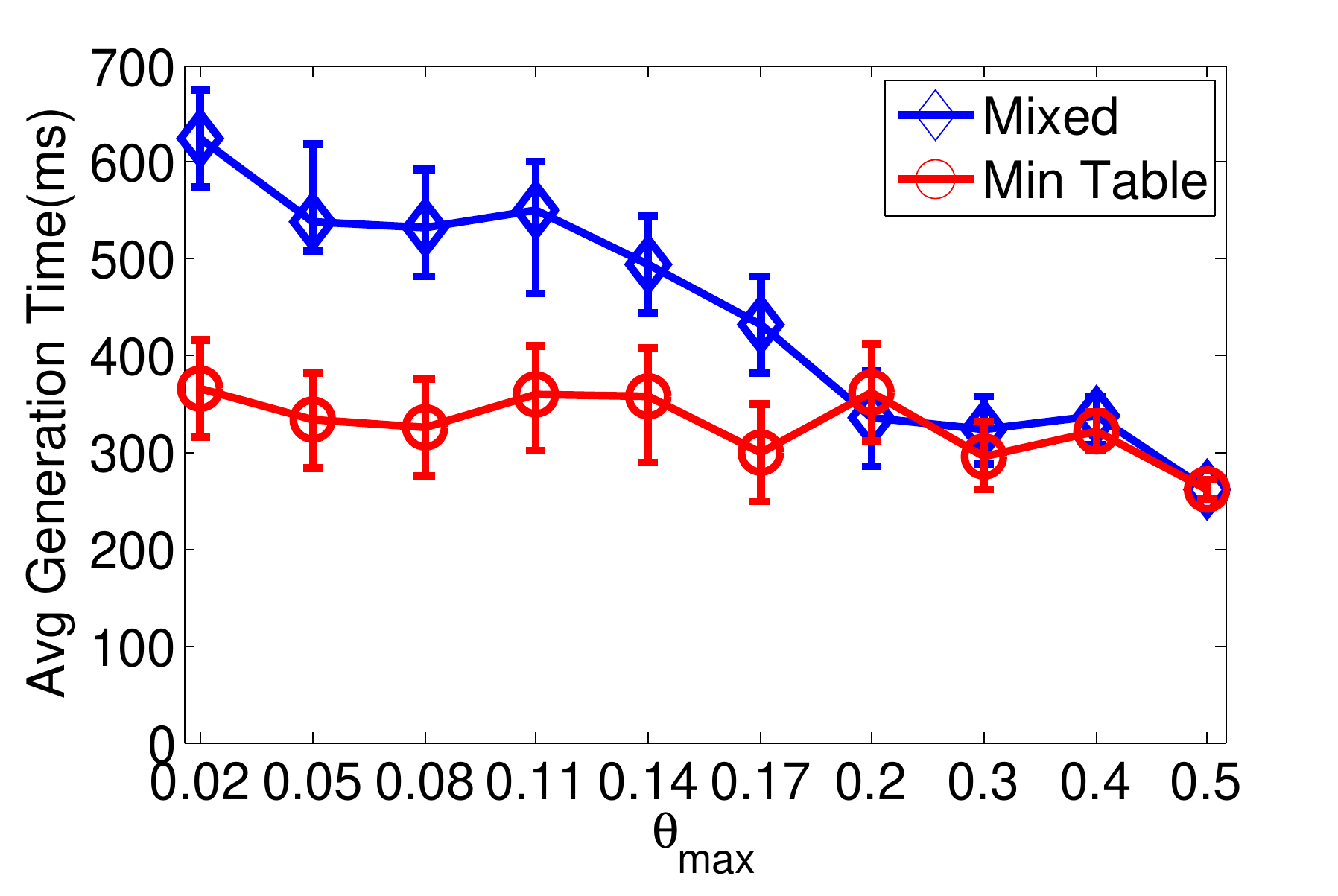}
      \label{fig:exp:thetaEfficiency}
}
\hspace{-12pt}
   \subfigure[$\theta_{\max}$ \emph{vs} migration cost]{
      \includegraphics[height = 3.6cm, width = 4.7cm]{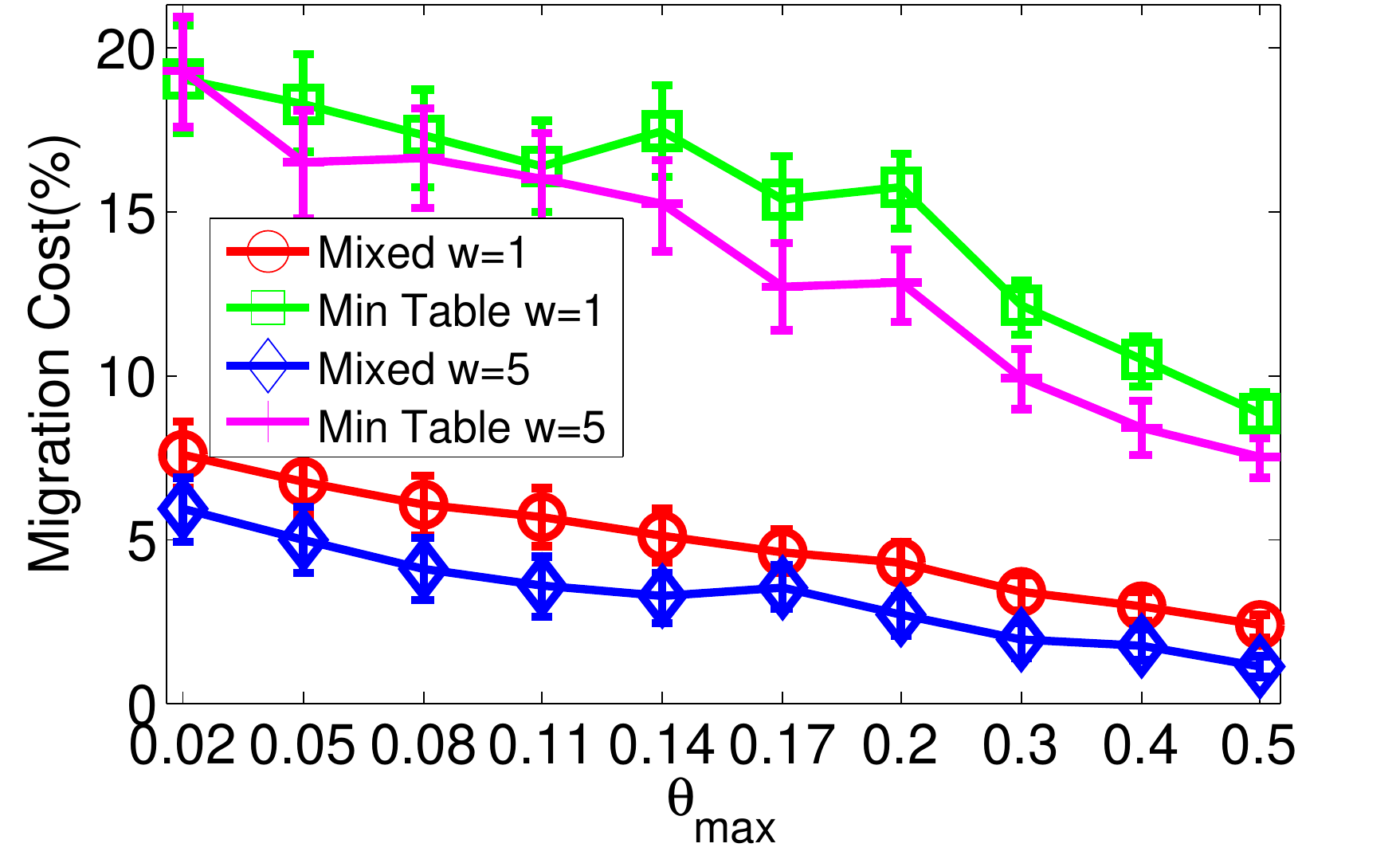}
      \label{fig:exp:thetaMC}
}
\end{tabular}
\vspace{-5pt}
 \centering \caption{Performance with varying $\theta_{\max}$.}\label{fig:exp:Theta}
 \vspace{-10pt}
 %\label{fig:Parameters}
\end{figure}

In Fig.~\ref{fig:exp:K}, we report the results on varying key domain size $K$. By varying $K$ from 5,000 to 1,000,000, \emph{Mixed} spends more computation time on migration planning but incurs less migration cost than \emph{MinTable}. As shown in Fig.~\ref{fig:exp:parameterK}, the smaller the key domain is, the more skewed the workload distribution will be. But our proposed solution \emph{Mixed} shows stable performance, regardless of the domain size, based on the results in Fig.~\ref{fig:exp:KEfficiency}. In particular, migration cost decreases for both \emph{MinTable} and \emph{Mixed} algorithms, when the window size grows to $w=5$.

\begin{figure}
\centering
\begin{tabular}[t]{c}
\hspace{-20pt}
   \subfigure[ K \emph{vs} scheduling efficiency]{
      \includegraphics[height = 3.9cm, width = 4.7cm]{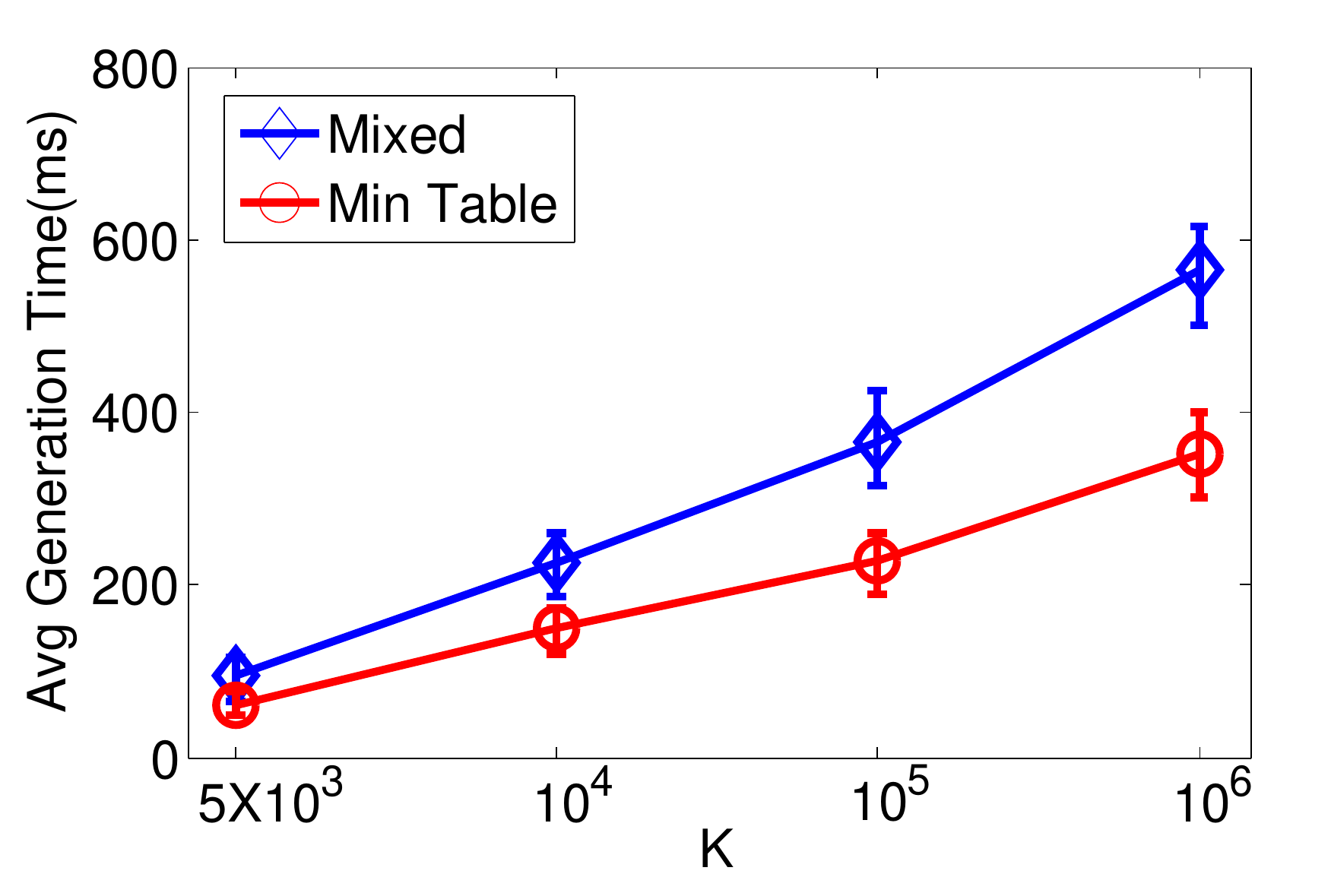}
      \label{fig:exp:KEfficiency}
   }
\hspace{-12pt}
   \subfigure[ K \emph{vs} migration cost]{
      \includegraphics[height = 3.6cm, width = 4.7cm]{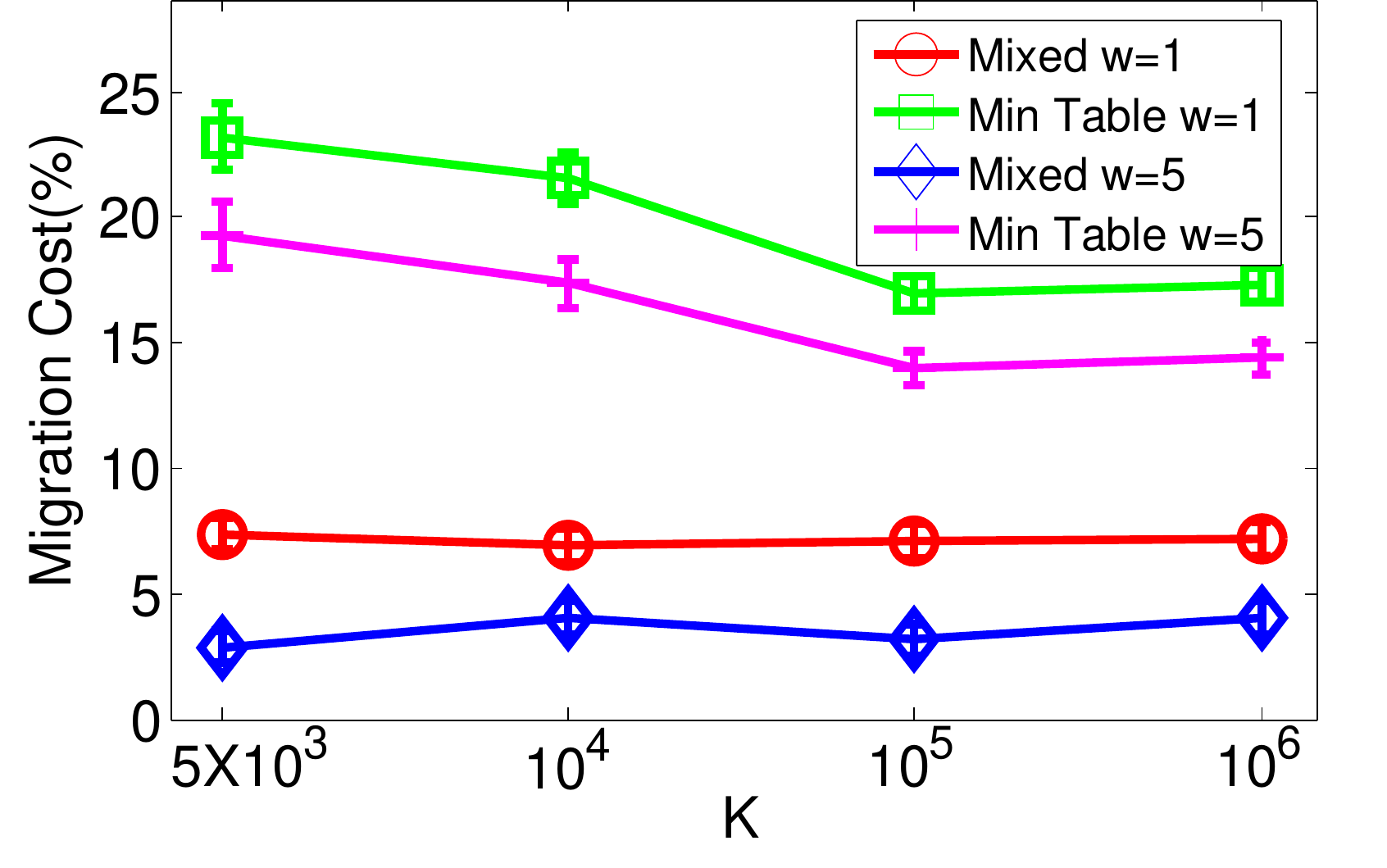}
      \label{fig:exp:KMC}
}
\end{tabular}
\vspace{-5pt}
 \centering \caption{Scheduling efficiency in terms of average generation time and migration cost under different key domain size, K}\label{fig:exp:K}
 \vspace{-10pt}
 %\label{fig:Parameters}
\end{figure}

In Sec.~\ref{sec:optimization}, we present the possibility of efficiency improvement by applying compact representation for key related information. In this group of experiments, we report the performance of this technique in Fig.~\ref{fig:exp:R}, by varying the degree of discretization (i.e., the value of $R$) on values of computation cost $v_c$ and memory consumption $v_S$. Fig.~\ref{fig:exp:REfficiency} shows that discretization on both $v_c$ and $v_S$ is an important factor to the efficiency of migration scheduling. The average generation time of migration plan is quickly reduced when we allow the system to discretize the values at a finer granularity. Note that the point with label ``Original Key Space'' is the result without applying the compact representation on keys, while the point at $R = 1$ is the case of the finest degree of discretization on $v_c$ and $v_S$. The efficiency is improved by an order of magnitude when $R = 8$ (i.e., 0.6 second) compared to ``Original Key Space'' (i.e., 6 seconds). Although larger $R$ leads to smaller $|v_c|$, $|v_S|$ and smaller $K^c$ and makes the migration plan faster, the error on load estimation grows (i.e., the percentage of divergence between actual workload of a task instance and the estimated workload based on the discretizated workload over the keys), as shown in Fig.~\ref{fig:exp:RMC}, because the discretization generates inaccurate load approximation for the keys. However, such errors are no more than 1\% in all cases, while the degree of discretiazton $R$ varies from 1 to 256 as shown in Fig.~\ref{fig:exp:RMC}.
%{\color{red}{
%For any instance, we have the absolute difference between the load generated by level partition and the load without level partition. Load Estimation Error is defined as the ratio of such different to the average load.
%For our defaulted data skewness  with $K=10^6$, the absolute average load is $3.6\cdot 10^5$ tuples.
%We show the Load Estimation Error in Fig.~\ref{fig:exp:RMC} with $\theta_{max}$ representing the number of migrating keys.
%When we set $r=7$, we find that the Load Estimation Error is as small as around $300$ tuples which is only $0.08 \%$ percentage of the average tuples.}}
\begin{figure}
\centering
\begin{tabular}[t]{c}
\hspace{-20pt}
   \subfigure[$R$ \emph{vs} scheduling efficiency]{
      \includegraphics[height = 3.6cm, width = 4.7cm]{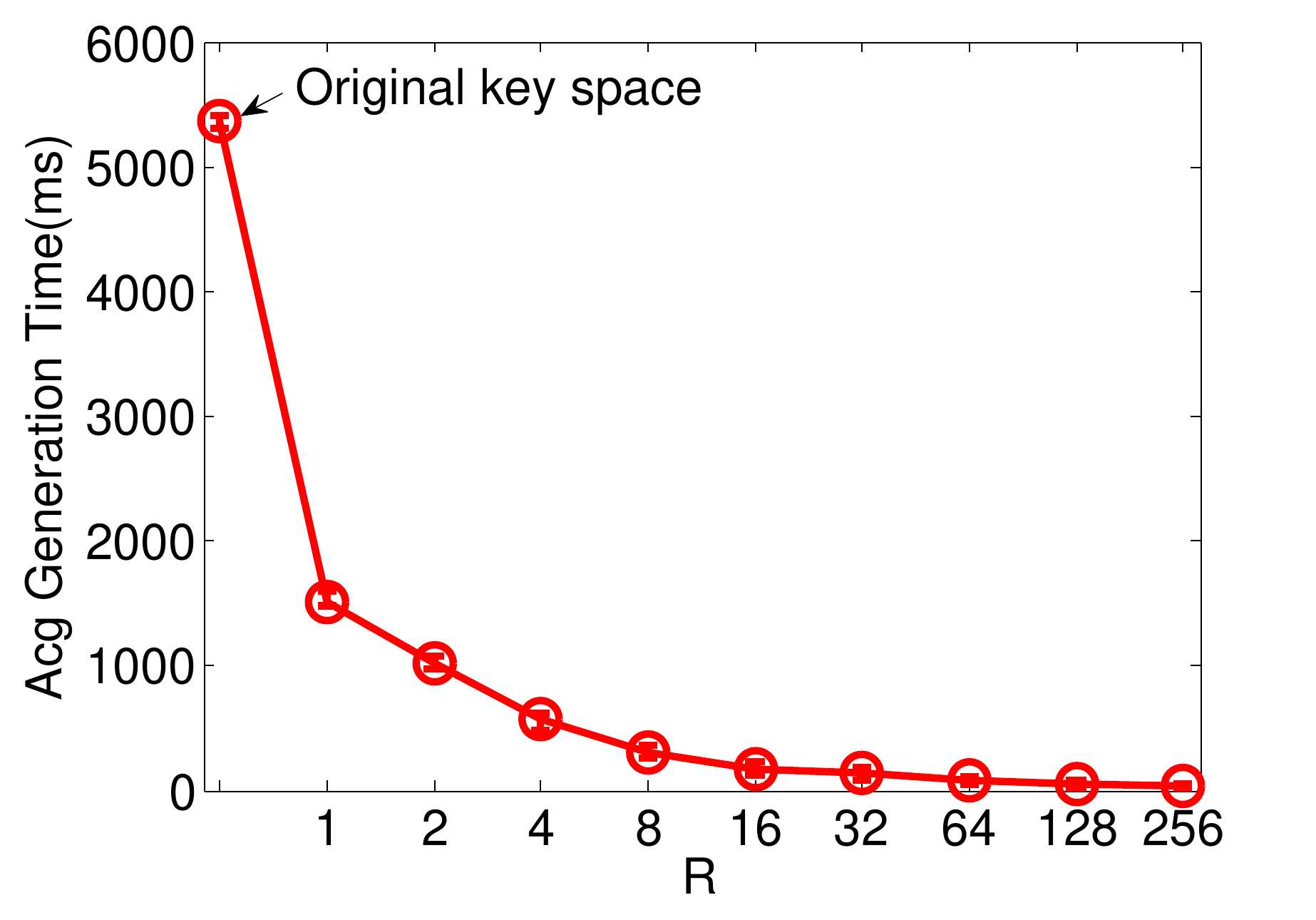}
      \label{fig:exp:REfficiency}
   }
\hspace{-12pt}
   \subfigure[$R$ \emph{vs} estimation error]{
      \includegraphics[height = 3.8cm, width = 4.7cm]{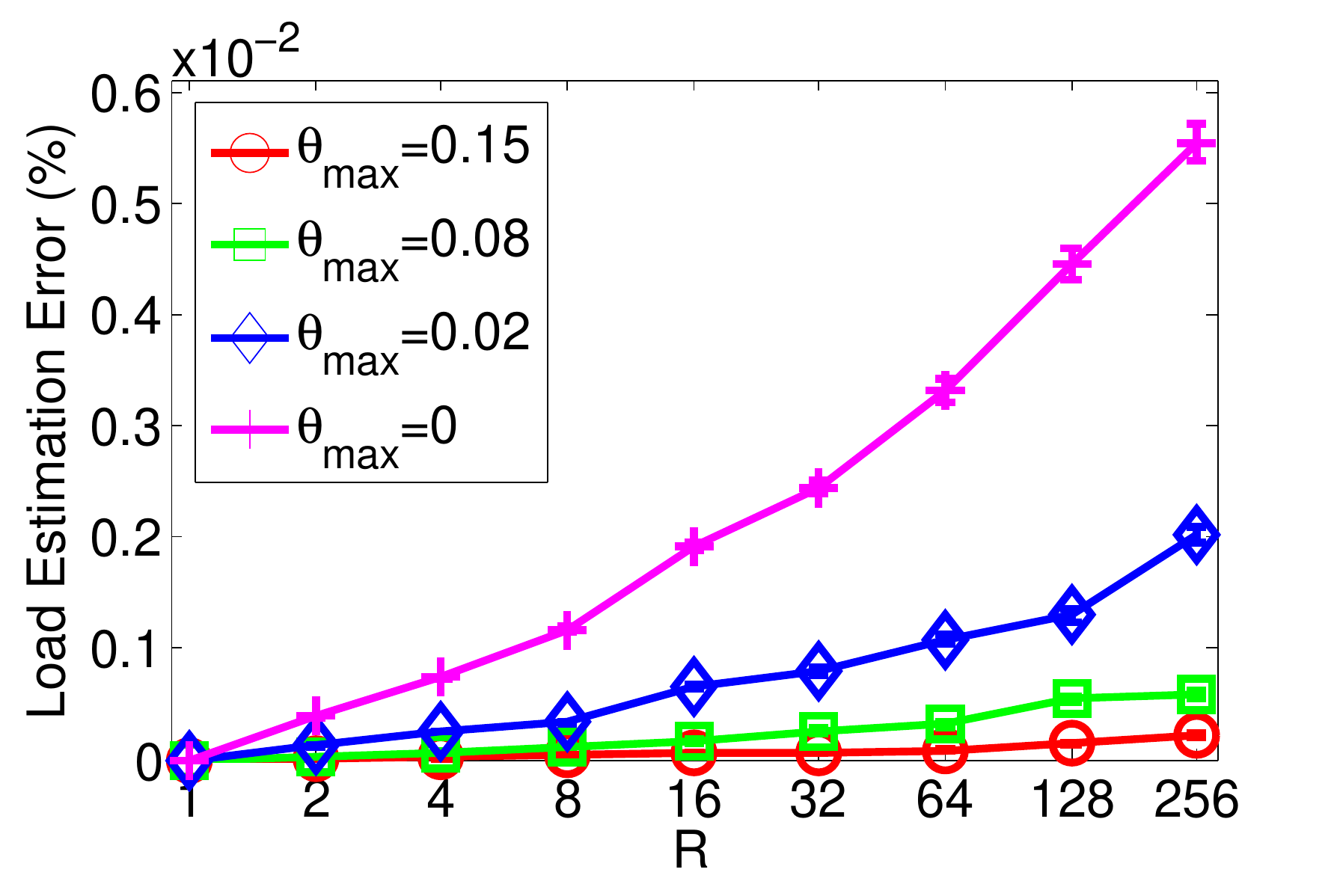}
      \label{fig:exp:RMC}
}
\end{tabular}
 \vspace{-5pt}
 \centering \caption{Performance with varied degrees of
 	discretization for partitioning granularity.}\label{fig:exp:R}
 \vspace{-10pt}
 %\label{fig:Parameters}
\end{figure}

\begin{figure}[htp]
	\vspace{-5pt}
	\centering
	\begin{tabular}[t]{c}
\hspace{-20pt}
\subfigure[Stream dynamics \textit{vs} efficiency]{
	\includegraphics[height = 3.6cm, width = 4.6cm]{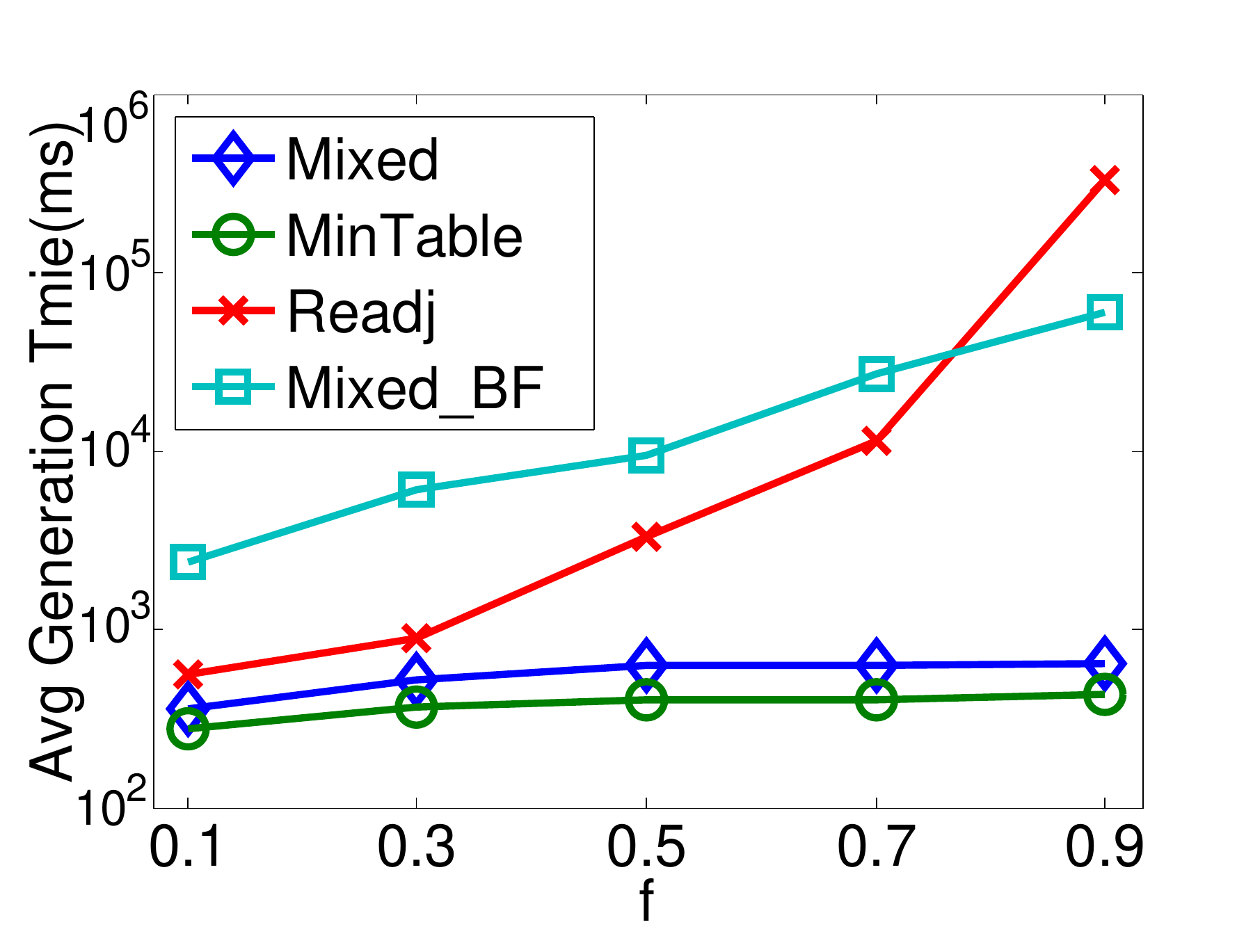}
	\label{fig:exp:CVEfficiency}
}
\hspace{-12pt}
\subfigure[Stream dynamics \textit{vs} MC]{
	\includegraphics[height = 3.4cm, width = 4.6cm]{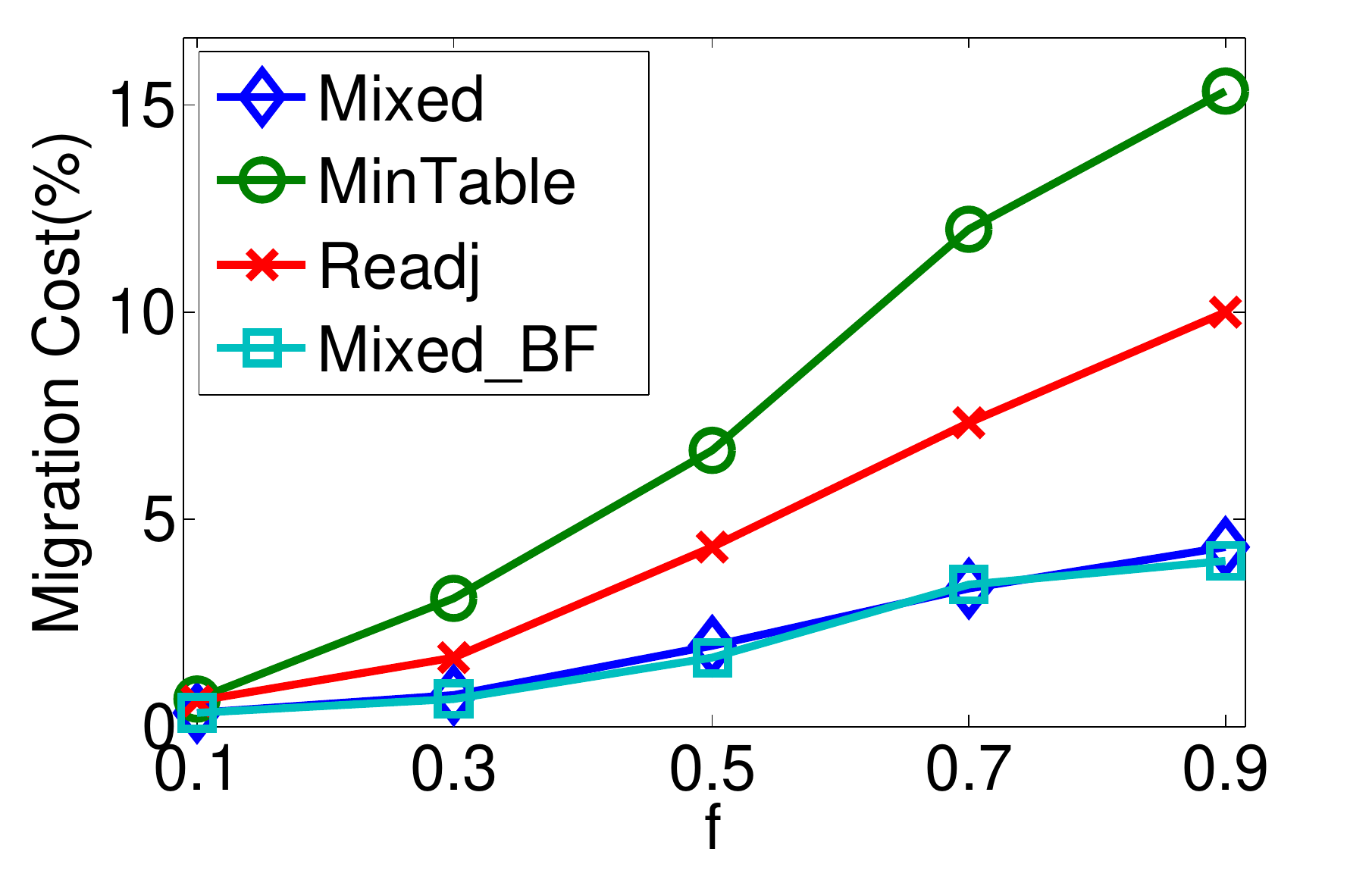}
	\label{fig:exp:CVMC}
}
	\end{tabular}%\vspace{-1em}
	\vspace{-7pt}
	\caption{Scheduling efficiency and migration cost with varying distribution change frequency.}
	\label{fig:vary_f}
	\vspace{-2pt}
\end{figure}

Since \emph{Readj} is the most similar technique to our proposal in the literature, we conduct a careful investigation on performance comparison to evaluate the effectiveness of our proposal. To optimize the performance of \emph{Readj}, we adopt binary search to find the best $\delta$ for \emph{Readj}. Fig.~\ref{fig:vary_f} shows the performance on dynamic stream processing with imbalance tolerance $\theta_{\max}=0.08$, by varying distribution change frequency $f$.  When increasing $f$, \emph{Readj} presents less promising efficiency when generating migration plan, since it evaluates every pair of task instances and considers all possible movements across the instances. Instead, \emph{Mixed} makes the migration plan based on heuristic information, which outperforms \emph{Readj} by a large margin. The results also imply that brute force search with \emph{Mixed$_{\mbox{BF}}$} is a poor option for migration scheduling. When variances occur more frequently (i.e., with a higher $f$), migration cost of \emph{Mixed} grows slower than that of \emph{Readj}, while \emph{Mixed$_{\mbox{BF}}$} performs similarly to \emph{Mixed}.

%\subsection{Performance Comparison with \textbf{Readj} on Synthetic Data}
\vspace{3pt}
\noindent\textbf{Throughput and Latency on Synthetic and Real Data :}
In Fig.~\ref{fig:readJ}, we draw the theoretical limit of the performance with the line labeled as \emph{Ideal}, which simply shuffles the workload regardless of the keys. Obviously, \emph{Ideal} always generates a better throughput and lower processing latency than any key-aware scheduling, but cannot be used in stateful operators for aggregations. When varying the distribution change frequency $f$, both the throughput and latency of \emph{Readj} change dramatically. In particular, \emph{Readj} works well only in the case with less distribution variance (smaller $f$). On the other hand, our \emph{Mixed} algorithm always performs well, with performance very close to the optimal bound set by \emph{Ideal}.

%\subsection{Performance Test On Real Data}
\begin{figure}[H]
\vspace{-10pt}
\centering
\begin{tabular}[t]{c}
\hspace{-20pt}
   \subfigure[Stream dynamics \textit{vs} throughput]{
      \includegraphics[height = 3.6cm, width = 4.6cm]{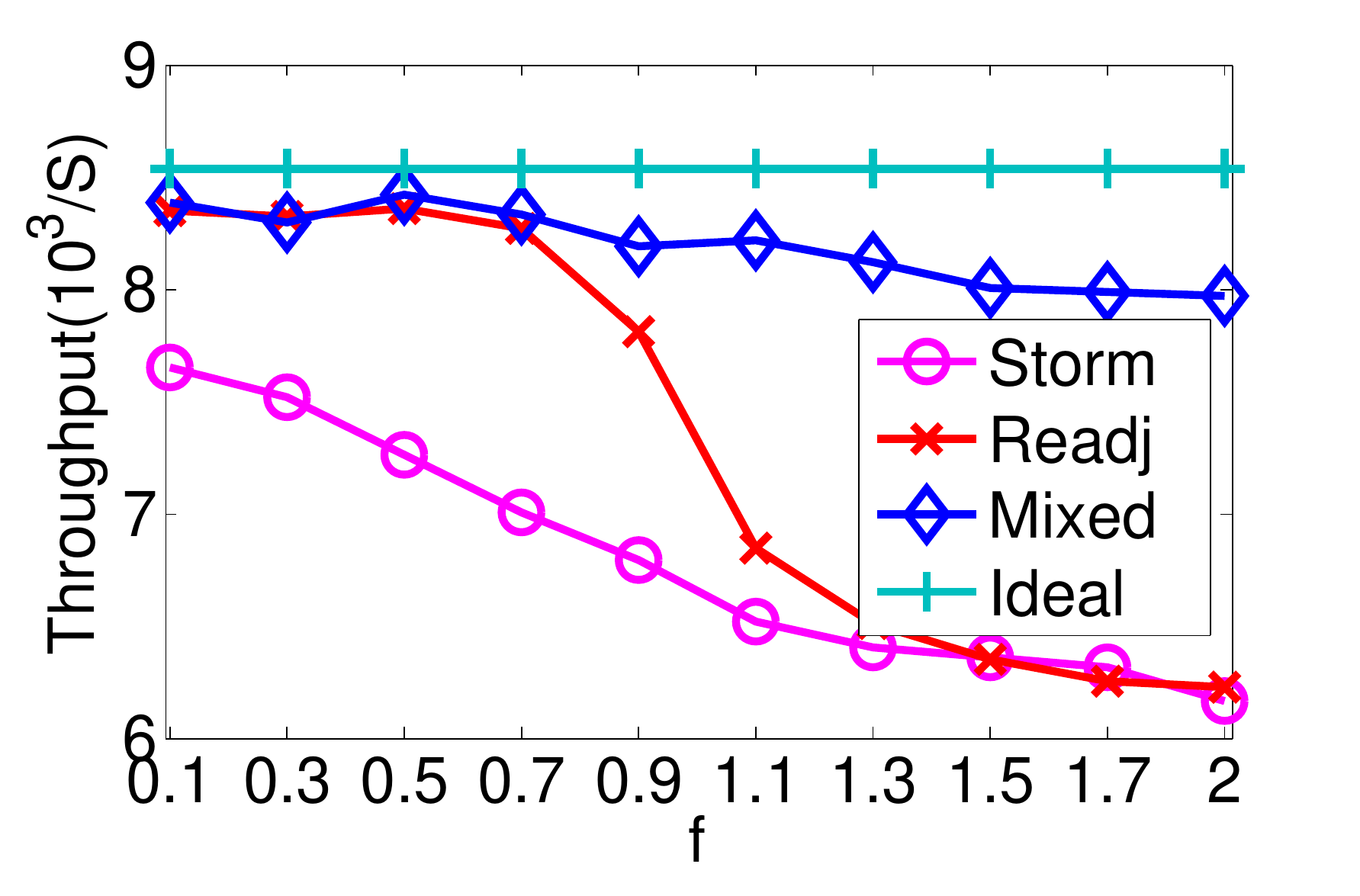}
      \label{fig:exp:stormSyn:a-err}
   }
\hspace{-12pt}
   \subfigure[Stream dynamics \textit{vs} latency]{
      \includegraphics[height = 3.6cm, width = 4.6cm]{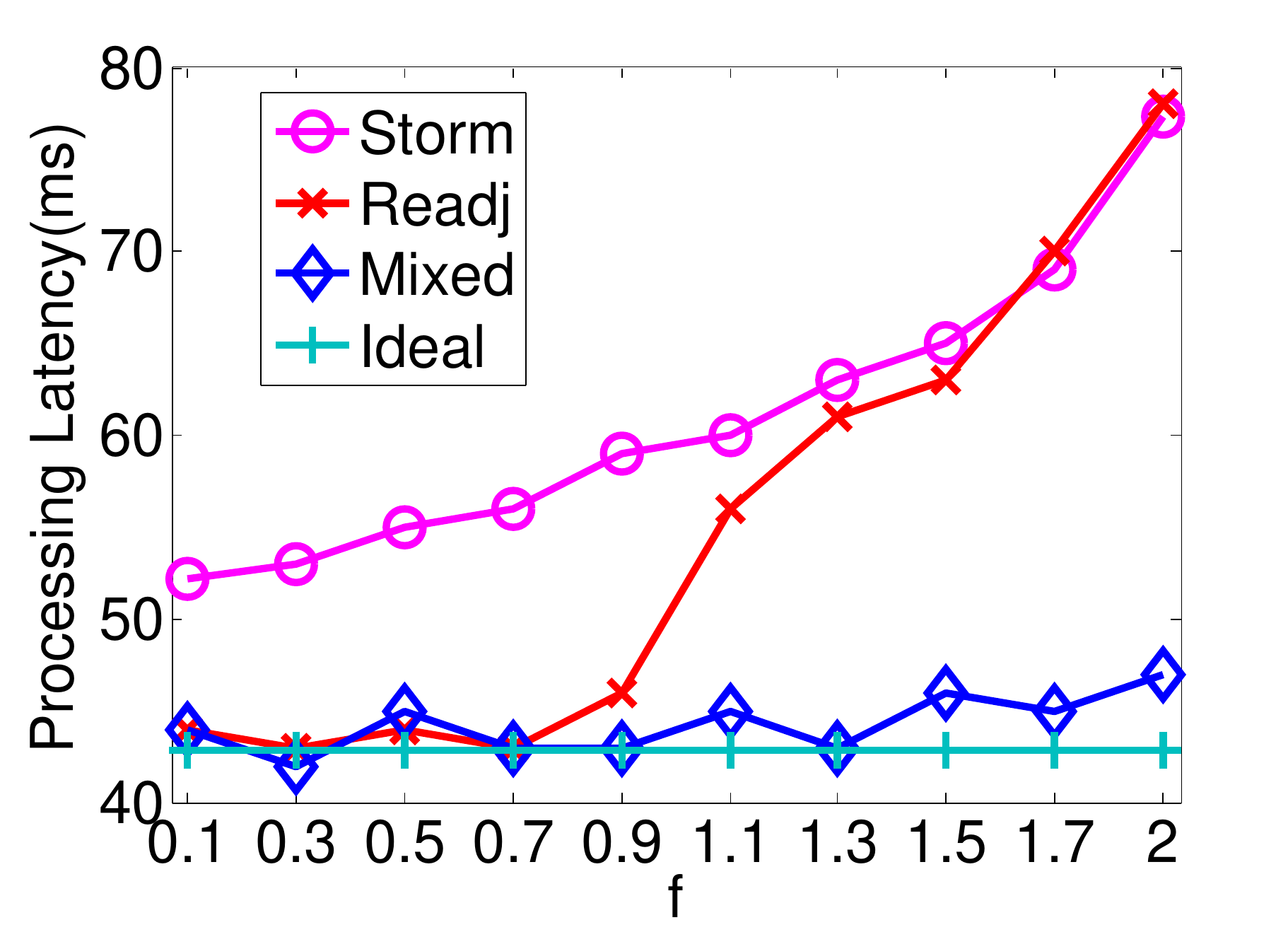}
      \label{fig:exp:stormSyn:b-err}
}
\end{tabular}%\vspace{-1em}
 \vspace{-5pt}
 \caption{Throughput and latency with varying distribution change frequency.}
 \label{fig:readJ}
 \vspace{-10pt}
\end{figure}

%\vspace{6pt}
 On \emph{Social} data, we implement a simple word count topology on Storm, with upstream instances distributing tuples to downstream instances for store and aggregation on keywords. On \emph{Stock} data, a self-join on the data over sliding window is implemented, which maintains the recent tuples based on the size of the window over intervals. The result throughputs are presented in Fig.~\ref{fig:StormReal}. The most important observation is that the best throughput, on both of the workloads, is achieved by running \emph{Mixed} with $\theta_{\max}=0.02$, implying that strict load balancing is beneficial to system performance. \emph{Mixed} also presents huge performance advantage over the other two approaches, with throughput about 2 times better than \emph{Storm} and \emph{Readj} at smaller $\theta_{\max}$ in Fig.~\ref{fig:Stockdata}. The performance of \emph{Readj} improves by relaxing the load balancing condition, catching up with the throughput of \emph{Mixed} at $\theta_{\max}=0.3$ ($\theta_{\max}=0.15$ resp.) on \emph{Social} (\emph{Stock} resp.) This is because \emph{Readj} works only when the system allows fairly imbalance among the computation tasks, for example $\theta_{\max}=0.3$.
\emph{MinTable} does not care about migration cost and then it incurs larger migration volume, which reduces the throughput of system during the process of adjustment.
\emph{PKG} splits keys into smaller granularity and distributes them to different tasks selectively. Therefore, throughput of \emph{PKG} is independent of the choice of $\theta_{\max}$, validated by the results in Fig.~\ref{fig:exp:StormRealTheta}. The throughput of \emph{PKG} is worse than \emph{Mixed}, because its processing involves coordination between two operators. Despite of its excellent performance on load balancing, the overhead of partial result merging leads to additional response time increase and overall processing throughput reduction. Overall, as shown in Fig.~\ref{fig:exp:StormRealTheta}, when $\theta_{max}=0.02$, our method outperform \emph{PKG} on throughput by 10\% and on response latency by 40\%. Moreover, we emphasize that \emph{PKG} cannot be used for complex processing logics, such as join, and therefore is not universally applicable to all stream processing jobs.

%In Fig.~\ref{fig:exp:StormRealTheta} and Fig.~\ref{fig:Stockdata}, \emph{Readj} achieves similar performance as \emph{Mixed}, only when we relax the imbalance tolerance significantly (e.g., $\theta_{\max}=0.3$). \emph{Mixed} beats \emph{Readj} on system performance by a large margin, when the system is expected to have better load balance (e.g., $\theta_{\max} =0.08$). \textbf{ZZJ: The rest of the discussion is boring. Unnecessary to discuss the difference between the workload. Focus on the approaches! Talk about how much we improve over Storm!} Obviously, the workload of \emph{Social} is much lower than that of \emph{Stock}, the processing throughput of \emph{Social} is much higher than that of \emph{Stock} in Fig.~\ref{fig:exp:StormRealTheta} and Fig.~\ref{fig:Stockdata}. The keys in stock application are much fewer than the keys in \emph{social media} data, \emph{Readj} handles \emph{Stock} data much easier and achieves almost the same performance as \emph{Mixed} algorithm even for  $\theta_{\max}=0.15$ as shown in Fig.~\ref{fig:Stockdata}.  Obviously, for both \emph{aggregation} and \emph{join} operations on real datasets, \emph{Mixed} is the best, especially when stronger balancing constraint is required.

%In Fig.\ref{fig:StormReal}, we use $3$ days of \emph{social media} dataset ($3.2\cdot 10^6$ tuples for $8.5\cdot 10^4$ topics) and $1$ day of stock data ($2\cdot 10^6$ tuples for $1036$ keys), and use the whole dataset for scale out test in Fig.~\ref{fig:scaleout}.
\begin{figure}[htp]
%\vspace{-10pt}
\centering
\begin{tabular}[t]{c}
\hspace{-20pt}
   \subfigure[\emph{Social} Data]{
      \includegraphics[height = 3.2cm, width = 4.6cm]{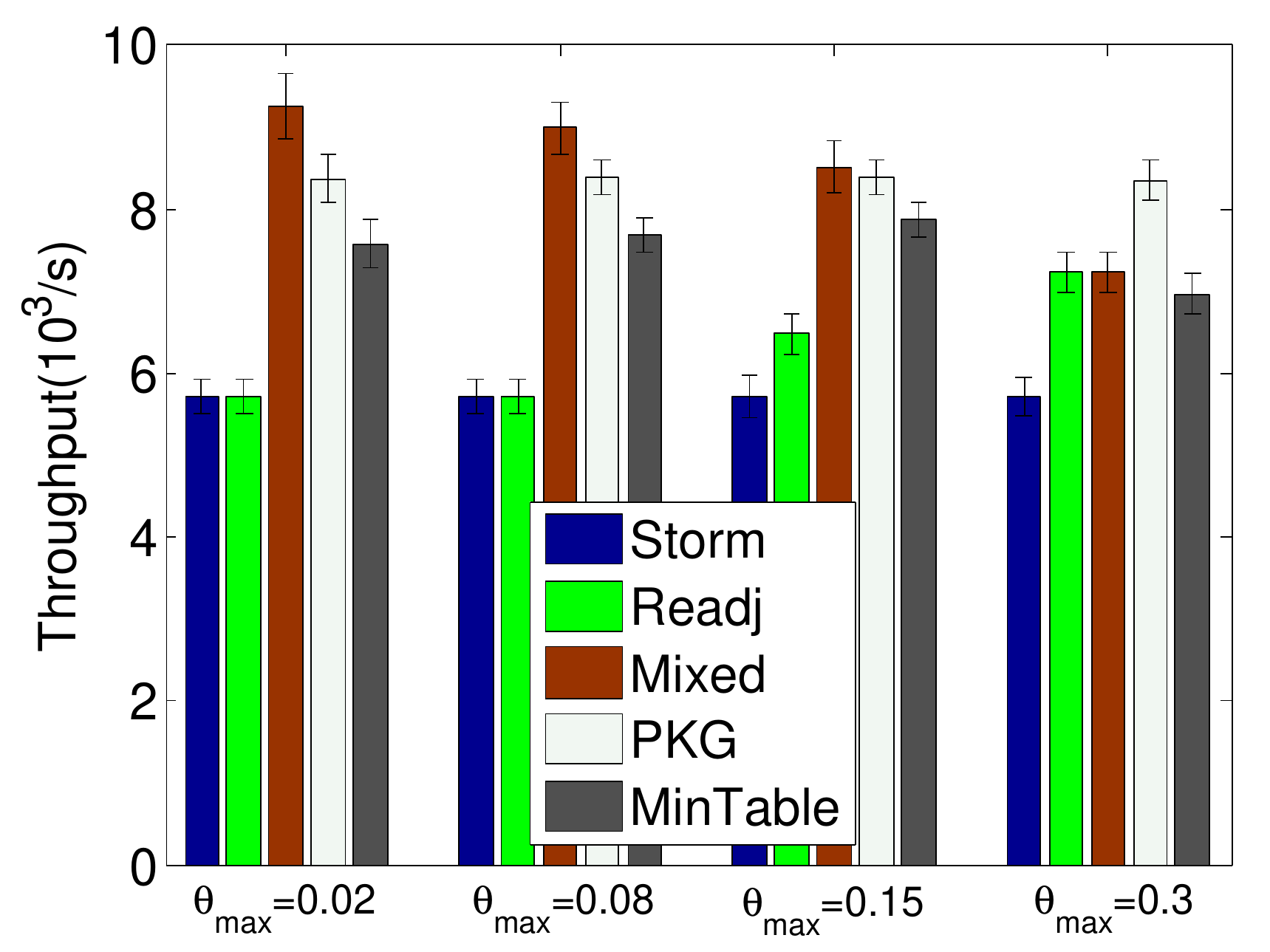}
      \label{fig:exp:StormRealTheta}
}
\hspace{-12pt}
\subfigure[\emph{Stock} Data]{
\includegraphics[height = 3.2cm, width = 4.6cm]{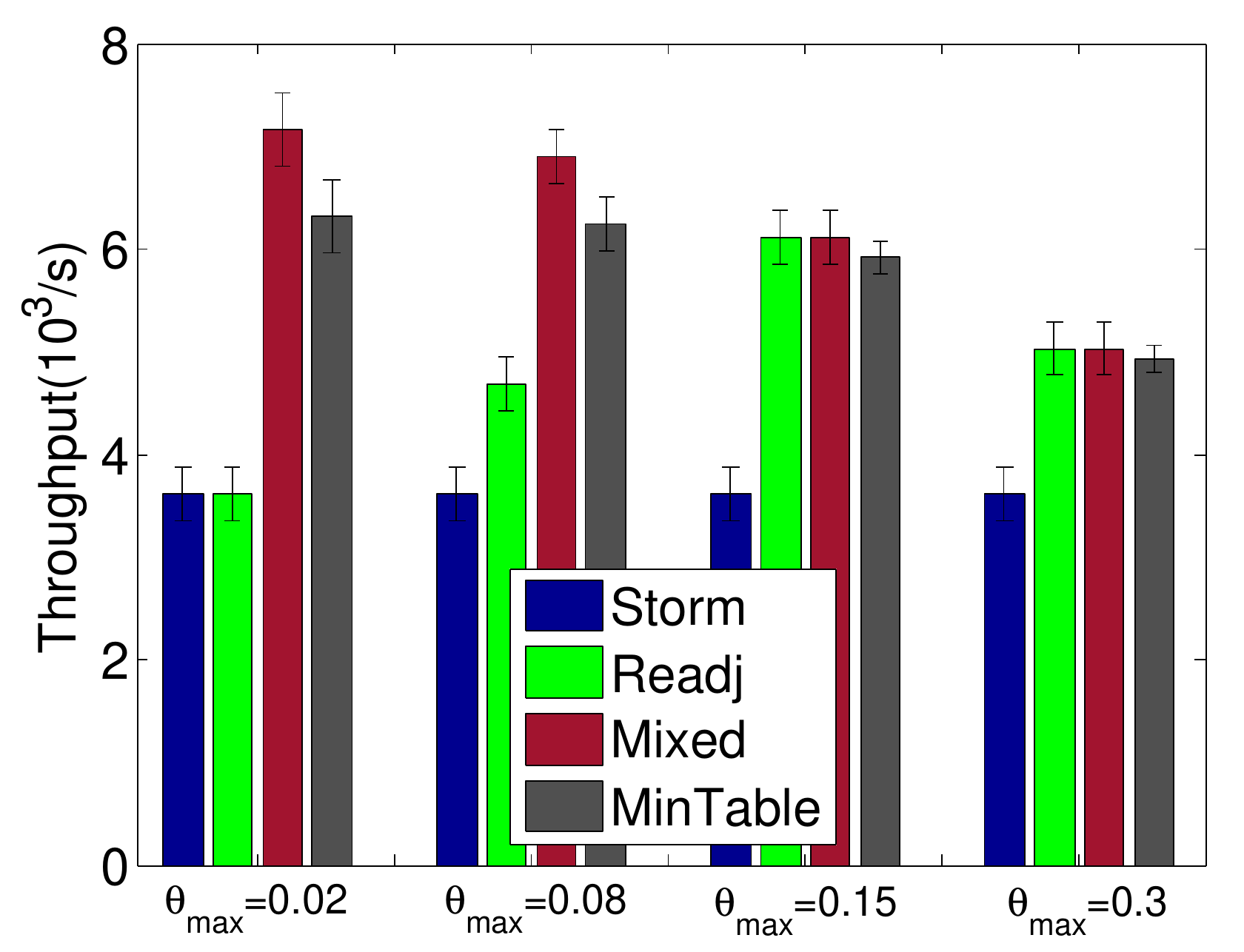}
%\caption{Stock Data}
\label{fig:Stockdata}
 }
\end{tabular}%\vspace{-1em}
 \vspace{-7pt}
 \caption{Throughput on real data.}
 \label{fig:StormReal}
 \vspace{-10pt}
\end{figure}

\vspace{3pt}
\noindent\textbf{Scalability on Real Data: }
To better understand the performance of the approaches in action, we present the dynamics of the throughput over time on two real workloads, especially when the system scales out the resource by adding new computation resource to the operator. The results are available in Fig.~\ref{fig:scaleout}. In order to test this kind of scale-out ability of different algorithms, we run the stream system to a balance status, and then add one more working thread (instance) to the system  starting the balance processing algorithms. The results show that our method \emph{Mixed} perfectly rebalances the system within a much shorter response time than that of \emph{Readj}. Though \emph{PKG} is $\theta_{max}$ insensitive, it produces a lower throughput than \emph{Mixed} while $\theta_{max}=0.1$. As the explanation of Fig.~\ref{fig:exp:StormRealTheta}, \emph{PKG} needs to keep track of all the derived data from a spout until it receives ack response and this action exacerbates its processing latency. On \emph{Social\ data} with $\theta_{max}=0.10$, \emph{Readj} takes at least 5 minutes to generate the migration plan for the new thread added to the system. Such a delay leads to huge resource waste, which is definitely undesirable to cloud-based streaming processing applications.  Similar results are also observed on \emph{Stock}. The quick response of \emph{Mixed} makes it a much better option for real systems.

%\textbf{ZZJ: The key observation is the delay, not the throughput? Again, no need to discuss the difference between workload. }At smaller $\theta$, both $Mixed$ and $Readj$ could reach almost the same throughput after a while. The relaxing to balance status causes the whole \emph{throughput} decreasing. In both Fig.~\ref{fig:Social} and Fig.~\ref{fig:Stock}, \emph{Mixed} demonstrates the best scalability for its fast load adjustment ability to reach higher throughput.  For $Readj$, even for relaxed balance  requirement, e.g. $\theta_{max}=0.2$, it will spend three times of time for load balance compared to $Mixed$. Since \emph{Stock} data has less keys for computation, it spends about two times of time for balancing load with $\theta_{max}=0.2$ compared to $Mixed$ shown in Fig.~\ref{fig:Stock}.

\begin{figure}[htp]
\centering
\vspace{-5pt}
\begin{tabular}[t]{c}
\hspace{-20pt}
\subfigure[Social data]{
\includegraphics[height = 3.2cm, width = 4.6cm]{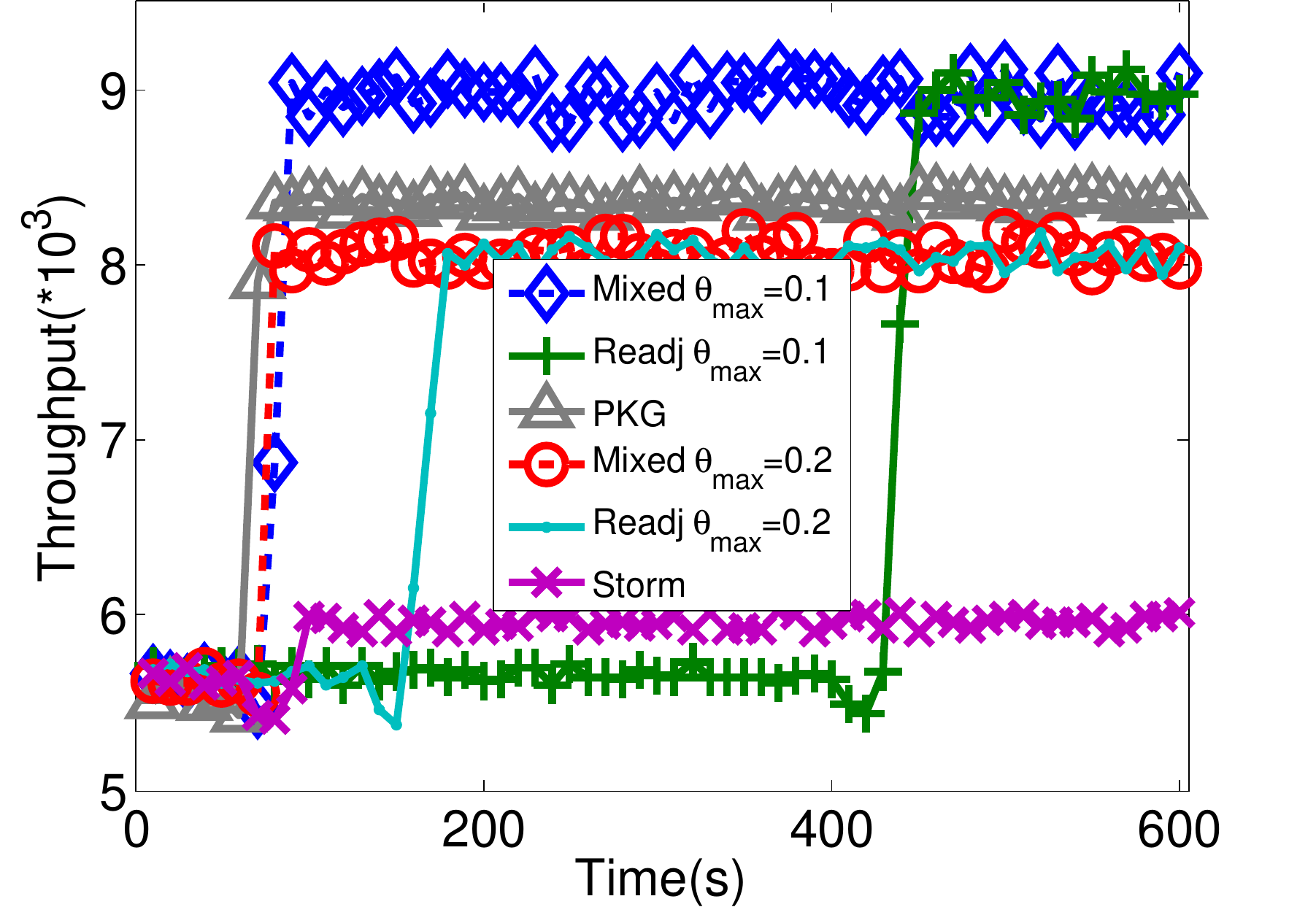}
%\caption{Stock Data}
\label{fig:Social}
 }
 \hspace{-12pt}
\subfigure[Stock data]{
\includegraphics[height = 3.2cm, width = 4.6cm]{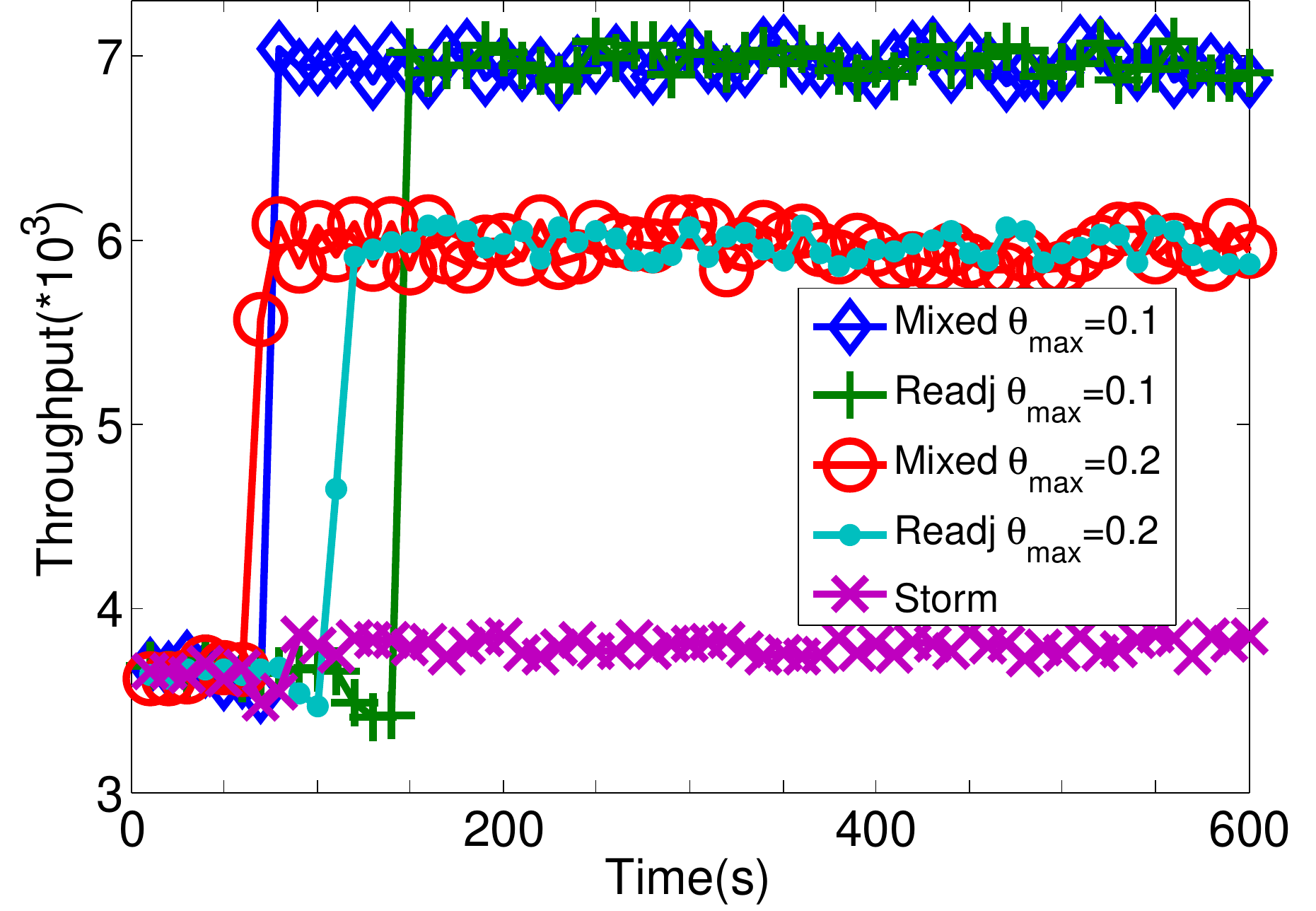}
%\caption{Stock Data}
\label{fig:Stock}
 }
\end{tabular}
\vspace{-8pt}
 \caption{Performance during system scale-out.}
 \label{fig:scaleout}
 \vspace{-10pt}
\end{figure}

%\begin{figure}[htp]
%\centering
%\begin{tabular}[t]{c}
%\hspace{-20pt}
%\subfigure[$\theta_{max}=0.1$]{
%\includegraphics[height = 3.2cm, width = 4.6cm]{expNew/socialscale_ten}
%%\caption{Stock Data}
%\label{fig:socialThetatwo}
% }
% \hspace{-12pt}
%\subfigure[$\theta_{max}=0.2$]{
%\includegraphics[height = 3.2cm, width = 4.6cm]{expNew/socialscale_twenty}
%%\caption{Stock Data}
%\label{fig:socialThetatwo}
% }
%\end{tabular}\vspace{-1em}
% \caption{Scale out on social media data}
% \label{fig:socialscaleout}
%\end{figure}
%
%
%\begin{figure}[htp]
%\centering
%\begin{tabular}[t]{c}
%\hspace{-20pt}
%\subfigure[$\theta_{max}=0.1$]{
%\includegraphics[height = 3.2cm, width = 4.6cm]{expNew/stocklscale_ten}
%%\caption{Stock Data}
%\label{fig:stockThetaten}
% }
% \hspace{-12pt}
%\subfigure[$\theta_{max}=0.2$]{
%\includegraphics[height = 3.2cm, width = 4.6cm]{expNew/stocklscale_twenty}
%%\caption{Stock Data}
%\label{fig:stockThetatwenty}
% }
%\end{tabular}\vspace{-1em}
% \caption{Scale out on stock data}
% \label{fig:stockscaleout}
%\end{figure}

\begin{figure}[htp]
	\centering
	\begin{tabular}[t]{c}
		\hspace{-20pt}
		\subfigure[$\theta_{max}=0.1$]{
			\includegraphics[height = 3.2cm, width = 4.6cm]{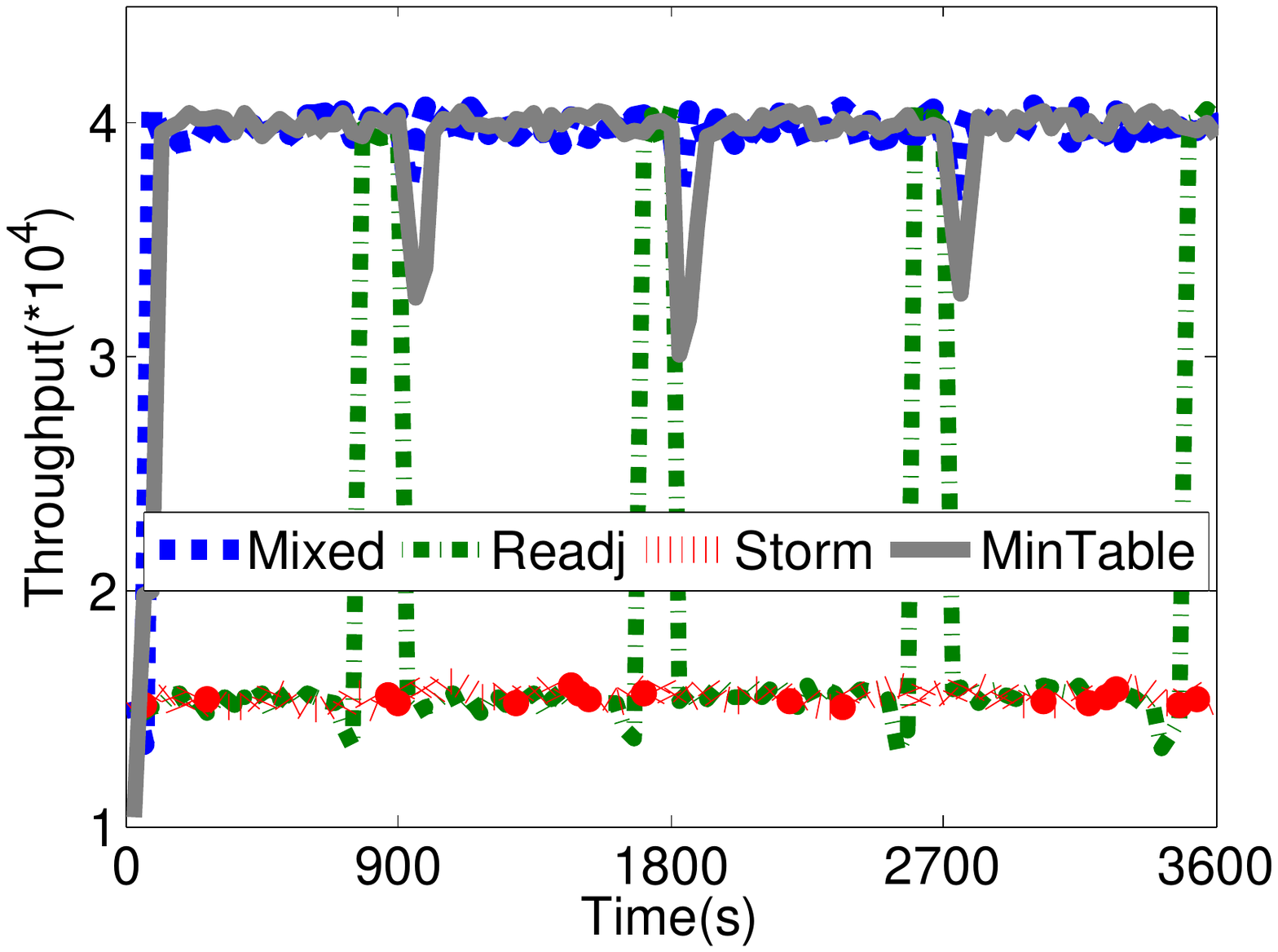}
			%\caption{Stock Data}
			\label{fig:tPChThetaten}
		}
		\hspace{-12pt}
		\subfigure[$\theta_{max}=0.2$]{
			\includegraphics[height = 3.2cm, width = 4.6cm]{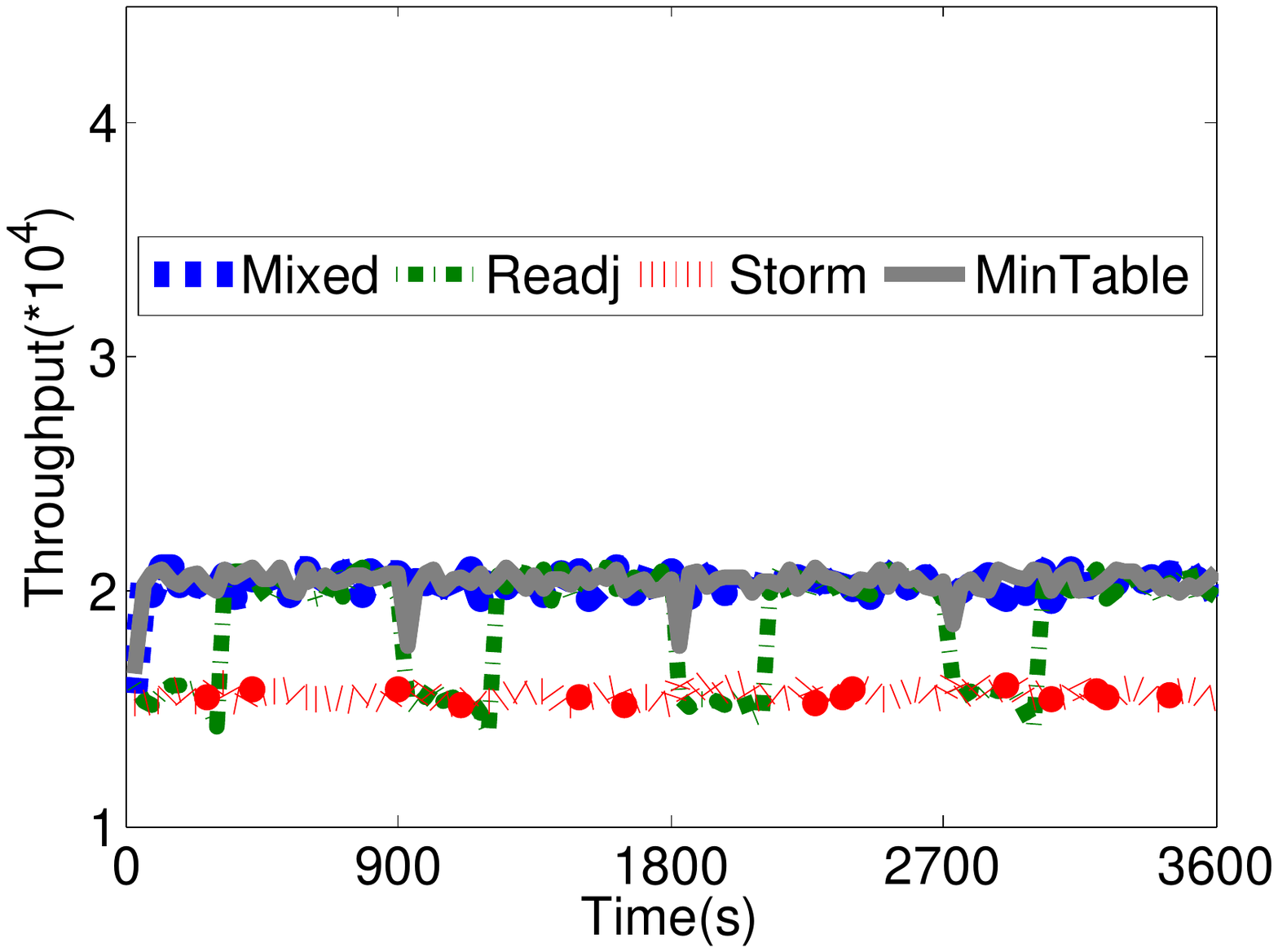}
			%\caption{Stock Data}
			\label{fig:tpchThetatwenty}
		}
	\end{tabular}
	\vspace{-8pt}
	\caption{Dynamic adjustment on TPC-H data for $Q5$.}
	\label{fig:tcphchangeload}
	\vspace{-10pt}
\end{figure}
\vspace{3pt}
\noindent\textbf{Dynamics on TPC-H for $Q5$:}
We use $DBGen$
\cite{url:TPCH} to generate $1\ GB$ TPC-H dataset by producing zipf skewness on foreign keys with $z=0.8$. We run $Q5$ on the generated dataset for one hour and set window size as 5 minutes, since the join operations in $Q5$ are implemented by different processing operators. The data imbalance slows down the  previous join operator  (upstream instances) and  suspends the processing on downstream join operators. This bad consequence of such suspension may be amplified with the growing number of task instances. In particular we test the effects by triggering the distribution change in every 15 minutes with $f=1$. The results are shown in Fig.~\ref{fig:tcphchangeload}. Without any balancing strategy, \emph{Storm} presents poor throughputs. \emph{Mixed} is capable of balancing the workload in an efficient manner and achieving the best throughput under any balancing tolerance.  
%\noindent\textbf{TCPH Data:}

%% file: relatedwork.tex
\section{Related Work}\label{sec:relatedwork}

% overview

%Load balancing is a general problem in distributed systems, which is applied in distributed systems to maximize the utilization of computation resource by evenly assigning the tasks to the distributed nodes. 

Different from batch processing and traditional distributed database \cite{dewitt1992parallel, walton1991taxonomy, xu2008handling, kwon2012skewtune, gufler2012load}, the problem of load balancing is more challenging on distributed stream processing systems, because of the needs of continuous optimization and difficulty with high dynamism.
%
%It has no information about the subsequent stream data, and then we can not know what will happen during the next time unit which may seriously affect our balance mechanism, or else, it may slow down the overall performance, e.g. low throughput or high latency.
%
There are two common classes of strategies to enable load balancing in distributed stream processing systems, namely \textit{operator-based} and \textit{data-based}.

%Operator-level will migrate the whole operator from one node to the other; Data-level will group data and do migration  based on these groups. Specifically, some operations must use the historic information together with current data, which is called stateful migration.

\vspace{6pt}
\noindent\textbf{Operator-based }strategies generally assume the basic computation units are operators. Therefore, load balancing among distributed nodes is achieved by allocating the operators to the nodes. In Borealis~\cite{xing2005dynamic}, for example, the system exploits the correlation and variance of the workloads of the operators, to make more reliable and stable assignments. In~\cite{xing2006providing}, Xing et al. observe that operator movement is too expensive for short-term workload bursts. This observation motivates them to design a new load balance model and corresponding algorithms to support more resilient operator placement.
System S ~\cite{wolf2008soda}, as another example, also generates scheduling decisions for jobs in submission phase and migrates jobs or sub-jobs to less loaded machines on runtime based on complex statistics, including operators workload and the priority of the applications.
%In ~\cite{khandekar2009cola}, it combines multiple operators to reduce communication.
Zhou et al. \cite{ZOT05} presents a flow-aware load selection strategy to minimize communication cost with their new dynamic assignment strategy adaptive to the evolving stream workloads.
In order to improve system balance property, \cite{aniello2013adaptive} presents a more flexible mechanism by using both online and offline methods under the objective of network traffic minimization.
A common problem with operator-based load balancing is the lack of flexible workload partitioning. It could lead to difficulty to any operator-based load balancing technique, when an operator is much more overloaded than all other operators. 

%Fine granularity will be essential for improving system scalability.  Recently, data-level balance processing comes to work, such as ~ \cite{lin2015scalable,gedik2014partitioning,elseidy2014scalable,nasir2015power,rivetti2015efficient,shah2003flux}.

\vspace{6pt}
\noindent\textbf{Data-based }strategies allow the system to repartition the workload based on keys of the tuples in the stream, motivated by the huge success of MapReduce system and its variants.

% elasticity

Elastic stream processing is a hot topic in both database and distributed system communities. Such systems attempt to scale out the computation parallelism to address the increasing computation workload, e.g., \cite{GSH+14,WT15}. By applying queuing theory, it is possible to model the workload and expected processing latency, which can be used for better resource scheduling \cite{fu2015drs}. When historical records are available to the system, it is beneficial to generate a long-term workload evolution plan, to schedule the migrations in the future with smaller workload movement overhead \cite{DFM+15}. Note that all these systems and algorithms are designed to handle long-term workload variance. All these solutions are generally too expensive if the workload fluctuation is just a short-term phenomenon. The proposal in this work targets to solve the short-term workload variance problem with minimal cost.

% stream join

A number of research work focus on load balancing in distributed stream join systems.
\cite{elseidy2014scalable} models the join operation on a square matrix, each side of which represents one join stream. It distributes tuples randomly to cells on each line (or column) in matrix, which produces a potential join output.
To avoid these problems, in~\cite{lin2015scalable}, it proposes a join-biclique model which organizes the clusters as a complete bipartite graph for joining big data streams. In ~\cite{bruno2014advanced}, it classifies skewness on join key granularities and divides the joins into three types, e.g., B-Skew join, E-Skew join and H-Skew join. It proposes to deal with load imbalance by using different join algorithms. In ~\cite{taft2014store}, it designs a two-tiered partitioning method which handles hot tuples separately from cold tuples and distributes keys with heavy keys (big granularities) first. DKG~\cite{rivetti2015efficient} also distinguishes heavy keys from light ones by granularities and applies greedy algorithms for load balance. Although these techniques are effective for stream job applications, they are not directly extensible to general-purpose stream processing.

% flux

Flux \cite{shah2003flux} is the widely adopted load balancing strategy, designed for traditional distributed streaming processing systems.   It simply measures the workload of the tasks, and attempts to migrate workload from overloaded nodes to underloaded nodes. One key limitation of Flux is the lack of consideration on the routing overhead. In traditional stream processing systems, the workload of a logical operator is pre-partitioned into tasks, such that each task may handle a huge number of keys but processed by an individual thread at any time. The approach proposed in this paper allows the system to reassign keys in a much more flexible manner. Our approach also takes routing overhead into consideration, because it is unrealistic to fully control the destination for all keys from a large domain.

% two-choice
Nasir et al. ~\cite{nasir2015power, nasir2015two} design a series of randomized routing algorithms to balance the workload of stream processing operators. Their strategy is based on the theoretical model called power-of-two, which evaluates two randomly chosen candidate destinations for each tuples and chooses the one with smaller workload estimation. Their approach is more appropriate for stateless operators in streaming processing, and a subset of stateful operators by introducing an aggregator to combine results of tuples sent to different working threads. A number of stateful operators, such as join, may need all historical tuples with certain keys in order to generate complete and accurate computation results, which cannot be supported by such scheme. The method proposed in this paper, however, does not have such limitation, thus applicable to any stateful operator with perfect load balancing performance.

% readj

Readj~\cite{gedik2014partitioning} is proposed to resolve the stateful load balance problem with a small routing table, which is the most similar work to our proposal. It introduces a similar tuple distribution function, consisting of a basic hash function and an explicit hash table.
%
%The explicit hash table maintains the keys with high or low frequencies. When it meets with imbalance, each task will first move back the keys with small granularities that are mapped to it by the explicit hash function.
%
However, the workload redistribution mechanism used in Readj is completely different from ours. The algorithm in Readj always tries to move back the keys to their original destination by hash function, followed with migration schedules on keys with relatively larger workload. Their strategy might work well when the workload of the keys are almost uniform. When the workloads of the keys vary dramatically, their approach either fails to find a reasonable load balancing plan, or incurs huge routing overhead by generating a large routing table. The routing algorithms designed in this paper completely tackle this problem, which presents high efficiency as well as good balancing performance in almost all circumstances.

%% file: conclusion.tex
\section{Conclusion and Future Work}\label{sec:conclusion}
%As the increasing of data volume and the change of data distribution, adaptive load balance  algorithm on keys for dynamic stream is important for improving the performance of distributed stream process system.
%In this paper, we present a general framework to support dynamic workload assignment for stateful operators.
%In order to balance loads with relative small migration cost, we formulate it as an optimization problem and propose an adaptive balance adjustment algorithm \textbf{MMC-LATS}. We order the keys by loads and then segment those keys into levels by considering the distance compensation for reducing balance deviation. Through both theoretic analysis and extensive experimental verification, it is proved that our algorithm can achieve load balance with the requirement of limited routing table size efficiently and effectively.

This paper presents a new dynamic workload distribution mechanism for intra-operator load balancing in distributed stream processing engines. Our mixed distribution strategy is capable of assigning the workload evenly over task workers of an operator, under short-term workload fluctuations. New optimization techniques are introduced to improve the efficiency of the approach, to enable practical implementation over mainstream stream processing engines. Our testings on Apache Storm platform show excellent performance improvement with a variety of workload from real applications, also present huge advantage over existing solutions on both system throughput and response latency. In the future, we will investigate the theoretical properties of the algorithms to better understand the optimality of the approaches under general assumptions. We will also try to design a new mechanism, to support smooth workload redistribution suitable to both long-term workload shifts and short-term workload fluctuations.

%% file: appendix-upboundproof.tex
%%overview
\section{Proofs of Theorem 1}\label{appsec:proof}
In order to derive theoretic results about the LLFD algorithm, we first look at a more simplified key assignment algorithm, namely the Simple algorithm. We next derive a serious of theoretic results based on the Simple algorithm. Lastly, we show how these results are applicable to the LLFD algorithm.
\begin{algorithm}[t]
	%\small
	\caption{Simple Algorithm}
	\label{alg:simple}
	\begin{algorithmic}[1]
		\Require task instances in $\mathcal{D}$ %key candidate set $\mathcal{C}$, 
		%\Ensure instances set $\mathcal{D}$ after loading keys
		\Ensure $A'$
		\State Disassociate keys from all the instances		
		\ForAll{$d$ in $\mathcal{D}$}
		\State $L(d)\leftarrow 0$
		\EndFor
		\State Add all keys to $\mathcal{C}$, i.e., $\mathcal{C}\leftarrow\mathcal{K}$
		\ForAll{$k$ in $\mathcal{C}$ in descending order of $c(k)$}
		\ForAll{$d$ in $\mathcal{D}$ in ascending order of $L(d)$}
		\State Associate key $k$ with instance $d$, i.e.:
		\State L(d) = L(d) + c(k)
		\If {$h(k) \neq d$}
		\State Add entry $(k, d)$ to $A'$
		\EndIf
		\State Remove $k$ from $\mathcal{C}$; 
		\EndFor
		\EndFor
		\State \Return $A'$
	\end{algorithmic}
\end{algorithm}
As described in Algorithm~\ref{alg:simple}, the \emph{Simple} algorithm works in the following way, at first, it disassociates and puts all the keys into the candidate set $\mathcal{C}$ (Lines 1--4). Secondly it sorts these keys in a descending order of the computation cost $c(k)$. Finally it sequentially assigns each key to the instance with the least total workloads so far(Line 5--11). 

\begin{definition}\label{def:perfect_assignment}
	 A perfect assignment is defined as the approach that can assign keys to instances resulting in  $\forall d_i, d_j\in\mathcal{D}, L(d_i) = L(d_j) = \bar{L} = \frac{1}{N_D}\sum_{k\in\mathcal{K}} c(k)$.
\end{definition}

\begin{lemma}\label{lemma1}
	Given the instance set $\mathcal{D}$ of size $N_D$, key set $\mathcal{K}$ of size $K$ and computation cost of each key $c(k)$, where keys are in a non-increasing order of their computation costs, i.e., $c(k_1) \geq c(k_2) \geq \dots \geq c(K)$, if the perfect assignment exists, we have: 
	\begin{equation}\label{equ:lemma}
	c(k_{qN_D + 1}) \leq \frac{1}{q+1}\bar{L},\;\;q = 1, 2, \dots, \lfloor \frac{K-1}{N_D}\rfloor.
	\end{equation}
\end {lemma}
\begin{proof}
Assuming $c(k_{qN_D + 1}) > \frac{1}{q+1}\bar{L}$, then we have $c(k_1) \geq c(k_2) \geq \dots \geq c(k_{qN_D + 1}) > \frac{1}{q+1}\bar{L}$. This means that for keys from $k_1$ to $k_{qN_D}$, each instance can at most be associate with $q$ of them. In result, any instance that is associated with the $(qN_D + 1)$-th key will generate workloads larger than $\bar{L}$, which contradicts the assumption of the existence of the perfect assignment. 
%According to the pigeonhole principle, there must be one instance associated with two keys both of which have workloads greater than $\bar{L}/2$, then the workload of this instance must be larger than $\bar{L}$, which contradicts the assumption of the existence of the perfect assignment.
%We skip the proofs for the remaining cases since the approach is quite similar.
\end{proof}

\begin{lemma}\label{lemma2}
	Given the instance set $\mathcal{D}$ of size $N_D$, key set $\mathcal{K}$ of size $K$ and computation cost of each key $c(k)$, where keys are in a non-increasing order of their computation costs, i.e., $c(k_1) \geq c(k_2) \geq \dots \geq c(K)$, if the perfect assignment exists and $c(k_1) < \bar{L}$ (the computation cost of any individual key is smaller than the average workload of task instances), we have $K \geq 2N_D$.
\end {lemma}
\begin{proof}
	This is straight forward given (a) the perfect assignment exists and (b) the computation cost of any individual key is smaller than $\bar{L}$, because for each instance, there must be at least two keys assigned to it. 
\end{proof}

\begin{lemma}\label{lemma3}
	Given the instance set $\mathcal{D}$ of size $N_D$, key set $\mathcal{K}$ of size $K$ and computation cost of each key $c(k)$, where keys are in a non-increasing order of their computation costs, i.e., $c(k_1) \geq c(k_2) \geq \dots \geq c(K)$, if the perfect assignment exists and $c(k_1) < \bar{L}$, we have: 
	\begin{equation}
	\theta_{max} \leq \dfrac{1}{3} \cdot (1-\dfrac{1}{N_D}),
	\end{equation}
	where $\theta_{max} = \max_{d\in\mathcal{D}}(\frac{L(d) - \bar{L}}{\bar{L}})$.
\end {lemma}
\begin{proof}
We prove by considering the worst case (in terms of load balance) where (a) the $(2N_D+1)$-th key has the largest possible computation cost $c(k_{2N_D+1}) = \bar{L}/3$, according to Lemma~\ref{lemma1} and Lemma~\ref{lemma2}; (b) Keys after the $(2N_D+1)$-th have equal amount of computation costs, denoted by $\varepsilon$, which are very close to zero; (c) The remaining workloads, i.e., $\sum_{k\in\mathcal{K}} c(k) - \frac{1}{3}\bar{L} - \varepsilon(K-2N_D-1)$, all concentrate on the first $2N_D$ keys and are evenly distributed, summarized as follows:
\begin{eqnarray}
c(k_i) =
	\left\{
	\begin{array}{ccl}
		\frac{N_D\bar{L} - \frac{1}{3}\bar{L} - \varepsilon(K-2N_D-1)}{N_D}
		& \mbox{for} & i = 1, 2, \dots, 2N_D; \nonumber\\
		\frac{1}{3}\bar{L} & \mbox{for} & i = 2N_D + 1;\nonumber\\
		\varepsilon & \mbox{for} & i > 2N_D + 1.
		\end{array}
		\right.
\end{eqnarray}
When $\varepsilon\rightarrow 0$, we have:
\begin{equation}
L_{\max} = \max_{d\in\mathcal{D}} L(d) = c(k_i) + c(k_{2N_D}) \leq \frac{4}{3}\bar{L} - \frac{\bar{L}}{3N_D},\nonumber
\end{equation}
where $i = 1, 2, \dots, 2N_D$. Note $L_{\max} = c(k_i) + c(k_{2N_D})$ is because according to the Simple algorithm, keys $k_i, i > 2N_D + 1$ will never be assigned to the instance with $L_{\max}$. This completes the proof according to our definition of $\theta_{\max}$.
\end{proof}
\begin{restate}\label{theo:theorem3}
Given the instance set $\mathcal{D}$ of size $N_D$, key set $\mathcal{K}$ of size $K$ and computation cost of each key $c(k)$, where keys are in a non-increasing order of their computation costs, i.e., $c(k_1) \geq c(k_2) \geq \dots \geq c(K)$, if the perfect assignment exists and $c(k_1) < \bar{L}$, LLFD always finds a solution resulting with balancing indicator $\theta(d,F)$ no worse than $\frac{1}{3}(1-\frac{1}{N_D})$ for any task instance $d$.
\end{restate}
\begin{proof}
According to Algorithm~\ref{alg:LLDF}, it has a larger search space than that of the Simple Algorithm, and is devoted to finding the assignment with more balanced workloads among instances, i.e., $\theta(d,F) \leq \theta_{max} \leq \frac{1}{3} \cdot (1-\dfrac{1}{N_D})$, which is proved in Lemma~\ref{lemma3}.
%
%For the LLFD algorithm, Theorem ~\ref{theo:theorem3} means that the load balance status produced by the LLFD algorithm denoted by $\theta_{LLFD}$ is not worse than that produced by the Simple algorithm, denoted by $\theta_{Simple}$.
%
%and now we take $\theta_{Sim}$ as the upper bound of imbalance tolerance $\theta_{max}$, then the \emph{Mixed} algorithm produces overload instance(s), in other words, $\exists c(k_p)$ in \emph{Mixed}'s migration process incurs that the overload is bigger than $\theta_{Sim}$. Because the \emph{Mixed} tries all instances to put $c(k_{p})$ to eliminate overload as shown in Algorithm~\ref{alg:Mixed}, then $\theta_{Mix} > \theta_{Sim}$ means that $\forall d$, $d \in \mathcal{D}$, $c(k_{p})+ L(d) - \sum_{k' \in \lbrace k'' | c(k'') < c(cop) \rbrace}{c(k')}$ is larger than upbound. In this case, \emph{Simple} algorithm also can not assign $c(k_{p})$ to $\mathcal{D}$ without exceeds the upbound. Then the balance degree produced by \emph{Mixed} is not worse than \emph{Simple}'s.
\end{proof}

%In addition, $\beta$ also affects the size of the result routing table, i.e., the larger $\beta$, the smaller size of routing table, which will be shown in the experiment results in Appendix~\ref{}.
\begin{theorem}\label{theo:theorem3}
	Balance status generated by the Mixed represented by $\theta_{Mix}$ is not worse than the
	balance status $\theta_{Sim}$ produced by the Simple algorithm algorithm.
	\end {theorem}
	\begin{proof}
		Supposing $\theta_{Mix} > \theta_{Sim}$, now we take $\theta_{Sim}$ as $\theta_{max}$, then the Mixed algorithm overload instance and $\exists cop$ in Mixed's migration process.  Because the Mixed has tried all instances to put $c(k_{cop})$ for without overload as shown in Algorithm~\ref{alg:Mixed}, then $\forall d, d \in \mathcal{D}$, $c(k_{cop})+ L(d) - \sum_{k' \in \lbrace k'' | c(k'') < c(cop) \rbrace}{c(k')}$ is larger than upbound. This is contradicts to begin supposing: $\theta_{Mix} > \theta_{Sim}$. 
	\end{proof}

%% file: appendix-routingtablechange.tex
%%overview
\section{Routing Table Changing}\label{sec:RoutingTable}

\begin{figure}
	\centering
		\includegraphics[height = 4.6cm, width = 7.9cm]{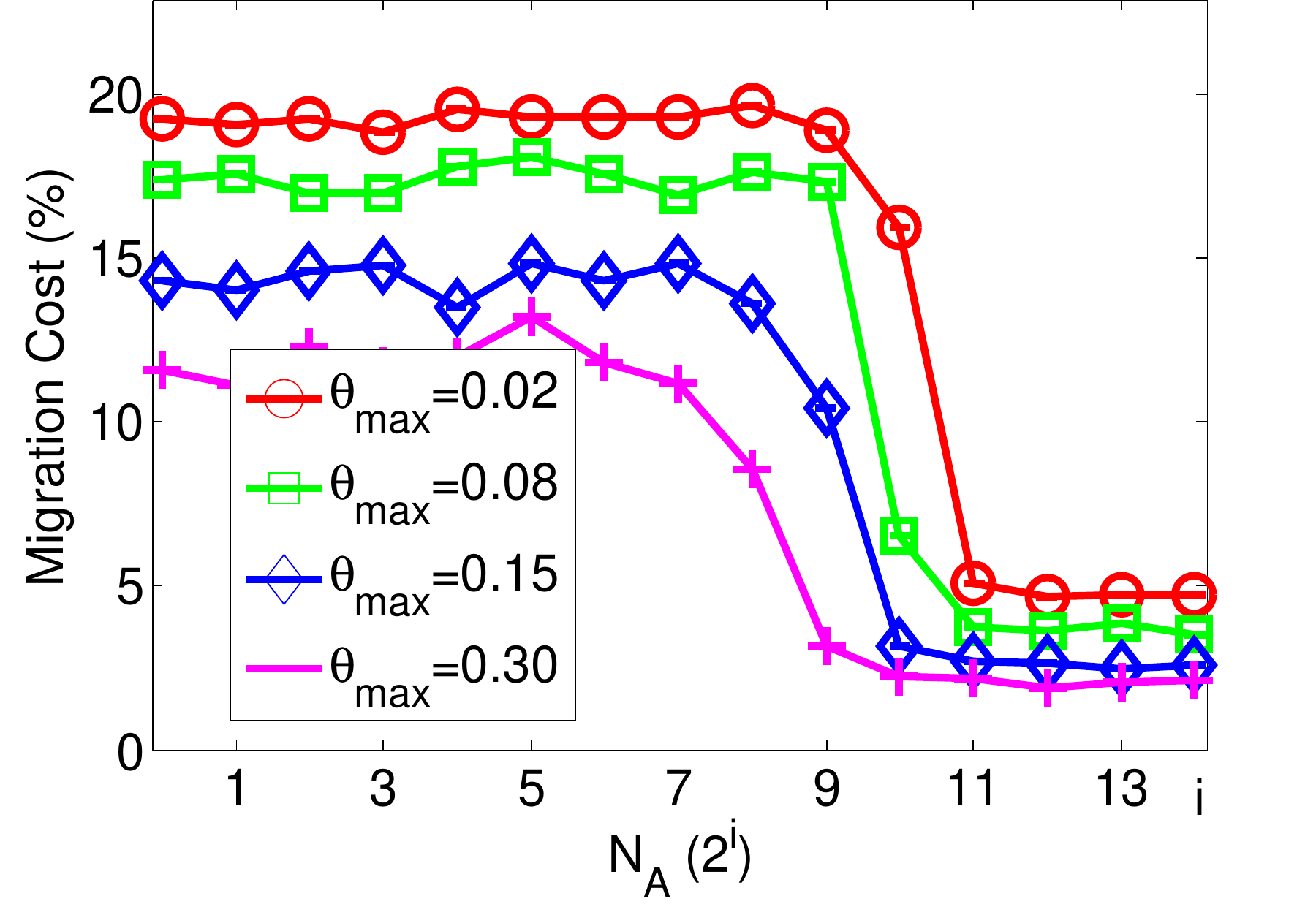}
		
	\centering \caption{Migration cost with different $N_{A}$ by \emph{Mixed}.}
	\label{fig:f2tmc_mutil_theta}
	% \vspace{-10pt}
	%\label{fig:Parameters}
\end{figure}

With the defaulted parameter settings, our \emph{Mixed} algorithm with different numbers of $N_{A}$ has different migration cost  as shown in Fig.~\ref{fig:f2tmc_mutil_theta}. 
When we set $N_{A}$ to the values less than $1000$ (calculated by $2^i$ with $i\leq 10$) for $\theta_{max}=0.08$, our algorithm generates a large migration cost for it acts as algorithm (Alg.~\ref{alg:MinTable}) and leaves alone $N_{A}$ requirement. 
%So In Figure~\ref{fig:exp:TMC}, $MC$ values are stable when $T< 900$,which is almost the same as $MC$ caused by $MRT$ method.
But when $N_{A}$ is relaxed to $2000$ ($i>11$), migration cost decreases greatly, since our algorithm acts according to \emph{Mixed} algorithm which can reduce the routing back actions greatly.
Furthermore, the different degrees of balance status ($\theta_{max}$) require different minimum $N_{A}$, and it will drastically reduce the migration cost shown in Fig.~\ref{fig:f2tmc_mutil_theta}.

Fig.~\ref{fig:T_change} shows the change of routing table size for different numbers of adjustment.
We use algorithm(Alg.~\ref{alg:MinMig}) to do balance and set $K=10^{4}$ to verify the results quickly.
Obviously, the smaller $\theta_{max}$ accelerates the growth of routing table. We also observe that routing table sizes with different $\theta_{max}$s  converge to the same size (around 9350 entries) in Fig.~\ref{fig:T_change}.
This is because \emph{Min Mig} does balancing without considering the constraint for routing table size.
In other words,  a key is paired to the task randomly when the basic assignment function has caused imbalance. Therefore, the probability of a key appears in routing table is $\frac{N_{D}-1}{N_{D}}$, and then, after a long period of load balancing, the routing table size should be $\frac{N_{D}-1}{N_{D}} \cdot K$ with K as the key size.

\begin{figure}
	\centering
	\includegraphics[height = 4.6cm, width = 8.2cm]{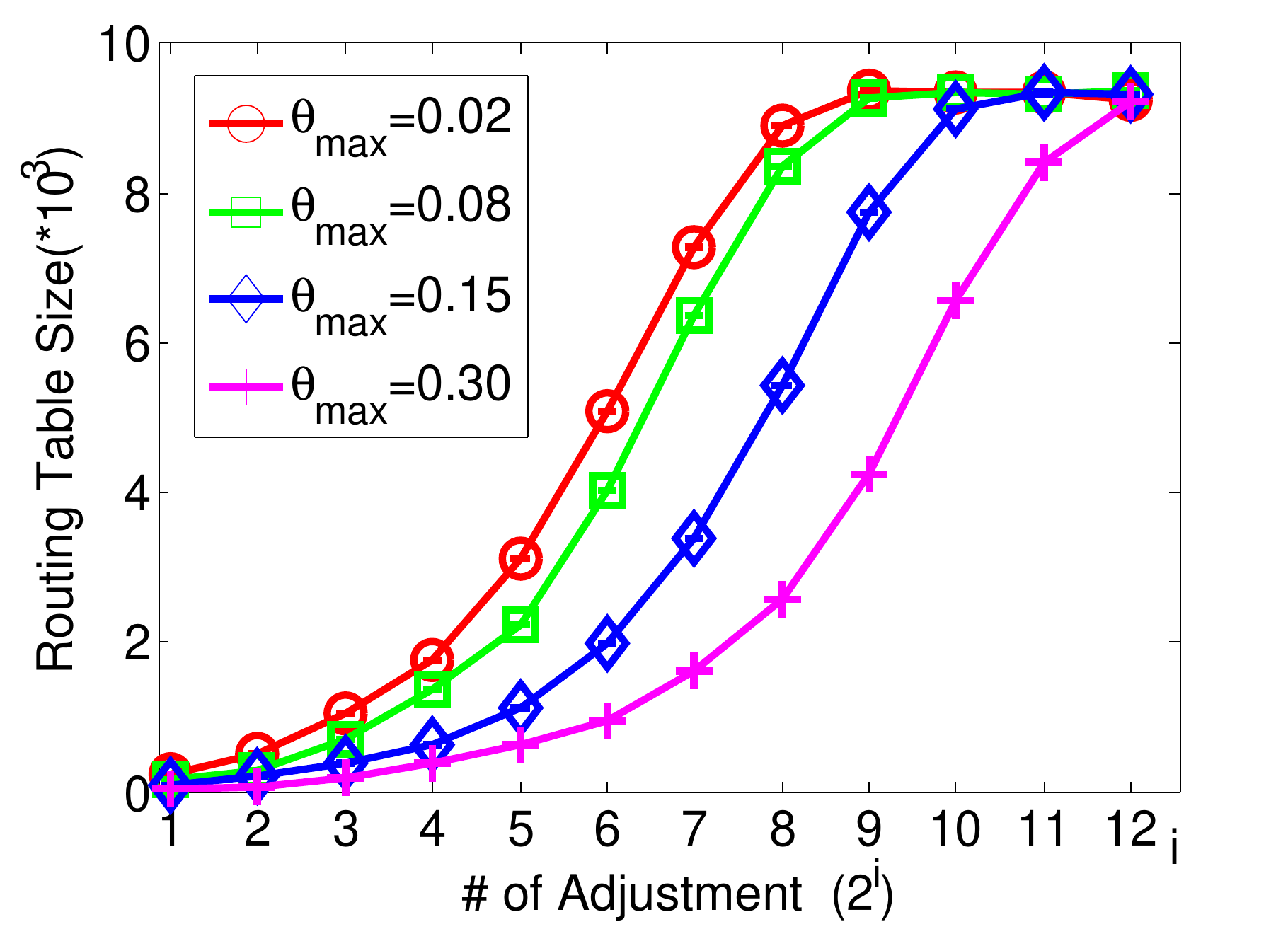}
	 \centering \caption{Routing table changing along with $\#$ of adjustments.}	\label{fig:T_change}
	% \vspace{-10pt}
	%\label{fig:Parameters}
\end{figure}

%% file: appendix-addparametertest.tex
%%overview
\section{Additional Parameter Test}\label{appsec:additionalparameter}
\begin{figure}
	\centering
	\includegraphics[height = 4.6cm, width = 8.2cm]{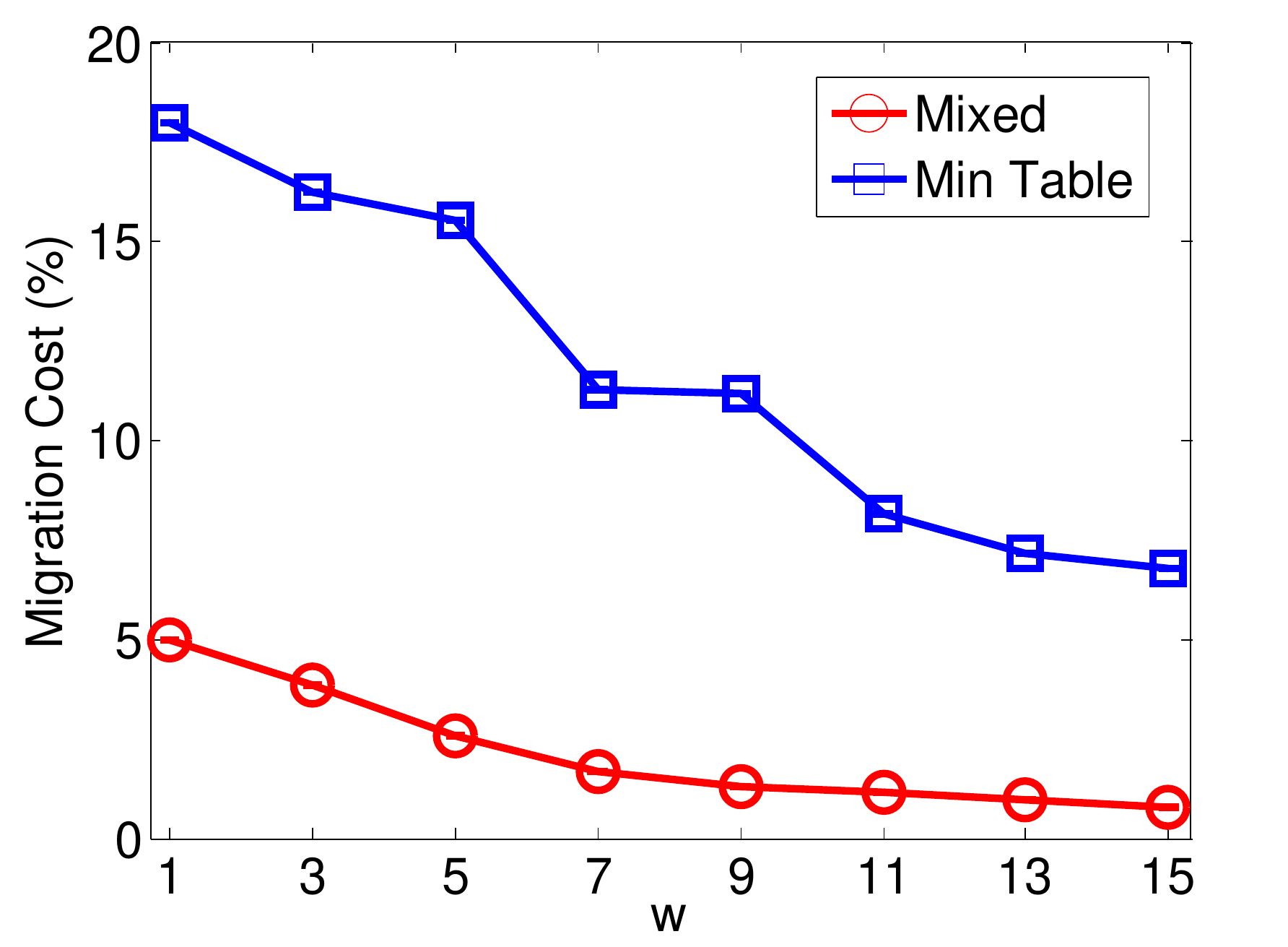}
	\centering \caption{Migration cost \emph{vs} window size.}	\label{fig:exp:WMC}
	% \vspace{-10pt}
	%\label{fig:Parameters}
\end{figure}
We show the migration cost with different window size in Fig.~\ref{fig:exp:WMC}. Since the larger window size provides more chance for finding the appropriate migration keys ($\gamma_i(k, w)$), the migration cost of \emph{Mixed} is smaller than \emph{MinTable}.
\begin{figure}
	\centering
	\includegraphics[height = 4.6cm, width = 8.2cm]{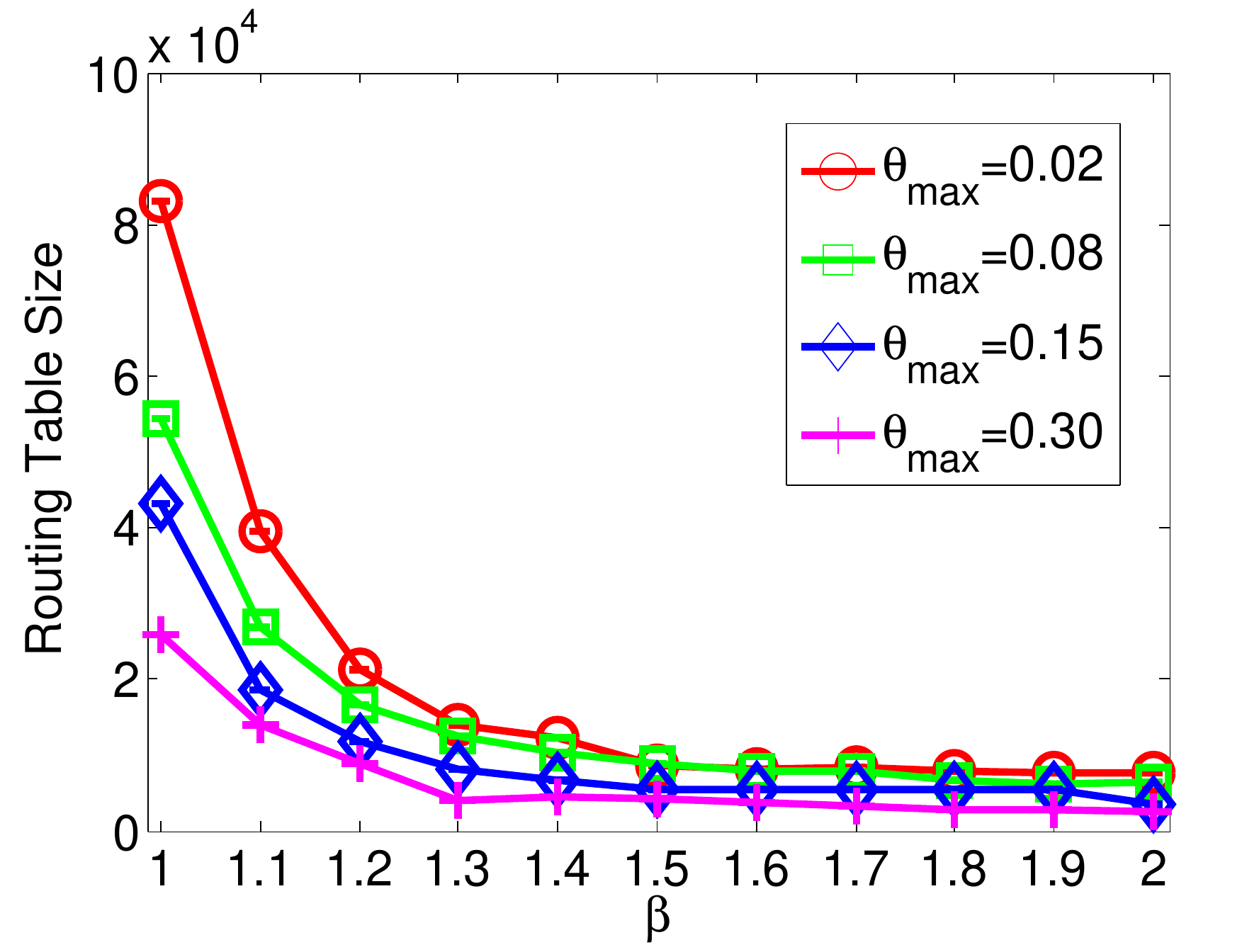}
	\centering \caption{Routing table size in different $\beta$.}	\label{fig:beta_routingtable}
	% \vspace{-10pt}
	%\label{fig:Parameters}
\end{figure}

\begin{figure}
	\centering
	\includegraphics[height = 4.6cm, width = 8.2cm]{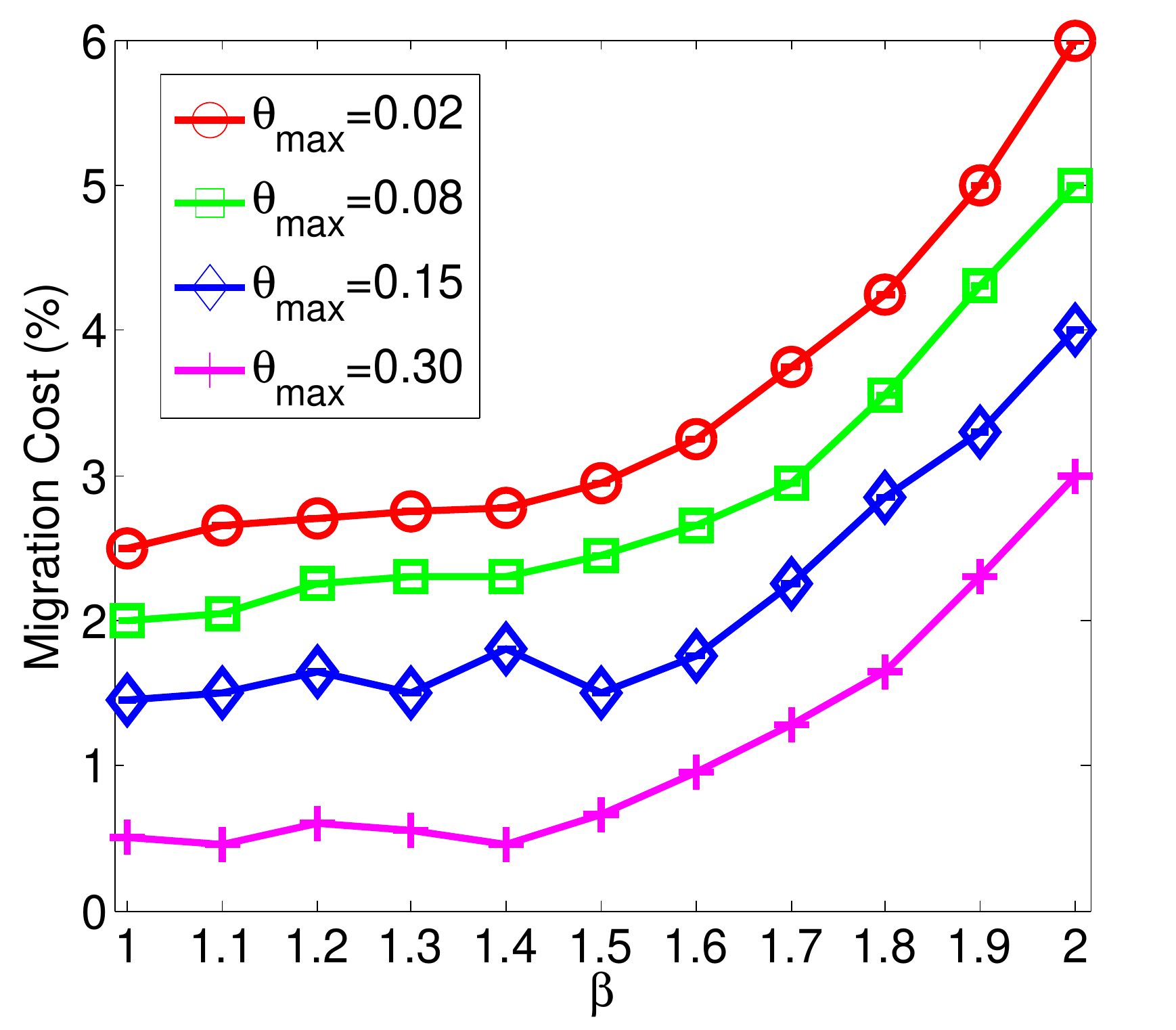}
	\centering \caption{Migration cost in different $\beta$.}	\label{fig:beta_migrationcost}
	% \vspace{-10pt}
	%\label{fig:Parameters}
\end{figure}
To characterize both computation and migration cost, we propose the \textit{migration priority index} for each key, defined as $\gamma_i(k, w) = c_i(k)^\beta S_i(k, w)^{-1}$ as in Sec.~\ref{subsec:MinTableandMinMigHeuristics}. $\beta$ is used to measure the importance of computation cost and the memory consumption which decides the \textit{migration priority} of keys.
In Alg.~\ref{alg:MinMig}, a key with the larger $\gamma_i(k, w)$ has the higher priority to be migrated.
Larger $\beta$ represents that the migration method  concerns more on faster computation rather than on less migration cost (memory consumption).
Furthermore, larger $\beta$ will produce smaller routing table since the migration method preferentially migrates keys with large load.
Fig.~\ref{fig:beta_routingtable} and Fig.~\ref{fig:beta_migrationcost} show the change of routing table size and migration cost with different values of $\beta$.
Those results are produced by the \emph{MinMig} algorithm, and each result is an average value generated by running 10 times of balance adjustments. In Fig.~\ref{fig:beta_routingtable}, 
$\beta=1$ means the migration candidates are evaluated according to the load per unit memory consumption. In this case, keys with smaller load would be selected and a larger routing table  will be generated. As $\beta$ is set larger, migration method gradually tends to move the bigger load keys, then routing table size becomes smaller. Furthermore, when $\beta \in [1.5 , 2]$, routing table size is stable because the migration candidates are almost selected only by the load of keys.
The lines in Fig.~\ref{fig:beta_migrationcost} show results by the influence of parameter $\beta$.
Based on these sets of parameter tests, we select $\beta=1.5$ as the default value in our experiments.